\documentclass[11pt]{article}

\usepackage[english]{babel}
\usepackage[T1]{fontenc}

\usepackage[letterpaper,top=3cm,bottom=3cm,left=2cm,right=2cm,marginparwidth=2cm]{geometry}

\usepackage{amsmath}
\usepackage{graphicx}
\usepackage{amsfonts}
\usepackage{amsthm}
\usepackage{calc}
\usepackage{accents}
\usepackage{comment}
\usepackage{slashed}
\usepackage[colorinlistoftodos]{todonotes}
\usepackage[colorlinks=true, allcolors=black]{hyperref}
\usepackage{thmtools}
\usepackage{thm-restate}
\usepackage[nameinlink]{cleveref}
\usepackage{amssymb}
\usepackage{float}
\usepackage{tikzfill}
\usepackage{lipsum}
\usepackage[justification=centering]{caption}
\usepackage{subcaption}

\allowdisplaybreaks
\newtheorem{lemma}{Lemma}[section]
\newtheorem{theorem}{Theorem}[section]
\newtheorem*{theorem*}{Theorem}
\newtheorem{definition}{Definition}[section]
\newtheorem{remark}{Remark}[section]
\newtheorem{corollary}{Corollary}[section]
\newtheorem{proposition}{Proposition}[section]
\newtheorem*{proposition*}{Proposition}
\newtheorem{assumption}{Assumption}[section]

\numberwithin{equation}{section}
\title{Weak cosmic censorship for the circularly symmetric \\ Einstein-scalar field system in $2+1$ dimensions}
\author{Serban Cicortas\footnote{Princeton University, Department of Mathematics, Fine Hall, Washington Road, Princeton, NJ 08544, USA}}

\begin{document}

\maketitle

\begin{abstract}
    We prove the weak cosmic censorship conjecture in $2+1$ spacetime dimensions for the circularly symmetric Einstein-scalar field system in the presence of a negative cosmological constant $\Lambda<0$. More precisely, we show that for any integer $k\geq2$, the maximal development of generic $C^k$ initial data does not contain naked singularities. An essential step of the proof is establishing the presence of a mass gap. In particular, this implies that all naked singularities have infinite blueshift, which represents the main instability mechanism.
\end{abstract}

\tableofcontents

\section{Introduction}

One of the main open problems in general relativity is the \textit{weak cosmic censorship conjecture}, proposed by Penrose in \cite{P69} for general Einstein-matter systems. The conjecture states that, for generic initial data, so-called "naked singularities" do not form. That is to say, singularities are generically hidden inside black holes and invisible to far away observers who exist for all time (more formally, null infinity $\mathcal{I}$ is future complete \cite{GH78, C99-GL}). The conjecture is motivated by the fact that the future incompleteness of $\mathcal{I}$ would threaten the predictability of general relativity as a physical~theory. 

While the weak cosmic censorship conjecture is completely open in its full generality, in the remarkable work \cite{ChrWCC}, Christodoulou proved a suitable formulation of the conjecture for the model problem of the spherically symmetric Einstein-scalar field system in $3+1$ dimensions. The argument of \cite{ChrWCC} relies on a novel \textit{blueshift} effect to show that, while explicit naked singularity spacetimes can be constructed \cite{C94, Cho93}, all naked singularities are in fact \textit{unstable} to black hole formation and the set of naked singularity spacetimes is indeed \textit{non-generic} in a suitable sense. However, the instability mechanism in \cite{ChrWCC} required working in a topology defined by a low-regularity class of solutions. Thus, the question of weak cosmic censorship in a smoother topology remains open even in this setting. We discuss the work \cite{ChrWCC} in detail in \Cref{Christodoulou section}.

In the present work, we initiate the mathematical study of gravitational collapse in $2+1$ dimensions for the circularly symmetric Einstein-scalar field system in the presence of a negative cosmological constant $\Lambda<0$, and we prove a version of the weak cosmic censorship conjecture for this system. We note that the natural problem in this setting is solving an initial-boundary value problem with reflective boundary conditions at null infinity $\mathcal{I}$ (which is now in fact timelike). This model was first studied in the numerical relativity literature by Pretorius and Choptuik \cite{Pretorius}, and we will discuss its motivation in detail in \Cref{self intro sectioom}. We also discuss the existence of examples of naked singularities for this model in \Cref{examples section}. To prove weak cosmic censorship, we use the blueshift effect as the main instability mechanism, similarly to \cite{ChrWCC}. A key new feature in our setting is the presence of a \textit{mass gap}, which states that no singularities form in the domain of dependence of a disk bounded by a circle with Hawking mass strictly less than~1. (In particular, this allows us to confirm a conjecture of \cite{ADS3}). Moreover, the different scaling in the present model allows us to activate the blueshift instability in the class of $C^k$ solutions, for any $k\geq2$. Thus, unlike \cite{ChrWCC}, our formulation of weak cosmic censorship in the present setting does not rely on low-regularity perturbations and holds in a topology of arbitrary smoothness.

\subsection{The Einstein-scalar field system in $2+1$ dimensions with $\Lambda<0$}\label{self intro sectioom}

We consider the Einstein-scalar field system of equations with a negative cosmological constant $\Lambda<0$:
\begin{equation}\label{Einstein equations}
    \begin{cases}
        Ric_{\mu\nu}-\frac{1}{2}Rg_{\mu\nu}+\Lambda g_{\mu\nu}=\nabla_{\mu}\phi\nabla_{\nu}\phi-\frac{1}{2}g_{\mu\nu}|\nabla\phi|^2,\\
        \square_{g}\phi=0.
    \end{cases}
\end{equation}
In $2+1$ dimensions, the system \eqref{Einstein equations} has been studied extensively in the physics and numerical relativity literature \cite{BTZ, BTZsurvey, Pretorius, Garfinkle, GarfGund, jalmuzna}. This represents
perhaps the most elementary model for gravitational collapse in $2+1$ dimensions which allows for the dynamical formation of singularities \cite{Pretorius, jalmuzna}.

The simplest solution of \eqref{Einstein equations} is given by Anti-de Sitter (AdS) spacetime $\big(\mathbb{R}^3,g_{AdS},\phi_{AdS}\big):$
\[g_{AdS}=-\big(1+|\Lambda|r^2\big)dt^2+\big(1+|\Lambda|r^2\big)^{-1}dr^2+r^2d\sigma_1^2,\ \phi_{AdS}=0,\]
where $d\sigma_1^2$ is the standard metric on $S^1.$ We note that AdS spacetime is circularly symmetric, does not contain any singularities, and is geodesically complete. Moreover, we can attach an asymptotic boundary at infinity, called \textit{null infinity} $\mathcal{I}$, which is timelike and complete.

In the present work, we initiate the rigorous mathematical study of the system \eqref{Einstein equations}. We consider the class of circularly symmetric solutions, which take the following form with respect to global double null coordinates $(u,v):$
\begin{equation}\label{def of axisym intro}
    g(u,v)=-\Omega^2(u,v)dudv+r^2(u,v)d\sigma_1^2,\ \phi=\phi(u,v).
\end{equation}

The solutions of \eqref{Einstein equations} studied in this paper will be the unique maximal developments of $C^k$ asymptotically AdS characteristic data sets, obtained by solving \eqref{Einstein equations} as an initial-boundary value problem with reflective boundary conditions at null infinity $\mathcal{I}$. Moreover, we only consider solutions arising in \textit{gravitational collapse}, i.e.~solutions with a regular center $\Gamma=\{r=0\}.$ In particular, AdS spacetime itself is contained in this class of solutions. We refer the reader to \Cref{set up section} for a detailed introduction to the circularly symmetric Einstein-scalar field system in $2+1$ dimensions.

An essential role in the study of \eqref{Einstein equations} is played by the Bañados--Teitelboim--Zanelli (BTZ) spacetimes \cite{BTZ}. These are a $1$-parameter family of circularly symmetric static solutions\footnote{We change the notation convention from the physics literature and translate the mass parameter by 1. For us, AdS spacetime is given by $M=0,$ whereas in \cite{BTZ} AdS spacetime is given by $M=-1$.}\textsuperscript{,}\footnote{More generally, \cite{BTZ} constructs a $2$-parameter family of axially symmetric stationary solutions $g_{M,J}$. The rotating BTZ solutions with $J\neq0$ are not circularly symmetric and they do not play a role for the rest of the paper.} of \eqref{Einstein equations} defined for $M\geq0$ with respect to static coordinates $(t,r):$
\begin{equation}\label{BTZ definition}
    g_M=-\Big(1-M+|\Lambda|r^2\Big)dt^2+\Big(1-M+|\Lambda|r^2\Big)^{-1}dr^2+r^2d\sigma_1^2,\ \phi_{M}=0.
\end{equation}
We briefly explain some properties of the BTZ solutions and we refer the reader to the survey \cite{BTZsurvey}. For $M=0,$ the solution coincides with AdS spacetime. In the case $M\in(0,1)$ the metric \eqref{BTZ definition} has a conical singularity at $r=0,$ while for $M=1$ the metric \eqref{BTZ definition} has a cusp singularity at $r=0$. Finally, for $M>1,$ the spacetime has no center but it can be extended to a two-ended black hole spacetime. In view of our restriction in the previous paragraph to the class of spacetimes arising in gravitational collapse, the only member of the BTZ family that lies in our class of solutions is AdS spacetime. Nonetheless, the entire BTZ family will play an important role in the asymptotic behavior of the solutions in our class at $\mathcal{I}.$

\subsection{The main result}

In order to state the main result in detail, we consider the moduli space $\mathfrak{M}$ of $C^k$ asymptotically AdS characteristic data posed on an outgoing null cone $C_0^+$ which extends from the center $\Gamma$ to null infinity $\mathcal{I}$. The moduli space $\mathfrak{M}$ is endowed with the structure of a metric space (see already \Cref{moduli space definition}). To each element $\mathfrak{D}\in\mathfrak{M}$, we associate the unique (up to isometry) maximal development of the initial data $\mathfrak{D}.$

We briefly define certain subsets of the moduli space $\mathfrak{M}$, and we refer the reader to \Cref{global structure section} for the precise definitions. These subsets are defined based on the global properties of the maximal development of the initial data. First, we denote by $\mathfrak{M}_{\mathrm{black}}\subset\mathfrak{M}$ the set of initial data that evolve to black hole spacetimes, i.e. spacetimes such that $J^-(\mathcal{I})$ has a nontrivial complement. We classify black hole spacetimes depending on the nature of the singularity inside the black hole region. We denote by $\mathfrak{M}_{\mathrm{black}}^{\mathrm{spacelike}}\subset\mathfrak{M}_{\mathrm{black}}$ the set of initial data that evolve to black hole spacetimes with a spacelike-only singularity, as depicted in \Cref{fig:bh2}. Then, we denote by $\mathfrak{M}_{\mathrm{black}}^{\mathrm{loc.naked}}\subset\mathfrak{M}_{\mathrm{black}}$ the set of initial data that evolve to black hole spacetimes with the property that their Penrose diagram is that in \Cref{fig:bh1} (the superscript stands for locally naked, which we define below). We prove in \Cref{set up section} that we have the disjoint union: 
\begin{equation}
    \mathfrak{M}_{\mathrm{black}}=\mathfrak{M}_{\mathrm{black}}^{\mathrm{spacelike}}\sqcup\mathfrak{M}_{\mathrm{black}}^{\mathrm{loc.naked}}.
\end{equation}
Next, we denote by $\mathfrak{M}_{\mathrm{non}}\subset\mathfrak{M}$ the set of initial data that evolve to non-collapsing spacetimes, defined by the property that their Penrose diagram is that in \Cref{fig:non} (for example AdS spacetime). Furthermore, we denote by $\mathfrak{M}_{\mathrm{naked}}\subset\mathfrak{M}$ the set of  initial data that evolve to naked singularity spacetimes, defined\footnote{One can show that any solution with future incomplete $\mathcal{I}$ is the development of initial data in $\mathfrak{M}_{\mathrm{naked}}.$ We note that in our definition, $\mathfrak{M}_{\mathrm{naked}}$ will also include initial data that evolve to spacetimes with complete $\mathcal{I}$, see \Cref{remark about I incomplete}.} by the property that their Penrose diagram is that in \Cref{fig:nakedsing}. Finally, we denote by $\mathfrak{M}_{\mathrm{loc.naked}}\subset\mathfrak{M}$ the set of  initial data that evolve to locally naked singularity spacetimes, defined by the property that their causal boundary contains a null component $\mathcal{B}_0$ arising out of the first singularity at the center $b_{\Gamma}$. Equivalently, locally naked singularities do not contain trapped surfaces in a neighborhood of $b_{\Gamma}$. We prove in \Cref{set up section} that we have the disjoint union: 
\begin{equation}
    \mathfrak{M}_{\mathrm{loc.naked}}=\mathfrak{M}_{\mathrm{naked}}\sqcup\mathfrak{M}_{\mathrm{black}}^{\mathrm{loc.naked}}.
\end{equation}

The classification of the general structure of the causal boundary of spacetimes leads to a decomposition of the moduli space $\mathfrak{M}$. Adapting the work of \cite{kommemi, localwellposed} to the $(2+1)$-dimensional setting, we prove in \Cref{set up section} that the moduli space $\mathfrak{M}$ can be written as the disjoint union:
\begin{equation}\label{moduli space decomposition intro}
    \mathfrak{M}=\mathfrak{M}_{\mathrm{black}}^{\mathrm{spacelike}}\sqcup\mathfrak{M}_{\mathrm{black}}^{\mathrm{loc.naked}}\sqcup\mathfrak{M}_{\mathrm{non}}\sqcup\mathfrak{M}_{\mathrm{naked}}.
\end{equation}
For a discussion of the non-emptiness of the various subsets of $\mathfrak{M},$ see \Cref{global structure section}.

Using these definitions, we state the main result of our paper:
\begin{theorem}[Weak cosmic censorship for circularly symmetric solutions of \eqref{Einstein equations}]\label{main theorem}
    For any $k\geq 2$, the maximal development of generic $C^k$ asymptotically AdS characteristic data does not contain naked singularities. More precisely, we have that:
    \begin{equation}\label{M generic open and dense}
        \mathfrak{M}_{\mathrm{generic}}:=\mathfrak{M}_{\mathrm{black}}\sqcup\mathrm{int}\big(\mathfrak{M}_{\mathrm{non}}\big)\text{ is open and dense in }\mathfrak{M},
    \end{equation}
    \begin{equation}\label{inclusions for main theorem}
        \mathfrak{M}_{\mathrm{naked}}\subset\mathrm{cl}\big(\mathfrak{M}_{\mathrm{black}}\big)\backslash\mathfrak{M}_{\mathrm{black}}\subset\mathfrak{M}\backslash\mathfrak{M}_{\mathrm{generic}},
    \end{equation}
    \begin{equation}\label{loc naked for main theorem}
        \mathfrak{M}_{\mathrm{loc.naked}}\subset\mathrm{cl}\big(\mathfrak{M}_{\mathrm{black}}^{\mathrm{spacelike}}\big)\backslash\mathfrak{M}_{\mathrm{black}}^{\mathrm{spacelike}}.
    \end{equation}
    In particular, naked singularity spacetimes are non-generic and are unstable to black hole formation. Additionally, locally naked singularities are unstable to the formation of spacelike-only singularities.
\end{theorem}

\begin{remark}
    For reasons to be explained later, see already \Cref{Christodoulou section}, we choose a more elementary notion of genericity than \cite{ChrWCC}. Thus, we prove that $\mathfrak{M}_{\mathrm{naked}}$ is contained in the complement of an open and dense subset of $\mathfrak{M}$, rather than proving that $\mathfrak{M}_{\mathrm{naked}}$ has positive codimension in $\mathfrak{M}$.
\end{remark}

We outline the structure of the remainder of the introduction. In \Cref{proof idea section} we discuss the ideas of the proof in some detail. In \Cref{previous results section} we discuss some relevant previous results. Finally, in \Cref{outline section} we outline the structure of the rest of the paper.

\begin{figure}[H]
\centering

% -------- top row --------
\begin{subfigure}[t]{0.48\textwidth}
\centering
\begin{tikzpicture}[scale=0.8]
    \path[use as bounding box] (-0.5,-2.4) rectangle (3.0,5.4);

    \draw (0,-2) .. controls (0.2,-0.5) and (0.2,1.5) .. (0,3);
    \draw[thick, dashdotted] (0,3) .. controls (0.3,2.8) and (1.2,3) .. (2.1,3.5);
    \draw[dashed] (2.1,0.1) .. controls (2.3,1.2) and (2.3,3.3) .. (2.1,3.5); 
    \draw (0,-2) -- (2.1,0.1);
    \draw (0.125,1.525) -- (2.1,3.5);

    \filldraw[color=black, fill=white] (0,3) circle (2pt);
    \filldraw[color=black, fill=white] (2.1,3.5) circle (2pt);

    \draw (2.2,3.6) node[anchor=west] {$i^+$};
    \draw (2.3,2)   node[anchor=west] {$\mathcal{I}$};
    \draw (0.1,0.5) node[anchor=east] {$\Gamma$};
    \draw (-0.1,3)  node[anchor=east] {$b_{\Gamma}$};
    \draw (1.2,-1)  node[anchor=north] {$C_0^+$};
    \draw (1.1,3.2) node[anchor=south] {$\mathcal{B}$};
    \draw (1.3,2)   node[anchor=south] {$\mathcal{H}$};
\end{tikzpicture}
\caption{Black hole spacetime with a \\ spacelike-only singularity ($\mathfrak{M}_{\mathrm{black}}^{\mathrm{spacelike}}$)}
\label{fig:bh2}
\end{subfigure}
\hfill
\begin{subfigure}[t]{0.48\textwidth}
\centering
\begin{tikzpicture}[scale=0.8]
    \path[use as bounding box] (-0.5,-2.4) rectangle (3.0,5.4);

    \draw (0,-2) .. controls (0.2,-0.5) and (0.2,1.5) .. (0,3);
    \draw[dashed] (0,3) -- (0.7,3.7);
    \draw[thick, dashdotted] (0.7,3.7) .. controls (1.3,3.4) and (1.7,3.3) .. (2.1,3.5);
    \draw[dashed] (2.1,0.1) .. controls (2.3,1.2) and (2.3,3.3) .. (2.1,3.5); 
    \draw (0,-2) -- (2.1,0.1);
    \draw (0.125,1.525) -- (2.1,3.5);

    \filldraw[color=black, fill=white] (0,3) circle (2pt);
    \filldraw[color=black, fill=white] (2.1,3.5) circle (2pt);

    \draw (2.2,3.6) node[anchor=west] {$i^+$};
    \draw (2.3,2)   node[anchor=west] {$\mathcal{I}$};
    \draw (0.1,0.5) node[anchor=east] {$\Gamma$};
    \draw (-0.1,3)  node[anchor=east] {$b_{\Gamma}$};
    \draw (1.2,-1)  node[anchor=north] {$C_0^+$};
    \draw (0.2,3.4) node[anchor=south] {$\mathcal{B}_0$};
    \draw (1.5,3.5) node[anchor=south] {$\mathcal{B}$};
    \draw (1.3,2)   node[anchor=south] {$\mathcal{H}$};
\end{tikzpicture}
\caption{Black hole spacetime with a \\ locally naked singularity ($\mathfrak{M}_{\mathrm{black}}^{\mathrm{loc.naked}}$)}
\label{fig:bh1}
\end{subfigure}

% -------- bottom row --------
\begin{subfigure}[t]{0.48\textwidth}
\centering
\begin{tikzpicture}[scale=0.72]
    \path[use as bounding box] (-0.5,-2.4) rectangle (3.0,5.4);

    \draw (0,-2) .. controls (0.2,-0.5) and (0.2,1.5) .. (0.8,4);
    \draw[dashed] (2,0) .. controls (1.8,1.5) and (1.4,3.3) .. (0.8,4); 
    \draw (0,-2) -- (2,0);

    \filldraw[color=black, fill=white] (0.8,4) circle (2pt);

    \draw (1.6,2.5) node[anchor=west] {$\mathcal{I}$};
    \draw (0.1,1)   node[anchor=east] {$\Gamma$};
    \draw (1,4.1)   node[anchor=west] {$i^+$};
    \draw (1.2,-1)  node[anchor=north] {$C_0^+$};
\end{tikzpicture}
\caption{Non-collapsing spacetime ($\mathfrak{M}_{\mathrm{non}}$)}
\label{fig:non}
\end{subfigure}
\hfill
\begin{subfigure}[t]{0.48\textwidth}
\centering
\begin{tikzpicture}[scale=0.68]
    \path[use as bounding box] (-0.5,-2.4) rectangle (3.0,5.4);

    \draw (0,-2) .. controls (0.2,-0.5) and (0.2,1.5) .. (0,3);
    \draw[dashed] (0,3) -- (2,5);
    \draw[dashed] (2,0) .. controls (2.2,1.5) and (2.2,3.5) .. (2,5); 
    \draw (0,-2) -- (2,0);

    \filldraw[color=black, fill=white] (0,3) circle (2pt);

    \draw (2.3,2.5) node[anchor=west] {$\mathcal{I}$};
    \draw (0.1,0.5) node[anchor=east] {$\Gamma$};
    \draw (-0.1,3)  node[anchor=east] {$b_{\Gamma}$};
    \draw (1.2,-1)  node[anchor=north] {$C_0^+$};
    \draw (2.1,5)   node[anchor=west] {$i^+$};
    \draw (0.8,4)   node[anchor=south] {$\mathcal{B}_0$};
\end{tikzpicture}
\caption{Naked singularity spacetime ($\mathfrak{M}_{\mathrm{naked}}$)}
\label{fig:nakedsing}
\end{subfigure}

\caption{All the possible Penrose diagrams of %circularly symmetric solutions of \eqref{Einstein equations} arising in gravitational collapse. 
the maximal development of $\mathfrak{D}\in\mathfrak{M}.$ \\ \ref{fig:bh2} and \ref{fig:bh1} are black hole spacetimes, while \ref{fig:bh1} and \ref{fig:nakedsing} are locally naked singularity spacetimes.}
\label{fig:all_spacetimes}
\end{figure}
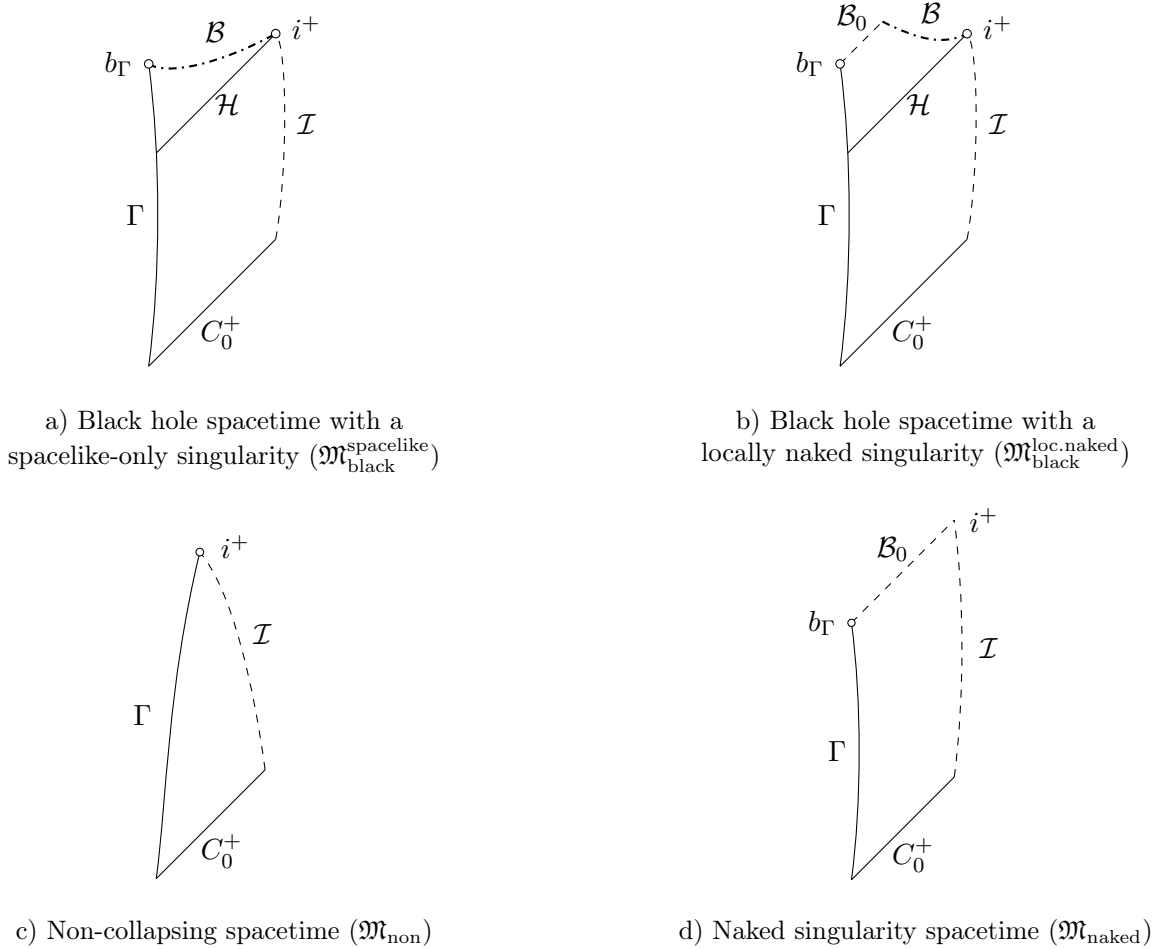

\subsection{Idea of the proof}\label{proof idea section} 

We briefly outline the main ideas of the proof, then we flesh out our discussion with more details in Sections~\ref{mass gap intro section}-\ref{proof of main thm section intro}. The overall strategy of the proof goes back to the work of Christodoulou \cite{ChrWCC}. In the present section we explain our argument, then we describe in detail in \Cref{Christodoulou section} the similarities and differences to \cite{ChrWCC}. We refer the reader to \Cref{set up section} for the precise notation and definitions used in this section.

The main step in the proof of \Cref{main theorem} is establishing \eqref{loc naked for main theorem}. To show that locally naked singularities are unstable to the formation of spacelike-only singularities, it suffices to prove that locally naked singularities can be perturbed to form trapped surfaces arbitrarily close to $b_{\Gamma}$.

The proof relies on two fundamental results about circularly symmetric solutions of \eqref{Einstein equations}. Firstly, in \Cref{mass gap intro section} we state \Cref{mass gap theorem}, which establishes the presence of a mass gap, showing that no singularities form in the domain of dependence of a disk bounded by a circle with Hawking mass strictly less than 1. We also state some notable consequences of \Cref{mass gap theorem}. Secondly, in \Cref{trapped surface section intro} we state \Cref{trapped surface formation theorem}, which provides a quantitative criterion for the formation of trapped surfaces. 

We explain in \Cref{locally naked section intro} how \Cref{mass gap theorem} and \Cref{trapped surface formation theorem} are used to prove a quantitative description of locally naked singularity spacetimes. In particular, we note that any first singularity at the center must have unbounded blueshift. In \Cref{blueshift instability section intro}, we explain why locally naked singularities can be perturbed to form trapped surfaces arbitrarily close to $b_{\Gamma}$. The idea is to consider perturbations that leave the causal past of $b_{\Gamma}$ unchanged and activate the blueshift instability. The blueshift effect will then amplify the initial data perturbation and allow us to apply \Cref{trapped surface formation theorem} to conclude the existence of trapped surfaces arbitrarily close to $b_{\Gamma}$. Finally, we explain in \Cref{proof of main thm section intro} how to complete the proof of \Cref{main theorem}.

\subsubsection{The mass gap}\label{mass gap intro section}

The first essential ingredient of our proof is the presence of a mass gap for circularly symmetric solutions of \eqref{Einstein equations}. We define\footnote{Similarly to the BTZ family \eqref{BTZ definition}, we change the notation convention from the physics literature and translate the Hawking mass by 1.} the Hawking mass by:
\begin{equation}\label{Hawking mass}
    m=-g(\nabla r,\nabla r)-\Lambda r^2+1.
\end{equation}
We note that the Hawking mass is non-negative for solutions arising in gravitational collapse. Moreover, a solution is exactly AdS spacetime in the domain of dependence of a disk bounded by a circle if and only if the Hawking mass of the circle is zero.

For asymptotically AdS characteristic data $\mathfrak{D}\in\mathfrak{M},$ we define the Bondi mass $M=\lim_{r\rightarrow\infty}m|_{C_0^+}(r).$ We note that the reflective boundary conditions at $\mathcal{I}$ imply that $m|_{\mathcal{I}}$ is constant, so $M=m|_{\mathcal{I}}.$

The numerical relativity work \cite{ADS3} conjectures that initial data must have Bondi mass at least 1 to form singularities in evolution, and provides some evidence in the case of very small Bondi mass. In \Cref{mass gap section}, we prove the following local version of this conjecture in terms of the Hawking mass: %We refer to this property of circularly symmetric solutions of \eqref{Einstein equations} as the \textit{mass gap}.
\begin{restatable}{theorem}{extension}\label{mass gap theorem}
    For any $k\geq 2$, let $\big(\mathcal{M},g,\phi\big)$ be the maximal globally hyperbolic development of $C^k$ characteristic data on the outgoing null cone $C_0^+=\{u=u_0,\ v_{\Gamma}\leq v\leq\Tilde{v}\}$ with $r(u_0,v_{\Gamma})=0,\ r(u_0,\Tilde{v})<\infty.$ If $m(u_0,v_0)=1-\delta<1$ for some $v_{\Gamma}< v_0<\Tilde{v},$ then $\big(\mathcal{M},g,\phi\big)$ is regular in the domain of dependence of $C_0^+\cap\{v\leq v_0\}.$
\end{restatable}

We briefly explain the idea of the proof of \Cref{mass gap theorem}. We note that in \Cref{mass gap section} it is convenient to work in Bondi coordinates $(u,r).$ First, we show that the bound on $m$ implies quantitative control of the geometry of the spacetime. We then prove a non-concentration estimate for the scalar field, which in particular implies that the scalar field cannot blow up at a faster than self-similar rate, so $r\partial_r\phi$ is bounded. We further improve this estimate by a small power of $r,$ showing in particular that  $r^{1-\rho}\partial_r\phi$ is bounded, where $\rho>0$ is sufficiently small. The preliminary bounds for the geometry of the spacetime and the scalar field allow us to propagate from initial data boundedness for the local well-posedness norm, concluding that no singularities form in evolution.

The main use of \Cref{mass gap theorem} for the purposes of this paper is the following description of first singularities at the center:
\begin{corollary}
    Let $b_{\Gamma}=(u_*,v_*)$ be a first singularity at the center $\Gamma.$ Then $m(u,v_*)\rightarrow1^+$ as $u\rightarrow u_*^-.$
\end{corollary}

As a corollary of \Cref{mass gap theorem}, we also obtain the mass gap result in terms of the Bondi mass, thus proving the conjecture of \cite{ADS3}:
\begin{corollary}\label{bizcorollary}
    Solutions with Bondi mass less than 1 do not form singularities. More precisely, if $\mathfrak{D}\in\mathfrak{M}$ has $M<1,$ then $\mathfrak{D}\in\mathrm{int}\big(\mathfrak{M}_{\mathrm{non}}\big).$ In particular, small perturbations of AdS spacetime are non-collapsing.
\end{corollary}
This result marks a sharp contrast between AdS spacetime in $2+1$ dimensions and AdS spacetime in $3+1$ dimensions, since the latter is unstable to black hole formation \cite{dafhol,BzRo,nulldust,vlasov}. In view of the numerical study of \cite{ADS3}, it is expected that AdS spacetime in $2+1$ dimensions is dynamically unstable and perturbations of it exhibit weak turbulence. Finally, we refer the reader to \cite{dafhol}, which conjectures that the Eguchi-Hanson solitons in $4+1$ dimensions are unstable to naked singularity formation.

\subsubsection{Formation of trapped surfaces criterion}\label{trapped surface section intro}

The second essential ingredient in our argument is a quantitative criterion for trapped surface formation, inspired by the work of Christodoulou \cite{ChrBH}. Intuitively, this result states that if the mass within an annular region bounded by two circles on the outgoing null cone $C_0^+$ is large compared to the size of the region, then a trapped surface will form in evolution. We prove the following result in \Cref{trapped surface formation section}:
\begin{restatable}{theorem}{trapped}\label{trapped surface formation theorem}
    For any $k\geq 2$, let $\big(\mathcal{M},g,\phi\big)$ be the maximal globally hyperbolic development of $C^k$ characteristic data on the outgoing null cone $C_0^+=\{u=u_0,\ v_{\Gamma}\leq v\leq\Tilde{v}\}$ with $r(u_0,v_{\Gamma})=0,\ r(u_0,\Tilde{v})<\infty,$ and $\partial_vr(u_0,v)>0$ for $v\in[v_{\Gamma},\Tilde{v}].$ Assume that there exists an ingoing cone $C_1=\{u_0\leq u<0,\  v=v_1\}$ satisfying $r(u,v_1)\rightarrow0$ as $u\rightarrow0^-$ and $m(u,v_1)\geq1$ for all $u\in[u_0,0)$. We consider any ingoing cone $C_2=\{u\geq u_0,\ v=v_2\}$ such that $v_1<v_2<\Tilde{v},$ and we denote:
    \begin{equation}\label{def of delta0 and eta0}
        \delta_0=\frac{r(u_0,v_2)}{r(u_0,v_1)}-1,\ \eta_0=\frac{m(u_0,v_2)-m(u_0,v_1)}{|\Lambda|r^2(u_0,v_2)}.
    \end{equation}
    If $\delta_0>0$ is small enough and $\eta_0>\delta_0$, then a marginally trapped surface $(u_{\mathcal{A}},v_2)$ forms on $C_2$ with:
    \begin{equation}
        r(u_{\mathcal{A}},v_2)\geq\frac{5\delta_0/4}{1+\delta_0}r(u_0,v_2).
    \end{equation}
\end{restatable}

To prove \Cref{trapped surface formation theorem}, we adapt the proof of \cite{ChrBH} to the $(2+1)$-dimensional setting. The main idea is to use the monotonicity properties of the system \eqref{Einstein equations} in circular symmetry and derive a differential inequality for the quantity $\eta(u,v_2)=|\Lambda|^{-1}r^{-2}(u,v_2)\cdot\big(m(u,v_2)-m(u,v_1)\big).$

\begin{remark}
    In analogy with \cite{ChrBH}, one might naively expect on dimensional grounds that $\eta_0$ should be defined as $m(u_0,v_2)-m(u_0,v_1).$ However, since the cosmological constant $\Lambda$ scales like inverse length squared, the quantity $|\Lambda|r^2$ is dimensionless. Thus, the definition of $\eta_0$ in \eqref{def of delta0 and eta0} is dimensionally consistent.
\end{remark}

\subsubsection{Locally naked singularity spacetimes}\label{locally naked section intro}

As a foundation for the instability argument, in \Cref{naked singularity section} we establish a detailed description of locally naked singularity spacetimes in a suitable neighborhood of $b_{\Gamma}.$ 

We introduce the set-up of \Cref{naked singularity section}, which will be used in the present section and in \Cref{blueshift instability section intro}. We consider $\big(\mathcal{M},g,\phi\big)$ to be the maximal globally hyperbolic development of a $C^k$ characteristic data set $\mathfrak{D}\in\mathfrak{M}_{\mathrm{loc.naked}},$ posed on the outgoing null cone $C_0^+.$ We assume that the ingoing null cone $C_0^-$ passing through $b_{\Gamma}$ intersects $C_0^+$, and we choose double null coordinates $(u,v)$ such that: \[b_{\Gamma}=(0,0),\ C_0^-=\big\{v=0,\ 2r=-u\big\},\ C_0^+=\big\{u=u_0,\ 2r=v-u_0\big\}.\]
We define the renormalized lapse and renormalized outgoing derivative of the scalar field by:
\[\gamma=\frac{\Omega^2}{\partial_vr},\ \theta=\sqrt{r}\frac{\partial_v\phi}{\partial_vr}.\]
The quantity $\theta$ satisfies the following propagation equation:
\begin{equation}\label{du theta intro}
    \partial_u\theta=-\frac{\Lambda}{2}r\gamma\theta-\frac{1}{2\sqrt{r}}\partial_u\phi.
\end{equation}
We define the \textit{blueshift} as the integrating factor for \eqref{du theta intro}:
    \begin{equation}\label{blueshift definition eq intro}
        \mathbf{b}(u,v)=\int_{u_0}^u\frac{|\Lambda|}{2}r\gamma(u',v)du',\ \mathbf{b}(u)=\mathbf{b}(u,0).
    \end{equation}

The main implication of the mass gap property in \Cref{mass gap theorem} is that first singularities at the center have infinite blueshift. More precisely, in \Cref{blueshift lower bound proposition} we show that $\mathbf{b}$ is unbounded on $C_0^-$ and satisfies the following lower bound for any $u\in[u_0,0):$
\begin{equation}\label{lower bound blueshift intro}
    \mathbf{b}(u)\geq2\log\frac{u_0}{u}.
\end{equation}

\begin{figure}[H]
\begin{center}
\begin{tikzpicture}[scale=1]
    \coordinate (G0)   at (0,-2);
    \coordinate (bG)   at (0,3);
    \coordinate (I)    at (2.5,0.5);
    \coordinate (R)    at (4,2);
    \coordinate (Rp)   at (4.6,2.6);
    \coordinate (Bend) at (2.1,5.1);
    \coordinate (J)    at (3.55,1.55);
    \coordinate (K)    at (0.72,2.50);
    \fill[gray!20]
        (bG)
        .. controls (1.35,2.95) and (2.9,2.55) .. (R)
        -- (I)
        -- cycle;
    \draw (0,-2) .. controls (0.2,-0.5) and (0.2,1.5) .. (0,3);
    \draw[dashed] (bG) -- (Bend);
    \draw (bG) .. controls (1.35,2.95) and (2.9,2.55) .. (R);
    \draw (G0) -- (Rp);
    \draw (bG) -- (I);
    \filldraw[color=black, fill=white] (bG) circle (2pt);
    \draw (-0.1,1.1) node[anchor=north] {$\Gamma$}; 
    \draw (-0.05,3.08) node[anchor=east] {$b_{\Gamma}$}; 
    \draw (0.85,4) node[anchor=south] {$\mathcal{B}_0$}; 
    \draw (1.6,1.2) node[anchor=north] {$C_0^-$}; 
    \draw (1.6,-0.7) node[anchor=north] {$C_0^+$}; 
    \draw (2.4,2) node[anchor=north] {$\mathcal{P}$};
\end{tikzpicture}
\end{center}
\caption{Schematic representation of the region $\mathcal{P}\subset\mathcal{M}$.}
\label{fig:locallynakedsing}
\end{figure}
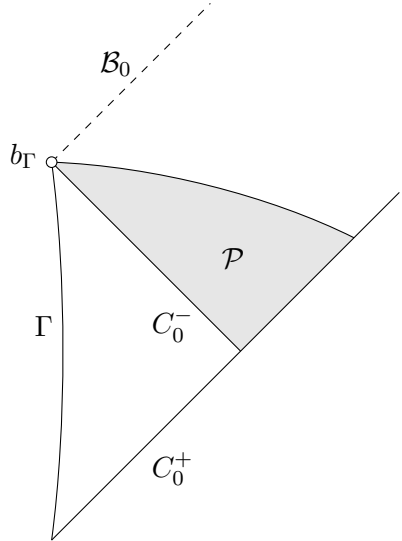

In \Cref{bounds for locally naked singularity proposition} we prove quantitative estimates for the solution in a suitable region $\mathcal{P},$ represented schematically in \Cref{fig:locallynakedsing}, whose definition depends on the strength of the blueshift, as made precise in \Cref{naked singularity section}. The remarkable aspect is that the absence of trapped surfaces in a neighborhood of $b_{\Gamma}$ implies by the contrapositive of \Cref{trapped surface formation theorem} bounds for the energy of the scalar field. More precisely, for any $(u,v)\in\mathcal{P}:$
\begin{equation}\label{bound for theta intro}
    \int_0^v\theta^2(u,v')dv'\lesssim v|u|\gamma(u,0).
\end{equation}
Additionally, in the region $\mathcal{P}$ we also prove that $\gamma(u,v)\sim\gamma(u,0),\ e^{\mathbf{b}(u,v)}\sim e^{\mathbf{b}(u)},\ \partial_vr(u,v)\sim e^{-\mathbf{b}(u)}$.

\subsubsection{The blueshift instability}\label{blueshift instability section intro}

In \Cref{instability section}, we prove that locally naked singularities can be perturbed to form trapped surfaces arbitrarily close to $b_{\Gamma}.$ We explain the main steps of the argument and highlight the key role of the blueshift instability.

We consider perturbations of the initial data on $C_0^+$, which are supported on $\{u=u_0,\ v\geq0\},$ so they leave the causal past of $b_{\Gamma}$ unchanged. By domain of dependence, it suffices to restrict to $v\leq v_0',$ for some $v_0'>0$ sufficiently small. For any integer $n>k+1$ and any $\lambda\neq0,$ we define the characteristic data $r_{\lambda}=r,\ \phi_{\lambda}=\phi$ on $\{u=u_0,\ v\in[u_0,0]\},$ and $r_{\lambda}=r,\ \theta_{\lambda}=\theta+\lambda v^n$ on $\{u=u_0,\ v\in[0,v_0']\},$ which determine a $C^k$ characteristic data set $\mathfrak{D}_{\lambda}.$ We denote by $\big(\mathcal{M}_{\lambda},g_{\lambda},\phi_{\lambda}\big)$  the maximal development of $\mathfrak{D}_{\lambda}.$

In \Cref{instability to trapped surfaces proposition}, we prove that $\big(\mathcal{M}_{\lambda},g_{\lambda}\big)$ contains a sequence of trapped surfaces converging to $b_{\Gamma}.$ The first step consists of establishing existence and quantitative estimates in a suitable region $\mathcal{P}_{\lambda}\subset\mathcal{P},$ represented schematically in \Cref{fig:instability}, which depends on $\lambda,\ n,$ and the strength of the blueshift on $C_0^-.$ In the region $\mathcal{P}_{\lambda}$, we show that $\gamma_{\lambda}(u,v)\sim\gamma(u,0),\ e^{\mathbf{b}_{\lambda}(u,v)}\sim e^{\mathbf{b}(u)},$ and $\partial_vr_{\lambda}(u,v)\sim e^{-\mathbf{b}(u)}$. %Moreover, we establish difference estimates for the quantities $\psi_{\lambda}-\psi$ with $\psi\in\big\{\gamma,r,\partial_vr,\partial_ur,\sqrt{r}\partial_u\phi\big\},$ which are consistent with the $n$-th order vanishing of the initial data perturbation at $v=0$. 
The main role of the blueshift instability is to amplify the initial data perturbation. To illustrate this, we use \eqref{du theta intro} for $\theta_{\lambda}$ and $\theta$ to derive the equation:
\begin{equation}\label{du theta difference intro}
    \partial_u\big(\theta_{\lambda}-\theta)=-\frac{\Lambda}{2}r_{\lambda}\gamma_{\lambda}\big(\theta_{\lambda}-\theta)+Err,
\end{equation}
where the term $Err$ is defined exactly in \Cref{A3 section}. The main contribution in $\theta_{\lambda}-\theta$ comes from the integrating factor (i.e.~the blueshift) acting on the initial data perturbation $\theta_{\lambda}(u_0,v)-\theta(u_0,v)=\lambda v^n$, while the contribution of the term $Err$ is proved to be of lower order. This behavior is captured by the following key estimate in \Cref{estimates for perturbed spacetime section}:
\begin{equation}\label{blueshift instability intro eq}
    \int_0^v\Big|\big(\theta_{\lambda}-\theta\big)(u,v')-e^{\mathbf{b}_{\lambda}(u,v')}\big(\theta_{\lambda}-\theta\big)(u_0,v')\Big|^2dv'\leq\lambda^2v^{2n+\frac{4}{3}}e^{2\mathbf{b}(u)}\ll\lambda^2v^{2n+1}e^{2\mathbf{b}(u)}.
\end{equation}

\begin{figure}
\begin{center}
\begin{tikzpicture}[scale=1]
    \coordinate (G0)   at (0,-2);
    \coordinate (bG)   at (0,3);
    \coordinate (I)    at (2.5,0.5);
    \coordinate (R)    at (4,2);
    \coordinate (Rp)   at (4.6,2.6);
    \coordinate (Bend) at (2.1,5.1);
    \coordinate (J)    at (3.55,1.55);
    \coordinate (K)    at (0.72,2.50);
    \fill[gray!45]
        (bG)
        .. controls (0.75,2.65) and (2.15,2.10) .. (J)
        -- (I)
        -- cycle;
    \draw (0,-2) .. controls (0.2,-0.5) and (0.2,1.5) .. (0,3);
    \draw (bG) .. controls (0.75,2.65) and (2.15,2.10) .. (J);
    \draw (bG)
    .. controls (0.16,2.90) and (0.42,2.70) .. (K)
    .. controls (1.00,2.33) and (1.22,2.32) .. (1.42,2.38)
    .. controls (1.62,2.45) and (1.84,2.58) .. (2.05,2.73)
    .. controls (2.24,2.87) and (2.50,3.00) .. (2.78,3.2);
    \fill (K) circle (1pt);
    \draw (G0) -- (Rp);
    \draw (bG) -- (I);
    \filldraw[color=black, fill=white] (bG) circle (2pt);
    \draw (-0.1,1.1) node[anchor=north] {$\Gamma$}; 
    \draw (-0.05,3.08) node[anchor=east] {$b_{\Gamma}$}; 
    \draw (1.6,1.2) node[anchor=north] {$C_0^-$}; 
    \draw (1.6,-0.7) node[anchor=north] {$C_0^+$}; 
    \draw (2.2,1.8) node[anchor=north] {$\mathcal{P}_{\lambda}$};
    \draw (1.75,3.1) node[anchor=west] {$\mathcal{C}_{\lambda}$};
    \draw[<-] (0.72,2.56) -- (0.72,3.1);
    \draw (0.72,3.29) node {$S_{\lambda}$}; 
\end{tikzpicture}
\end{center}
\caption{Schematic representation of the region $\mathcal{P}_{\lambda}\subset\mathcal{M}_{\lambda}$, the curve $\mathcal{C}_{\lambda},$ and $S_{\lambda}\in\mathcal{C}_{\lambda}.$}
\label{fig:instability}
\end{figure}
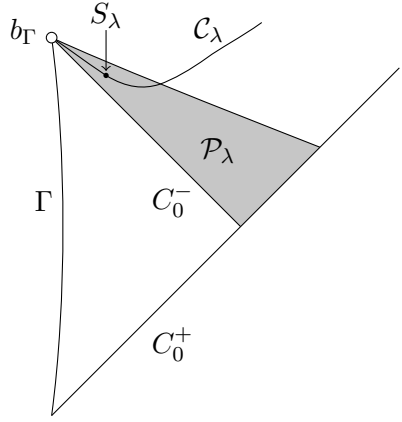

The second step in the proof of \Cref{instability to trapped surfaces proposition} consists of showing that the blueshift instability causes the amplification of $\theta_{\lambda}$ sufficiently close to $b_{\Gamma},$ and allows us to apply \Cref{trapped surface formation theorem} to obtain a sequence of trapped surfaces converging to $b_{\Gamma}.$ In \Cref{proof of instability thm section}, we define a suitable curve $\mathcal{C}_{\lambda},$ parametrized by $\big(u,v_{\sharp}(u)\big)$ for $u\in[u_0,0),$ which converges to $b_{\Gamma}$ and is contained in the region $\mathcal{P}_{\lambda}$ sufficiently close to $b_{\Gamma}$, as illustrated schematically in \Cref{fig:instability}. The curve $\mathcal{C}_{\lambda}$ is chosen such that \eqref{bound for theta intro} and \eqref{blueshift instability intro eq} imply the following lower bound for all $|u|$ sufficiently small:
\begin{equation}\label{lower bound for theta lambda intro}
        \int_0^{v_{\sharp}(u)}\big|\theta_{\lambda}(u,v')\big|^2dv'\gtrsim\lambda^2\big[v_{\sharp}(u)\big]^{2n+1}e^{2\mathbf{b}(u)}.
\end{equation}
We consider any $S_{\lambda}=\big(u,v_{\sharp}(u)\big)\in\mathcal{C}_{\lambda}$ sufficiently close to $b_{\Gamma},$ as in \Cref{fig:instability}. In \Cref{proof of instability thm section}, we show that the lower bound \eqref{lower bound for theta lambda intro} implies:
\begin{equation}\label{eta>delta intro}
    \eta_{\lambda}\big(u,v_{\sharp}(u)\big)=\frac{m_{\lambda}\big(u,v_{\sharp}(u)\big)-m_{\lambda}(u,0)}{|\Lambda|r_{\lambda}^2\big(u,v_{\sharp}(u)\big)}>\delta_{\lambda}\big(u,v_{\sharp}(u)\big)=\frac{r_{\lambda}\big(u,v_{\sharp}(u)\big)}{r_{\lambda}(u,0)}-1.
\end{equation}
Thus, \eqref{eta>delta intro} allows us to apply \Cref{trapped surface formation theorem} and show that the future ingoing null cone passing through $S_{\lambda}$ contains a trapped surface. As a result, we conclude that the spacetime $\big(\mathcal{M}_{\lambda},g_{\lambda}\big)$ contains a sequence of trapped surfaces converging to $b_{\Gamma}$.

\subsubsection{The proof of \Cref{main theorem}}\label{proof of main thm section intro}

We explain briefly how we complete the proof of \eqref{loc naked for main theorem} and \Cref{main theorem} in \Cref{WCC section}. In order to prove \eqref{loc naked for main theorem}, we show that any $\mathfrak{D}\in\mathfrak{M}_{\mathrm{loc.naked}}$ can be perturbed to some $\mathfrak{D}'\in\mathfrak{M}_{\mathrm{black}}^{\mathrm{spacelike}}.$ We notice that it is not always the case that for the maximal development of $\mathfrak{D}\in\mathfrak{M}_{\mathrm{loc.naked}}$ we have $C_0^-\cap C_0^+\neq\emptyset,$ as can be seen for example in \Cref{fig:bh1} and \ref{fig:nakedsing}. However, we can reduce the proof to the case $C_0^-\cap C_0^+\neq\emptyset$ by Cauchy stability, so we can use the results outlined in Sections~\ref{locally naked section intro}-\ref{blueshift instability section intro}. We note that the argument described in \Cref{blueshift instability section intro} constructs $\mathfrak{D}_{\lambda}$ only in a compact subset of $C_0^+.$ For $\lambda$ sufficiently small, we use the results in Appendix~\ref{appendix} to extend $\mathfrak{D}_{\lambda}$ to an asymptotically AdS characteristic data set $\mathfrak{D}'$ as a perturbation of $\mathfrak{D}$. We note this is a nontrivial step that is necessary due to the presence of a timelike asymptotic boundary at $\mathcal{I}.$ The results outlined in \Cref{blueshift instability section intro} show that the maximal development of $\mathfrak{D}_{\lambda}$ contains trapped surfaces arbitrarily close to $b_{\Gamma}$, which implies that $\mathfrak{D}'\in\mathfrak{M}_{\mathrm{black}}^{\mathrm{spacelike}},$ completing the proof of \eqref{loc naked for main theorem}. Having established \eqref{loc naked for main theorem}, the remaining statements in \Cref{main theorem} follow from \eqref{moduli space decomposition intro}, the fact that $\mathfrak{M}_{\mathrm{naked}}\subset\mathfrak{M}_{\mathrm{loc.naked}}$,  and the fact that $\mathfrak{M}_{\mathrm{black}}$ is open.

\subsection{Previous results}\label{previous results section}
We present some previous results relevant for the weak cosmic censorship conjecture. In \Cref{Christodoulou section}, we discuss \cite{ChrWCC} in detail and mention additional works regarding the instability of naked singularities. In \Cref{examples section}, we discuss some previous results on the construction of examples of naked singularity spacetimes. We refer the reader to the recent survey \cite{Wccsurvey} for an extended literature review.

\subsubsection{Instability of naked singularities}\label{Christodoulou section}

In the remarkable work \cite{ChrWCC}, Christodoulou proved a version of the weak cosmic censorship conjecture for the spherically symmetric Einstein-scalar field system in $3+1$ dimensions. We state the main result of \cite{ChrWCC} following the formulations in \cite{ChrWCC,C99-GL,Wccsurvey}:
\begin{theorem}[\cite{ChrWCC}]\label{Chr theorem}
    We parametrize initial data for the $(3+1)$-dimensional spherically symmetric Einstein-scalar field system of equations by a choice of $\vartheta=r\partial_v\phi$ along an initial outgoing null cone $C_0^+,$ normalized such that $\partial_vr=1/2.$ We consider as the space of initial data the space $AC$ of all absolutely continuous functions $\vartheta$ of finite total variation on the half-line. We denote by $\mathcal{E}\subset AC$ the subset consisting of initial data whose maximal development is a locally naked singularity spacetime. Then $\mathcal{E}$ has positive codimension in $AC.$
\end{theorem}

The present paper is inspired by the strategy of \cite{ChrWCC}, so we explain in more detail the common features and key differences of the two works. We begin with some comments regarding the differences in the statements of \Cref{main theorem} and \Cref{Chr theorem}. Firstly, the space of initial data in \Cref{Chr theorem} has a low-regularity topology, namely $AC$, which is necessary in order to activate the blueshift instability. On the other hand, \Cref{main theorem} holds in a topology of arbitrary smoothness, as the different scaling in our setting and the quantitative lower bound \eqref{lower bound blueshift intro} allow us to use the blueshift instability in the class of $C^k$ solutions, for any $k\geq2.$ Secondly, \Cref{Chr theorem} uses a stronger notion of genericity than \Cref{main theorem}, showing that $\mathcal{E}$ has positive codimension in $AC.$ The argument outlined in \Cref{blueshift instability section intro} is consistent with this stronger notion of genericity, but the need to extend the initial data perturbations globally and satisfy the boundary conditions at $\mathcal{I}$ forces us to settle for the notion of genericity in \Cref{main theorem}. For similar reasons, \Cref{main theorem} only proves the non-genericity of naked singularities, rather than locally naked singularities.

We make some further comments regarding the proofs of \Cref{main theorem} and \Cref{Chr theorem}. Firstly, we follow the strategy of \cite{ChrWCC} and use the formation of trapped surfaces criterion in two different ways: we use \Cref{trapped surface formation theorem} directly in \Cref{instability section}, but we also use the contrapositive of \Cref{trapped surface formation theorem} in \Cref{naked singularity section} to get a priori control of the scalar field for any locally naked singularity spacetime. Secondly, both \Cref{main theorem} and \Cref{Chr theorem} rely on the fundamental fact that locally naked singularity spacetimes have unbounded blueshift. In our case, this follows as a consequence of the presence of a mass gap in \Cref{mass gap theorem}, as explained in \Cref{locally naked section intro}. (Unfortunately, there is no published account of this part of the proof in the case of \Cref{Chr theorem}. Strictly speaking, from the published versions in \cite{ChrWCC,C99-GL}, one can only infer that \Cref{Chr theorem} applies for $\mathcal{E}\subset AC$ consisting of initial data whose maximal development is a locally naked singularity spacetime and has unbounded blueshift along the past lightcone of the singularity $b_{\Gamma}$, see also \cite{Wccsurvey}). %This assumption is justified in \cite{ChrWCC,C99-GL} by the connection to the extension principle in \cite{BV}. 
Thirdly, we follow the idea of \cite{ChrWCC} and use the blueshift effect to prove the instability of locally naked singularities. We explain how the two approaches are somewhat different. \cite{ChrWCC} derives refined necessary conditions for a spacetime to be a locally naked singularity, in the form of compatibility conditions between the initial data on $C_0^+$ and the solution on $C_0^-,$ which are then violated by suitable perturbations of the initial data. On the other hand, we start with the perturbation in \Cref{blueshift instability section intro} and we carry out most of the argument at the level of the perturbed spacetime, providing a detailed description of how the solution enters the regime where \Cref{trapped surface formation theorem} applies.

We mention some further results regarding the instability of naked singularities for the Einstein-scalar field system in $3+1$ dimensions. The work \cite{LiuLi} revisits the proof of \cite{ChrWCC}, while \cite{LiuLi2} considers non-spherically symmetric perturbations of spherically symmetric naked singularity spacetimes. Finally, in the recent work \cite{An}, An uses anisotropic perturbations to construct finitely many unstable directions for the examples of \cite{C94}.

\subsubsection{Previous examples of naked singularities and upcoming work}\label{examples section}

Constructing examples of naked singularities plays a fundamental role in understanding their general properties and their instability. Moreover, the existence of such examples justifies the need to include genericity conditions in any formulation of the weak cosmic censorship conjecture. We outline some of the previous results on the construction of examples of naked singularities, and we refer the reader to \cite{Wccsurvey,CKdss} for further discussion. We also discuss the upcoming work \cite{preparation}.

For the spherically symmetric Einstein-scalar field system in $3+1$ dimensions, \cite{C94} constructed examples of naked singularities, which were recently revisited\footnote{The recent works \cite{Singh2,Zheng,Zheng2} study the stability properties of the examples of \cite{C94}. In particular, Singh and Zheng show that the naked singularity spacetimes of \cite{C94} are stable to smooth perturbations. We refer the reader to \cite{Zheng} for the sharp result established.} by \cite{Singh1} and \cite{Singh2,Zheng,Zheng2}. Notably, the examples of \cite{C94} arise from low regularity initial data (more precisely, in the notation of \Cref{Chr theorem}, $\vartheta$ is only Hölder continuous). However, based on numerical relativity works on critical collapse \cite{Cho93,Gundlach2007}, it is conjectured that there exist naked singularities arising from smooth data on $\partial\mathfrak{M}_{\mathrm{black}}\cap\partial\mathfrak{M}_{\mathrm{non}},$ which represents the moduli space threshold between black hole formation and non-collapse. Outside of spherical symmetry, the groundbreaking work \cite{RSR23,SR22} constructed examples of naked singularities for the Einstein vacuum equations in $3+1$ dimensions, which also arise from low regularity initial data. 

In the case of certain matter models, there are more elementary examples of naked singularities arising from smooth initial data. The first known examples of naked singularities arising in gravitational collapse are due to Christodoulou \cite{dust} for the spherically symmetric Einstein-dust system in $3+1$ dimensions. Remarkably, \cite{dust} constructs examples of stable naked singularities, so weak cosmic censorship is false for this model. More recently, the remarkable work \cite{euler} constructed examples of naked singularities for the spherically symmetric Einstein--Euler system in $3+1$ dimensions. Since these models are prone to singularity formation even in the absence of gravity, it is not clear however what is the physical significance of these results for the issue of weak cosmic censorship.

In $2+1$ dimensions, to date, no rigorous examples of naked singularities for \eqref{Einstein equations} have been constructed. Similarly to our discussion in the previous paragraph, the numerical relativity works on critical collapse \cite{Pretorius,jalmuzna} suggest that there exist such examples arising from smooth data on $\partial\mathfrak{M}_{\mathrm{black}}\cap\partial\mathfrak{M}_{\mathrm{non}}.$ In the upcoming work \cite{preparation}, we prove a local version of this result and we provide the first rigorous construction of locally naked singularities arising in (local) critical collapse. More precisely, on the finite outgoing cone $C_0^+=\{u=u_0,\ 0\leq r<r_0\},$ we construct a 2-parameter family of $C^3$ characteristic data $\mathfrak{D}_{\lambda}^q$, for any integer $q\geq4$ and $\lambda\in(-\epsilon,\epsilon),$ which satisfies the following properties:
\begin{itemize}
    \item For $\lambda<0,$ the development of $\mathfrak{D}_{\lambda}^q$ does not form singularities.
    \item For $\lambda=0,$ the development of $\mathfrak{D}_{\lambda}^q$ is a locally naked singularity spacetime and contains no trapped surfaces. %, and $C_0^+\cap C_0^-=\{(u_0,r_0)\},$ with $0<r_0<r_1$. 
    Moreover, the case $(\lambda,q)=(0,4)$ corresponds to the numerical examples of \cite{Pretorius,jalmuzna}.
    \item For $\lambda>0,$ the development of $\mathfrak{D}_{\lambda}^q$ contains trapped surfaces.
\end{itemize}
In the above, it is then the $\lambda=0$ solutions which represent examples of locally naked singularities arising in (local) critical collapse. An interesting quantitative aspect of the examples of \cite{preparation} is that for $\lambda=0,$ the blueshift defined in \eqref{blueshift definition eq intro} satisfies the following rate as $u\rightarrow0^-:$
\begin{equation}\label{blueshift rate for CR}
    \mathbf{b}(u)=\bigg(3-\frac{1}{2q}\bigg)\log\frac{u_0}{u}+o(1),
\end{equation}
which should be compared to the general lower bound \eqref{lower bound blueshift intro} for first singularities at the center. 

In the context of the moduli space $\mathfrak{M}$ of the present paper, an immediate consequence of the results in \cite{preparation} is that $\mathfrak{D}_{0}^q$ can be extended to asymptotically AdS data $\widetilde{\mathfrak{D}}_{0}^q$ such that $\widetilde{\mathfrak{D}}_{0}^q\in\mathfrak{M}_{\mathrm{loc.naked}}.$ In particular, it follows that $\mathfrak{M}_{\mathrm{loc.naked}}\neq\emptyset.$ It would be an interesting problem to prove that the family $\mathfrak{D}_{\lambda}^q$ can be extended to a family $\widetilde{\mathfrak{D}}_{\lambda}^q$ of asymptotically AdS data, such that in fact $\widetilde{\mathfrak{D}}_{0}^q\in\mathfrak{M}_{\mathrm{naked}}$ and $\widetilde{\mathfrak{D}}_{0}^q\in\partial\mathfrak{M}_{\mathrm{black}}\cap\partial\mathfrak{M}_{\mathrm{non}}.$

\subsection{Outline of the paper}\label{outline section}
We outline the rest of the paper. In \Cref{set up section}, we introduce general properties of the Einstein-scalar field system \eqref{Einstein equations} in circular symmetry. In \Cref{mass gap section}, we prove the mass gap result in \Cref{mass gap theorem}. In \Cref{trapped surface formation section}, we prove the formation of trapped surface criterion in \Cref{trapped surface formation theorem}. In \Cref{naked singularity section}, we provide a detailed description of locally naked singularity spacetimes. In \Cref{instability section}, we prove the instability of locally naked singularity spacetimes to trapped surface formation. In \Cref{WCC section}, we prove the main result in \Cref{main theorem}. Finally, in Appendix~\ref{appendix} we construct asymptotically AdS characteristic data perturbations.

\paragraph{Acknowledgments.} The author would like to acknowledge Igor Rodnianski for his valuable advice in the process of writing this paper. The author would also like to thank Mihalis Dafermos and Christoph Kehle for the very helpful discussions and comments.

\section{The Einstein-scalar field system in circular symmetry}\label{set up section}

The goal of this section is to introduce the general properties of circularly symmetric solutions of the Einstein-scalar field system \eqref{Einstein equations}. In \Cref{eq in double null section}, we rewrite the system \eqref{Einstein equations} using double null coordinates. In \Cref{asymptotic ads section} we define the notion of asymptotically AdS characteristic data sets. In \Cref{well posedness section}, we introduce the local well-posedness theory for asymptotically AdS spacetimes. In \Cref{global structure section}, we state additional important properties of circularly symmetric solutions of \eqref{Einstein equations}, including a description of the general structure of maximal globally hyperbolic developments arising in gravitational collapse.

\subsection{Equations in double null coordinates}\label{eq in double null section}

In this section, we introduce global systems of double null coordinates for circularly symmetric spacetimes, which allow us to represent certain causal properties of spacetimes using Penrose diagrams. We also rewrite the Einstein-scalar field system of equations \eqref{Einstein equations} in double null coordinates.

We first make precise the notion of circular symmetry:
\begin{definition} We say that $\big(\mathcal{M},g,\phi\big)$ is circularly symmetric if the group $U(1)$ acts by isometry on $\mathcal{M}$ and preserves the scalar field $\phi.$ Moreover, we require that:
        \[\mathcal{Q}=\mathcal{M}/U(1)\]
    inherits the structure of a $1+1$ dimensional Lorentzian manifold with boundary. Denoting by $\overline{g}$ the metric induced on $\mathcal{Q},$ we also require that the metric $g$ on $\mathcal{M}$ takes the form:
    \[g=\overline{g}+r^2d\sigma_1^2,\]
    where $d\sigma_1^2$ is the standard metric on $S^1.$ Moreover, $r$ and $\phi$ are functions on $\mathcal{Q},$ and $r$ is nonnegative.
\end{definition}

As a $1+1$ dimensional Lorentzian manifold, $\mathcal{Q}$ can be globally covered by double null coordinates $(u,v)$. In these coordinates, the induced metric $\overline{g}$ takes the form:
\[\overline{g}=-\Omega^2dudv.\]
Thus, the spacetime $(\mathcal{M},g)$ can be covered globally by the double null coordinates $(u,v,\theta):$
\begin{equation}\label{metric in double null}
    g=-\Omega^2dudv+r^2d\theta^2.
\end{equation}
Since the metric on $\mathcal{Q}$ is conformal to that of the $1+1$ dimensional Minkowski space, we can represent
$\mathcal{Q}$ as a submanifold of $\mathbb{R}^{1+1}$ using Penrose diagrams. This approach offers a canonical way of attaching a boundary to the quotient $\mathcal{Q}$. We point out that we use the topology and causality of the ambient $\mathbb{R}^{1+1}$.

Next, we define two components of the boundary of $\mathcal{Q}$: the \textit{center} $\Gamma$ and \textit{null infinity} $\mathcal{I}.$ We denote by $\Gamma$ the portion of the boundary of $\mathcal{Q}$ corresponding to the set of fixed points of the action of $U(1)$, and note that $r=0$ on $\Gamma$. We also define the set $\mathcal{U}=\big\{u: \sup_{(u,v)\in\mathcal{Q}}r(u,v)=\infty\big\},$ which is nonempty according to \Cref{well posedness section}. We notice that for any $u\in\mathcal{U},$ there exists a unique $v_{\mathcal{I}}(u)$ such that $\big(u,v_{\mathcal{I}}(u)\big)\in\overline{\mathcal{Q}}\backslash\mathcal{Q}.$
\begin{definition}
    We define null infinity as the set:
    \[\mathcal{I}=\bigcup_{u\in\mathcal{U}}\big(u,v_{\mathcal{I}}(u)\big).\]
    According to \Cref{well posedness section}, null infinity $\mathcal{I}$ is a timelike curve. We denote by $i^+$ the future limit point of $\mathcal{I}$.
\end{definition}

We rewrite the Einstein-scalar field system with respect to double null coordinates. The system of equations \eqref{Einstein equations} is equivalent to:
\begin{align}
    2r\partial_u\partial_v\phi&=-\partial_ur\partial_v\phi-\partial_vr\partial_u\phi,\label{wave equation phi}\\
    2\partial_u\partial_vr&=\Lambda\Omega^2r,\label{wave equation r}\\
    2\partial_u\partial_v\log\Omega^{2}&=-2\partial_u\phi\partial_v\phi+\Lambda\Omega^2,\label{wave equation Omega}\\
    \partial_u\big(\Omega^{-2}\partial_ur\big)&=-\Omega^{-2}r(\partial_u\phi)^2,\label{Ray u}\\
    \partial_v\big(\Omega^{-2}\partial_vr\big)&=-\Omega^{-2}r(\partial_v\phi)^2\label{Ray v}.
\end{align}

In double null coordinates, the Hawking mass defined in \eqref{Hawking mass} takes the form:
\begin{equation}\label{Hawking mass double null}
    m(u,v)=4\Omega^{-2}\partial_ur\partial_vr-\Lambda r^2+1.
\end{equation}
We impose the boundary condition at the center $m|_{\Gamma}=0$, necessary for the regularity of $(\mathcal{M},g)$.

We derive the following propagation equations for the mass:
\begin{align}
    \partial_vm&=-4\Omega^{-2}\partial_ur\cdot r(\partial_v\phi)^2,\label{dv m} \\
    \partial_um&=-4\Omega^{-2}\partial_vr\cdot r(\partial_u\phi)^2.\label{du m}
\end{align}
We define the renormalized outgoing derivative of the scalar field by:
\begin{equation}\label{definition of theta}
    \theta=\sqrt{r}\frac{\partial_v\phi}{\partial_vr}.
\end{equation}
Finally, we define the quantity $\gamma,$ which plays a key role in our analysis:
\begin{equation}\label{definition of gamma}
    \gamma=\frac{\Omega^2}{\partial_vr}.
\end{equation}
Using this notation, the Raychaudhuri equation \eqref{Ray v} is equivalent to:
\begin{equation}\label{dv log gamma}
    \partial_v\log\gamma=\frac{r(\partial_v\phi)^2}{\partial_vr}=\theta^2\partial_vr.
\end{equation}

\subsection{Asymptotically AdS characteristic data}\label{asymptotic ads section}

It is essential to have a local well-posedness theory for circularly symmetric solutions of \eqref{Einstein equations}. In compact regions of spacetime, such results are standard in the literature, see for example \cite{BV,Dafermos_extension,kommemi}. The main difficulty lies in dealing with the boundary at infinity, as $\mathcal{I}$ represents a timelike asymptotic boundary of spacetime. Our strategy is to adapt the $(3+1)$-dimensional theory of \cite{localwellposed} to our $(2+1)$-dimensional setting. The main goal of this section is to define a suitable notion of asymptotically AdS characteristic data and the moduli space of initial data.

Since we want to prove a well-posedness theory which accounts for the boundary at infinity, the discussion only applies to the $J^-(\mathcal{I})$ region of spacetime, i.e. the causal past of $\mathcal{I}.$ In particular, in the case of black hole spacetimes, our discussion only applies to the exterior region as defined in \Cref{global structure section}. To emphasize the distinction from the Sections~\ref{eq in double null section} and \ref{global structure section}, in Sections~\ref{asymptotic ads section} and \ref{well posedness section} we use the notation $(U,V)$ for systems of double null coordinates that are adapted to the region $J^-(\mathcal{I})$.

The first goal is to define asymptotically AdS spacetimes. In particular, the definition must apply to the BTZ spacetimes \eqref{BTZ definition}. For the metric \eqref{BTZ definition}, we consider double null coordinates $(U,V)$ which are Eddington-Finkelstein normalized as follows:
\[\frac{\partial_Vr}{1-M+|\Lambda|r^2}=\frac{1}{2},\ \frac{\partial_Ur}{1-M+|\Lambda|r^2}=-\frac{1}{2},\ \text{and }U=V\text{ on }\mathcal{I}.\]
To achieve this, we define $U=t-r_*,\ V=t+r_*,$ where $r_*$ satisfies:
\[dr_*=\frac{dr}{1-M+|\Lambda|r^2},\ \lim_{r\rightarrow\infty}r_*(r)=0.\]
With respect to the Eddington-Finkelstein coordinates $(U,V),$ the BTZ metric \eqref{BTZ definition} becomes:
\begin{equation}\label{BTZ in double null}
    g_M=-\Big(1-M+|\Lambda|r^2\Big)dUdV+r^2d\theta^2.
\end{equation}

We remark that the leading behavior of the lapse at $\mathcal{I}$ in \eqref{BTZ in double null} is dictated by the cosmological constant and is independent of the dimension. We adapt the definition of asymptotically AdS spacetimes in \cite[Section~2.4]{localwellposed} to our setting:
\begin{definition}\label{asympt ads def}
    We say that a circularly symmetric spacetime $\big(\mathcal{M},g\big)$ is asymptotically AdS if for every $p\in\mathcal{I}$ there exists a coordinate system $(U,V)$ and a neighborhood of $p$ in $\mathcal{M}$:
    \[\Big\{V\geq U_0,\ U\in(V,U_0+\delta]\Big\},\]
    such that $U=V$ on $\mathcal{I},$ $\partial_Ur<0,\ \partial_Vr>0$ and the metric satisfies the asymptotic behavior for some $M\geq0$:
    \[g=-\Big(1+O\big(r^{-3}\big)\Big)\Big(1-M+|\Lambda|r^2\Big)dUdV+r^2d\theta^2,\]
    \[T(\Omega^2)=O\big(r^{-1}\big),\ R_*(\Omega^2)=O\big(r^{3}\big),\ T(r)=O\big(r^{-1}\big),\ R_*(r)=1-M+|\Lambda|r^2+O\big(r^{-1}\big),\]
    \[T(\partial_Ur)=O\big(r^{-1}\big),\ R_*(\partial_Ur)=O\big(r^{3}\big),\ T(\partial_Vr)=O\big(r^{-1}\big),\ R_*(\partial_Vr)=O\big(r^{3}\big),\]
    where $T=\partial_U+\partial_V,\ R_*=\partial_V-\partial_U.$
\end{definition}

Our next goal is to define characteristic initial data. We first define the notion of an asymptotically Eddington-Finkelstein coordinate, similarly to \cite[Proposition~3.1]{localwellposed}:

\begin{definition}
    We consider the interval $[0,\infty)$ parametrized by $r.$ Let $V:[0,\infty)\rightarrow I=\big[V_{\Gamma},V_{\mathcal{I}}\big)$ be a coordinate function such that $V(0)=V_{\Gamma},$ and $\lim_{r\rightarrow\infty}V(r)=V_{\mathcal{I}}.$ We say that the coordinate $V$ on $I$ is asymptotically Eddington-Finkelstein if there exists $M\geq0$ such that as $V\rightarrow V_{\mathcal{I}}$:
     \begin{equation}\label{asympt EF coordinate}
         \partial_Vr=\frac{1}{2}\big(1-M+|\Lambda|r^2\big)+o\big(r^{-1}\big),\ \partial_V\bigg(\frac{\partial_Vr}{1-M+|\Lambda|r^2}\bigg)=o(r^{-2}).
     \end{equation}
\end{definition}

We introduce the weighted spaces used in the definition of initial data:
\begin{definition}
    We define the weighted Sobolev norm $H^n\big([0,\infty)\big)$ for any $n\geq0:$
    \begin{equation}\label{weighted Sobolev space r}
        \big\|f\big\|_{H^n([0,\infty))}^2=\int_{0}^{\infty}\bigg[\big(1+r^2\big)^n\Big(\partial_r^nf\Big)^2+\ldots+\big(1+r^2\big)\Big(\partial_rf\Big)^2+f^2\bigg]rdr.
    \end{equation}
    With respect to an asymptotically Eddington-Finkelstein coordinate $V$ on  $I=\big[V_{\Gamma},V_{\mathcal{I}}\big)$, we define:
    \begin{equation}\label{weighted Sobolev space V}
        \big\|f\big\|_{H^n(I)}^2=\int_{V_{\Gamma}}^{V_{\mathcal{I}}}\bigg[\big(1+r^2\big)^n\bigg[\bigg(\frac{\partial_V}{\partial_Vr}\bigg)^nf\bigg]^2+\ldots+\big(1+r^2\big)\bigg(\frac{\partial_Vf}{\partial_Vr}\bigg)^2+f^2\bigg]r(V)\partial_VrdV.
    \end{equation}
\end{definition}

We define asymptotically AdS characteristic data sets, similarly to \cite[Definition 3.3]{localwellposed}:
\begin{definition}[Asymptotically AdS characteristic data]\label{asympt ads data def}
    For any $k\geq2,$ we say that the tuple:
    \[\mathfrak{D}=\big(r,\ldots,\partial_U^{k+1}r,\Omega,\ldots,\partial_U^k\Omega,m,\ldots,\partial_U^km,\phi,\ldots,\partial_U^k\phi,\mathcal{T}\phi,\ldots,\mathcal{T}^k\phi\big)\]
    represents a $C^k$ asymptotically AdS characteristic data set if there exists $M\geq0$ and an asymptotically Eddington-Finkelstein coordinate $V$ on $I=\big[V_{\Gamma},V_{\mathcal{I}}\big)$ satisfying \eqref{asympt EF coordinate}, such that the following hold:
    \begin{enumerate}
        \item Regularity conditions:
        \begin{equation}\label{regularity conditions geometry}
            \partial_U^{i}r\in C^{k+1-i}(I),i=\overline{0,k+1},\ \partial_U^{j}\Omega,\partial_U^{j}m\in C^{k-j}(I),j=\overline{0,k},
        \end{equation}
        \begin{equation}\label{regularity conditions scalar field L2}
            \phi\in H^{k+1}(I),\mathcal{T}\phi\in H^{k}(I),\ldots,\mathcal{T}^k\phi\in H^1(I),
        \end{equation}
        \begin{equation}\label{regularity conditions pointwise}
            \partial_U^i\phi\in C^{k-1-i}(I),i=\overline{0,k-1},\ \sqrt{r}\partial_V^{k-j}\partial_U^j\phi\in C^0(I),j=\overline{0,k}.
        \end{equation}
        %could even have 1/\log r instead of \sqrt{r}
        \item Compatibility conditions: the elements of $\mathfrak{D}$ satisfy the compatibility conditions obtained by differentiating equations \eqref{wave equation phi}-\eqref{du m} and taking the Kodama vector field $\mathcal{T}$ to be given by:
            \begin{equation}\label{definition of T}
                \mathcal{T}=\Omega^{-2}\partial_Vr\partial_U-\Omega^{-2}\partial_Ur\partial_V.
            \end{equation}
        \item No trapped or anti-trapped surfaces: $\partial_Vr>0,\ \partial_Ur<0$ on $I.$
        \item Boundary conditions at $\Gamma:$ for any $0\leq i\leq k+1,\ 0\leq j\leq k,$ and $0\leq l\leq k-2,$ we have that $\mathcal{T}^ir(V_{\Gamma})=0,\ \mathcal{T}^jm(V_{\Gamma})=0,$ $\mathcal{T}^l(\Omega^{-2}\partial_Vr\partial_U\phi+\Omega^{-2}\partial_Ur\partial_V\phi)(V_{\Gamma})=0.$ 
        \item Boundary conditions at $\mathcal{I}:$ $\lim_{V\rightarrow V_{\mathcal{I}}}m(V)=M.$
        \item Asymptotic behavior as $V\rightarrow V_{\mathcal{I}}$:
        \begin{itemize}
            \item $r$ satisfies \eqref{asympt EF coordinate} and for $0\leq i\leq k+1:$
            \begin{equation}\label{asympt behavior dv r}
                \partial_V^ir=O\big(r^{i+1}\big).
            \end{equation}
            \item $\partial_Ur$ satisfies the asymptotics for $0\leq i\leq k+1:$
            \begin{equation}\label{asympt behavior du r}
                \partial_Ur=-\frac{1}{2}\big(1-M+|\Lambda|r^2\big)+O\big(r^{-1}\big),\ \partial_V^i\partial_Ur=O\big(r^{i+2}\big).
            \end{equation}
            \item $\Omega$ satisfies the asymptotics for $0\leq i\leq k:$
            \begin{equation}\label{asympt behavior Lapse}
                \Omega^2=1-M+|\Lambda|r^2+O\big(r^{-1}\big),\ \partial_V^i\Omega^2=O\big(r^{i+2}\big).
            \end{equation}
            \item The scalar field satisfies the asymptotics for $0\leq i+j\leq k$:
            \begin{equation}\label{asympt behavior phi}
            \partial_V^{i}\partial_U^j\phi=O\big(r^{i+j-\frac{3}{2}}\big).
        \end{equation}
        \end{itemize}
    \end{enumerate}
\end{definition}

\begin{remark}\label{differences from HS remark}
    In \Cref{asympt ads data def} we adapted \cite[Definition 3.3]{localwellposed} from the $(3+1)$-dimensional case to the $(2+1)$-dimensional case, so we highlight some of the differences:
    \begin{enumerate}
        \item The initial data in our case models the setting of gravitational collapse, so it extends from $\Gamma$ to $\mathcal{I},$ whereas the definition of \cite{localwellposed} is localized to a neighborhood of $\mathcal{I}.$ Thus, we also prescribe regularity conditions at $\Gamma.$ Moreover, in view of the analysis of \eqref{wave equation phi} near $\Gamma$ in $(2+1)$-dimensions, for the scalar field it is convenient to assume Sobolev regularity at top order as in \eqref{regularity conditions scalar field L2}, which can be propagated by energy estimates, instead of only assuming the pointwise regularity conditions in \eqref{regularity conditions pointwise}.

        \item The initial data in \Cref{asympt ads data def} is prescribed on an outgoing null cone, whereas \cite{localwellposed} prescribes initial data on an ingoing null cone. Therefore, unlike \cite{localwellposed}, we cannot enforce the conditions in \eqref{regularity conditions scalar field L2}-\eqref{regularity conditions pointwise} by simply requiring suitable decay rates for $\phi,$ so we need to impose quantitative conditions on $\mathcal{T}\phi,\ldots,\mathcal{T}^k\phi$. A similar approach is taken by \cite{holzegellinear} in the setting of spacelike initial data, see \cite[Remark~5.1]{holzegellinear}.

        \item The notion of a characteristic data set $\mathfrak{D}$ is over-determined, subject to the compatibility conditions and boundary conditions at the center. Alternatively, one could attempt to prescribe only the function $\phi$ as initial data (referred to as \textit{seed data}), similarly to \cite{localwellposed}. However, since in our setting the initial data is prescribed on an outgoing null cone, \eqref{regularity conditions scalar field L2} imposes nonlinear and non-local constraints on $\phi$, see already \Cref{nonlocal nonlinear remark}, so the initial data conditions are stated more naturally for~$\mathfrak{D}$.
        \item We make assumptions on the asymptotic rates up to order $k\geq2$, whereas \cite{localwellposed} only requires asymptotic rates up to order $2.$ 
    \end{enumerate}
\end{remark}

\begin{remark}
    Equation \eqref{regularity conditions scalar field L2} represents the main quantitative condition in the definition of asymptotically AdS characteristic data sets and motivates our definition of a distance function below. In particular, \eqref{regularity conditions scalar field L2} implies \eqref{regularity conditions pointwise} by Sobolev embedding. Moreover, \eqref{regularity conditions scalar field L2} implies that $\lim_{r\rightarrow\infty}m(r)$ exists and that $\lim_{r\rightarrow\infty}\mathcal{T}m(r)=0,$ which are consistent with the reflective boundary conditions at infinity $m_{\mathcal{I}}=M.$
\end{remark}

The aim for the rest of the section is to define the moduli space of initial data. We define the following quantity measuring the size of asymptotically AdS characteristic data sets, which is independent of the choice of $(U,V)$ coordinates:
\begin{definition}
    We consider a $C^k$ asymptotically AdS characteristic data set $\mathfrak{D}$. We define:
    \begin{equation}\label{norm definition}
        \big\|\mathfrak{D}\big\|_{X([0,\infty))}=\sum_{i=0}^k\big\|\mathcal{T}^i\phi\big\|_{H^{k+1-i}([0,\infty))}.
    \end{equation}
\end{definition}

We note that the space of asymptotically AdS characteristic data sets is not a vector space, since the compatibility conditions in \Cref{asympt ads data def} are nonlinear, so $X([0,\infty))$ is not a norm. In the asymptotically flat case, the solution to this issue is to consider the linear space of seed data consisting of $\phi$ instead (see for example \cite{ChrWCC}). However, as explained in \Cref{nonlocal nonlinear remark}, for asymptotically AdS characteristic data equation \eqref{regularity conditions scalar field L2} imposes nonlinear constraints on $\phi$ already. Nevertheless, we use $X([0,\infty))$ to define a distance function as follows:

\begin{definition}
    For any two $C^k$ asymptotically AdS characteristic data sets $\mathfrak{D}$ and $\mathfrak{D}',$ we define the distance function:
    \begin{equation}\label{distance function definition}
        d_{X([0,\infty))}(\mathfrak{D},\mathfrak{D}')=\sum_{i=0}^k\Big\|\mathcal{T}^i\phi-\big(\mathcal{T}^i\phi\big)'\Big\|_{H^{k+1-i}([0,\infty))}.
    \end{equation}
\end{definition}

We note that any two asymptotically AdS characteristic data sets satisfying $d_{X([0,\infty))}(\mathfrak{D},\mathfrak{D}')=0$ are related by a change of double null coordinates $U\rightarrow\Tilde{U}(U), V\rightarrow\Tilde{V}(V)$. If this holds, we denote  $\mathfrak{D}\sim\mathfrak{D}'$, and we notice that $\sim$ defines an equivalence relation on the space of $C^k$ asymptotically AdS characteristic data sets.

We can finally use the distance function \eqref{distance function definition} to define a topology on the set of initial data. More specifically, we define the moduli space of $C^k$ asymptotically AdS characteristic data as follows:
\begin{definition}\label{moduli space definition}
    We define the moduli space of $C^k$ asymptotically AdS characteristic data as the metric space $\big(\mathfrak{M},d_{X([0,\infty))}\big),$ where $\mathfrak{M}$ is given by:
\[\mathfrak{M}=\big\{\mathfrak{D}:\ \mathfrak{D}\text{ is a } C^k \text{ asymptotically AdS characteristic data set}\big\}/_{\sim}.\]
\end{definition}
%We note that we prove a decomposition of the moduli space $\mathfrak{M}$ in \Cref{global structure section}.

\subsection{Local well-posedness theory}\label{well posedness section}
We state the main local well-posedness result, which applies to asymptotically AdS characteristic data as in \Cref{asympt ads data def}:
\begin{proposition}\label{lwp proposition}
    Let $k\geq2$ be an integer. 
    For any $C^{k}$ asymptotically AdS characteristic data set $\mathfrak{D}=\big(r,\ldots,\partial_U^{k+1}r,\Omega,\ldots,\partial_U^k\Omega,m,\ldots,\partial_U^km,\phi,\ldots,\partial_U^k\phi,\mathcal{T}\phi,\ldots,\mathcal{T}^k\phi\big)$ on the outgoing null cone $C_0^+=\big\{V_{\mathcal{I}}\big\}\times\big[V_{\Gamma},V_{\mathcal{I}}\big)$, there exists $\delta>0$ and a unique solution $\big(r,\Omega,\phi\big)$ of the equations \eqref{wave equation phi}-\eqref{Ray v} with reflective boundary conditions at infinity $m_{\mathcal{I}}=M,$ which is defined in the domain:
    \[\mathcal{Q}=\Big\{(U,V):\ U\in\big[V_{\mathcal{I}},V_{\mathcal{I}}+\delta\big],\ V\in\big[V_{\Gamma}(U),U\big)\Big\},\]
    where we have that the center $\Gamma$ and null infinity $\mathcal{I}$ given by:
    \[\Gamma=\Big\{\big(U,V_{\Gamma}(U)\big):\ U\in\big[V_{\mathcal{I}},V_{\mathcal{I}}+\delta\big]\Big\},\ \mathcal{I}=\Big\{\big(U,U\big):\ U\in\big[V_{\mathcal{I}},V_{\mathcal{I}}+\delta\big]\Big\}\]
    represent timelike curves. Moreover, the corresponding circularly symmetric spacetime $\big(\mathcal{M},g\big)$ is asymptotically AdS as in \Cref{asympt ads def}, and for each $U\in\big[V_{\mathcal{I}},V_{\mathcal{I}}+\delta\big]$ it induces $C^{k}$ asymptotically AdS  characteristic data on $\big\{U\big\}\times\big[V_{\Gamma}(U),U\big]$ as in \Cref{asympt ads data def}.
\end{proposition}

The proof of \Cref{lwp proposition} uses standard methods and follows in detail the steps of \cite{localwellposed}, so we do not repeat it here. We remark that one can also establish further standard properties, such as persistence of regularity and Cauchy stability. To illustrate how the methods of \cite{localwellposed} are adapted to our setting, we prove a priori estimates for $C^2$ solutions (in the sense of \Cref{asympt ads data def}) of the linear wave equation on asymptotically AdS spacetimes in a neighborhood of $\mathcal{I}$. It is clear that using the methods of \cite{localwellposed} and a propagation of regularity argument, one can turn this proof into the local well-posedness result.

\begin{proposition}\label{proposition linear wave on asympt ads}
    Let $\big(\mathcal{M},g\big)$ be a circularly symmetric asymptotically AdS spacetime. We consider a coordinate system $(U,V)$ and the domain:
    \[\mathcal{D}=\Big\{(U,V):\ U\in\big[V_{\mathcal{I}},V_{\mathcal{I}}+\delta\big],\ V\in\big[V_0,U\big),\ \ V_0<V_{\mathcal{I}}\Big\},\]
    such that the asymptotics in \Cref{asympt ads def} hold in $\mathcal{D}$, and similar rates hold for higher order derivatives of $\Omega^2$ and $r.$
    
    Let $\psi:\mathcal{D}\rightarrow\mathbb{R}$ be a $C^2$ solution of the wave equation with reflective boundary conditions:
    \begin{equation}
        \begin{cases}
            \square_{g}\psi=0,\\
            \psi=0,\ \text{ on }\mathcal{I},
        \end{cases}
    \end{equation}
    which satisfies the bound on $\mathcal{D}\cap\big\{U=V_{\mathcal{I}}\big\}$ and $\mathcal{D}\cap\big\{V=V_0\big\}$:
    \begin{equation}
        \sum_{i=0}^2\big\|\mathcal{T}^i\psi\big\|_{H^{3-i}\big(\{V_{\mathcal{I}}\}\times[V_0,V_{\mathcal{I}})\big)}+\sum_{i=0}^2\big\|\mathcal{T}^i\psi\big\|_{H^{3-i}\big([V_{\mathcal{I}},V_{\mathcal{I}}+\delta)\times\{V_0\}\big)}<C_0.
    \end{equation}
    Then, we have that $\psi$ satisfies the following estimates for any $\big(U,V\big)\in\mathcal{D}:$
    \begin{equation}\label{energy apriori estimate lwp}
            \sum_{i=0}^2\big\|\mathcal{T}^i\psi\big\|_{H^{3-i}\big(\{U\}\times[V_0,V]\big)}+\sum_{i=0}^2\big\|\mathcal{T}^i\psi\big\|_{H^{3-i}\big([\max(V,V_{\mathcal{I}}),U]\times\{V\}\big)}<C,
    \end{equation}
    \begin{equation}\label{pointwise apriori estimate lwp}
            r^{3/2}|\psi|+r^{1/2}|\partial_V\psi|+r^{1/2}|\partial_U\psi|+r^{-1/2}|\partial_V^2\psi|+r^{-1/2}|\partial_U^2\psi|<C,
    \end{equation}
    where the constant $C$ depends on $C_0$ and the background $g.$
\end{proposition}
\begin{proof}
    We follow the same steps as \cite{holzegellinear} and \cite[Proposition~6.4]{localwellposed}. The starting point is the following approximate conservation law derived from \eqref{wave equation phi}:
    \[\partial_U\Big(r(\partial_V\psi)^2\Big)+\partial_V\Big(r(\partial_U\psi)^2\Big)=-\partial_U\psi\partial_V\psi\cdot Tr,\]
    where $T=\partial_U+\partial_V.$ We integrate this equation in the domain $\mathcal{D}\cap\big\{U'\leq U,\ V'\leq V\big\}$ using the reflective boundary conditions, which imply that the flux terms on $\mathcal{I}$ vanish. The decay of $Tr$ in \Cref{asympt ads def} implies that the terms on the RHS are lower order, so we obtain the estimate:
    \begin{equation}\label{preliminary apriori estimate lwp}
        \int_{V_0}^{V}r\big(\partial_V\psi\big)^2\big(U,V'\big)dV'+\int_{\max(V,V_{\mathcal{I}})}^{U}r\big(\partial_U\psi\big)^2\big(U',V\big)dU'<C.
    \end{equation}
    Next, we commute \eqref{wave equation phi} by $\mathcal{T}$ and $\mathcal{T}^2$. Repeating the same proof, we obtain that \eqref{preliminary apriori estimate lwp} also holds for $\mathcal{T}\psi$ and $\mathcal{T}^2\psi.$

    The estimate \eqref{preliminary apriori estimate lwp} for $\psi,\mathcal{T}\psi,\mathcal{T}^2\psi,$ and the reflective boundary conditions also imply the pointwise estimate $r|\psi|+r|\mathcal{T}\psi|+r|\mathcal{T}^2\psi|\leq C.$ Furthermore, we have the following Hardy inequality for any function $f$ such that the RHS of \eqref{Hardy ineq near I} is finite:
    \begin{equation}\label{Hardy ineq near I}
        r^2f^2+\int_r^{\infty}\Tilde{r}f^2d\Tilde{r}\leq\lim_{\Tilde{r}\rightarrow\infty}\Tilde{r}^2f^2+\int_r^{\infty}\Tilde{r}^3\big(\partial_rf\big)^2d\Tilde{r}.
    \end{equation}
    Using \eqref{Hardy ineq near I} for $\psi,\mathcal{T}\psi,$ and $\mathcal{T}^2\psi$, we control the $H^1$ norms in \eqref{energy apriori estimate lwp}. Using \eqref{wave equation phi} as well, the $H^1$ bounds further imply that \eqref{energy apriori estimate lwp} holds. Moreover, we then get the pointwise estimates $|\partial_U\psi|+|\partial_V\psi|+|\partial_U\mathcal{T}\psi|+|\partial_V\mathcal{T}\psi|\leq C$ and $|\partial_U^2\psi|+|\partial_U\partial_V\psi|+|\partial_V^2\psi|\leq Cr.$

    For the rest of the proof we show the improved pointwise estimate \eqref{pointwise apriori estimate lwp}. We define the quantity:
    \[\zeta=\sqrt{r}\frac{\partial_U\psi}{\partial_Ur}.\]
    We compute that:
    \[\partial_V\zeta=-\frac{1}{2\sqrt{r}}\partial_V\psi-2|\Lambda|r\cdot\frac{\partial_Vr}{|\Lambda|r^2-m+1}\cdot\zeta=-\frac{T\psi}{2\sqrt{r}}+\bigg(\frac{\partial_Ur}{2r^2}-\frac{2|\Lambda|\partial_Vr}{|\Lambda|r^2-m+1}\bigg)r\zeta.\]
    We define $A=r^n\zeta,$ which satisfies the equation:
    \[\partial_VA=|\Lambda|rA\bigg(\frac{n\partial_Vr}{|\Lambda|r^2}+\frac{\partial_Ur}{2|\Lambda|r^2}-\frac{2\partial_Vr}{|\Lambda|r^2-m+1}\bigg)-r^n\frac{T\psi}{2\sqrt{r}}.\]
    We take $n=2,$ which results in having a bounded integrating factor. The remaining term on the RHS is lower order and can be controlled using the previous pointwise estimates. Therefore, $A$ is bounded, so we obtain the improved estimate $r^{1/2}|\partial_U\psi|\leq C.$ The pointwise bound for $\mathcal{T}\psi$ also implies that $r^{1/2}|\partial_V\psi|\leq C$ and $r^{3/2}|\psi|\leq C.$ Using the wave equation \eqref{wave equation phi}, we also have that $r^{-1/2}|\partial_U\partial_V\psi|\leq C$. The bound $|\partial_V\mathcal{T}\psi|+|\partial_U\mathcal{T}\psi|\leq C$ then implies $|\partial_U^2\psi|+|\partial_V^2\psi|\leq Cr^{1/2}.$
\end{proof}

\subsection{Global structure of spacetimes}\label{global structure section}

In this section, we state some general properties of circularly symmetric solutions of \eqref{Einstein equations}, including: monotonicity properties, a description of first singularities, the existence of maximal globally hyperbolic developments, properties of black hole spacetimes, a description of the general structure of the boundary of spacetimes arising in gravitational collapse, and the decomposition of the moduli space of initial data. The results of this section are analogous to the $(3+1)$-dimensional setting, see for example \cite{Dafermos_extension,Dafermos_trapped,localwellposed,kommemi}. Adapting the proofs to the $(2+1)$-dimensional setting is straightforward, so we simply record the results needed for the rest of the paper.

\paragraph{Monotonicity properties.}
We first state some important monotonicity properties, which follow directly from \eqref{wave equation r}, \eqref{du m}, \eqref{dv m}, and the boundary conditions at $\Gamma:$
\begin{lemma}\label{monotonicity lemma}
In $\mathcal{Q}$ we have that:
\begin{equation}\label{monotonicity relations}
    \partial_ur<0,\ \partial_u\partial_vr\leq0,\ m\geq0,\ \partial_vm\geq0.
\end{equation}
Moreover, $\partial_vr$ and $\partial_um$ have opposite signs. Additionally, $\partial_u\partial_vr<0$ on $\mathcal{Q}\backslash\Gamma$. Finally, in the causal past of $\Gamma$ we have $\partial_vr>0.$
\end{lemma}
Depending on the behavior of $r,$ we divide the spacetime $\mathcal{Q}$ into the following regions:
\begin{definition}
    We define the regular, trapped and marginally trapped regions of $\mathcal{Q}$:
    \begin{align*}
        \mathcal{R}&=\Big\{p\in\mathcal{Q}:\ \partial_vr(p)>0\Big\},\\
        \mathcal{T}&=\Big\{p\in\mathcal{Q}:\ \partial_vr(p)<0\Big\},\\
        \mathcal{A}&=\Big\{p\in\mathcal{Q}:\ \partial_vr(p)=0\Big\}.
    \end{align*}
\end{definition}
We record some simple properties of the trapped region, which follow directly from \eqref{Ray v},\eqref{Hawking mass double null}, and the monotonicity properties in \eqref{monotonicity relations}:
\begin{lemma}\label{trapped region lemma}
    We have an equivalent characterization of the trapped region in terms of the Hawking mass:
    \begin{equation}\label{trapped using mass}
        p\in\mathcal{T} \Longleftrightarrow m(p)>1-\Lambda r^2(p).
    \end{equation}
    Moreover, the trapped region is a future set: if $p\in\mathcal{T}$, then $J^+(p)\cap\mathcal{Q}\subset\mathcal{T}.$ Similarly, if $p\in\mathcal{T}\cup\mathcal{A}$, then $J^+(p)\cap\mathcal{Q}\subset\mathcal{T}\cup\mathcal{A}$ and $J^+(p)\cap\mathcal{Q}\cap\{u>u(p)\}\subset\mathcal{T}.$
\end{lemma}

\paragraph{First singularities.} We state the description of first singularities away from the center of \cite[Theorem 1.8]{kommemi} (see also \cite[Proposition A.1]{localwellposed}), adapted to our setting:

\begin{proposition}\label{extension principle away from center}
    Let $p\in\overline{\mathcal{Q}}\backslash\mathcal{Q}$ be a first singularity away from the center, i.e. $p\notin\overline{\Gamma}\backslash\Gamma.$ For any $q\in\big(I^-(p)\cap\mathcal{Q}\big)\backslash\{p\}$ such that $\mathcal{D}=\big(J^+(q)\cap J^-(p)\big)\backslash\{p\}\subset\mathcal{Q}$ we have that $\inf_{r\in\mathcal{D}}=0.$
\end{proposition}
The proof of this result follows the same steps as \cite[Theorem 1.8]{kommemi}, but is greatly simplified in our case due to the structure of equation \eqref{wave equation r}. In particular, \Cref{extension principle away from center} also implies that every first singularity away from the center must be preceded by trapped surfaces, which is analogous to the result of \cite{Dafermos_extension}. As a consequence, every first singularity in the regular region must occur at the center:
\begin{corollary}\label{first singularities at the center corollary}
    Let $p\in\overline{\mathcal{Q}}\backslash\mathcal{Q}$ be a first singularity such that $J^-(p)\backslash\{p\}\subset\mathcal{R}.$ Then $p\in\overline{\Gamma}\backslash\Gamma.$
\end{corollary}

\paragraph{Existence of maximal globally hyperbolic developments.} In view of the timelike boundary at infinity, we modify the usual definition of global hyperbolicity as follows, similarly to \cite[Definition~8.1]{localwellposed}: 
\begin{definition}
    An asymptotically AdS spacetime $\big(\mathcal{M},g\big)$ is globally hyperbolic if there exists an outgoing null cone $C_0^+$ such that $C_0^+\cap\Gamma\neq\emptyset,\ \overline{C_0^+}\cap\mathcal{I}\neq\emptyset$ and all past directed inextendible causal curves in $\mathcal{M}$ either intersect $C_0^+$ or have a limit endpoint on $\mathcal{I}$.
\end{definition}

It is standard to use the local well-posedness result in \Cref{lwp proposition} to prove the existence of a unique maximal globally hyperbolic development of initial data. We refer the reader to \cite[Section 4]{localwellposed} and the references therein.
\begin{proposition}
    For any $k\geq2$ and $C^{k}$ asymptotically AdS characteristic data set $\mathfrak{D}\in\mathfrak{M}$, there exists a maximal globally hyperbolic development $\big(\mathcal{M},g,\phi\big)$ of the initial data, which is unique up to isometry.
\end{proposition}

\paragraph{Black hole spacetimes.} As a consequence of the monotonicity properties established so far, we have that $J^-(\mathcal{I})\cap\mathcal{Q}\subset\mathcal{R}.$ Thus, the black hole region $\mathcal{BH}=\mathcal{Q}\backslash J^-(\mathcal{I})$ satisfies $\mathcal{T}\cup\mathcal{A}\subset\mathcal{BH}.$ In particular, if $\mathcal{T}\neq\emptyset$ then $\mathcal{BH}\neq\emptyset.$ In the case when $\mathcal{BH}\neq\emptyset,$ we define the event horizon $\mathcal{H}^+$ as the future boundary of $J^-(\mathcal{I})\cap\mathcal{Q}$ in $\mathcal{Q}.$

We state the following result for black hole spacetimes, similarly to \cite{Dafermos_trapped}. For the proof one adapts the arguments of \cite[Section 5]{price} and \cite[Section 8.2]{localwellposed} to the $(2+1)$-dimensional setting, once again using the good monotonicity properties in our system. 

\begin{proposition}\label{Penrose ineq and completeness proposition}
    Let $\big(\mathcal{M},g,\phi\big)$ be the maximal globally hyperbolic development of a $C^{k}$ asymptotically AdS characteristic data set with $k\geq2$, and denote $M=m|_{\mathcal{I}}.$ If $\mathcal{BH}\neq\emptyset,$ then $\mathcal{H}^+$ is future complete and satisfies the Penrose inequality $|\Lambda|r^2+1\leq M$. Moreover, in this case we also have that $\mathcal{I}$ is future complete.
\end{proposition}

\paragraph{Global structure of spacetimes.} We explain the general structure of the boundary of spacetimes. We first define the other possible components of the boundary, as in \cite{C99-GL,kommemi}.
\begin{definition}
    We denote $b_{\Gamma}=\overline{\Gamma}\backslash\Gamma.$  If $b_{\Gamma}\neq i^+$, we say that $b_{\Gamma}$ is a first singularity at the center. In this case, we define the \textit{central component of the singular boundary}:
    \vspace{-0.2em}
\[\mathcal{B}_0=\Big\{(u_{b_{\Gamma}},v)\in\overline{\mathcal{Q}}\Big\}.\]
\vspace{-0.5em}
Finally, we define the \textit{singular component of the boundary}:
\[\mathcal{B}=\Big\{(u,v)\in\overline{\mathcal{Q}}\backslash\mathcal{Q}:\ u_{i^+}<u<u_{b_{\Gamma}}\Big\}.\]
\end{definition}
\vspace{-0.2em}
Using the description of first singularities in \Cref{extension principle away from center} and the monotonicity properties \eqref{monotonicity relations}, it is straightforward to show that $\mathcal{B}$ is spacelike and $r=0$ on $\mathcal{B}$. Similarly to the $(3+1)$-dimensional case (see \cite[Section~1.4.1]{kommemi} and the references therein), one can prove that there is no incoming null component of the boundary arising from $i^+.$ Thus, the boundary of the quotient spacetime is given by: 
\begin{equation}
    \partial\mathcal{Q}=C_0^+\cup\mathcal{I}\cup i^+\cup\mathcal{B}\cup\mathcal{B}_0\cup b_{\Gamma}\cup\Gamma,
\end{equation}
and we point out that some of the boundary components might be empty.
\vspace{-0.4em}
\paragraph{The decomposition of the moduli space $\mathfrak{M}$.} We note that $\big(\mathcal{M},g\big)$ is a black hole spacetime if and only if $\mathcal{B}\neq\emptyset$. We define\footnote{In defining the subsets of $\mathfrak{M}$, we follow the notation convention of \cite{AKU26}.} the corresponding subset of the moduli space by:
\begin{equation}\label{M black}
    \mathfrak{M}_{\mathrm{black}}=\big\{\mathfrak{D}\in\mathfrak{M}:\ \mathcal{B}\neq\emptyset\big\}.
\end{equation}
Equivalently, $\mathfrak{D}\in\mathfrak{M}_{\mathrm{black}}$ if and only if $\mathcal{T}\neq\emptyset.$ The latter is an open condition, so $\mathfrak{M}_{\mathrm{black}}$ is open.

We can classify black hole spacetimes depending on the nature of the singularity inside the black hole region. In the case when $\mathcal{B}\neq\emptyset$ and $\mathcal{B}_0=\emptyset$, the singularity is spacelike-only and the spacetime has the Penrose diagram in \Cref{fig:bh2}.  We denote the corresponding subset of the moduli space by:
\begin{equation}\label{M black spacelike}
    \mathfrak{M}_{\mathrm{black}}^{\mathrm{spacelike}}=\big\{\mathfrak{D}\in\mathfrak{M}:\ \mathcal{B}\neq\emptyset,\ \mathcal{B}_0=\emptyset\big\}.
\end{equation}
We also have the case when $\mathcal{B}\neq\emptyset$ and $\mathcal{B}_0\neq\emptyset,$ so the spacetime has the Penrose diagram in \Cref{fig:bh1}.  We denote the corresponding subset of the moduli space by:
\begin{equation}\label{M black loc naked}
    \mathfrak{M}_{\mathrm{black}}^{\mathrm{loc.naked}}=\big\{\mathfrak{D}\in\mathfrak{M}:\ \mathcal{B}\neq\emptyset,\ \mathcal{B}_0\neq\emptyset\big\}.
\end{equation}
The superscript stands for locally naked, which we define below in \eqref{M locally naked}. Our discussion shows that $\mathfrak{M}_{\mathrm{black}}$ can be written as the disjoint union\footnote{We note that the results of \Cref{instability section} and \cite{preparation} can be used to show that $\mathfrak{M}_{\mathrm{black}}^{\mathrm{loc.naked}}\neq\emptyset$ and $\mathfrak{M}_{\mathrm{black}}^{\mathrm{spacelike}}\neq\emptyset$.} $\mathfrak{M}_{\mathrm{black}}=\mathfrak{M}_{\mathrm{black}}^{\mathrm{spacelike}}\sqcup\mathfrak{M}_{\mathrm{black}}^{\mathrm{loc.naked}}.$

Next, we have the case $\mathcal{B}=\emptyset$ and $\mathcal{B}_0=\emptyset$. As a result, we must have $b_{\Gamma}=i^+$ and the spacetime has the Penrose diagram in \Cref{fig:non}. We say that such solutions are \textit{non-collapsing}, and we define the corresponding subset of the moduli space as:
\begin{equation}\label{M non}
    \mathfrak{M}_{\mathrm{non}}=\big\{\mathfrak{D}\in\mathfrak{M}:\ \mathcal{B}=\emptyset,\ \mathcal{B}_0=\emptyset\big\}.
\end{equation}
We note that \Cref{bizcorollary} implies that $\mathrm{int}\big(\mathfrak{M}_{\mathrm{non}}\big)\neq\emptyset.$ Using the monotonicity properties of the system and a "zig-zag" argument as in \cite[Section~3]{SchwAdS}, one can prove that $\mathfrak{D}\in\mathfrak{M}_{\mathrm{non}}$ implies $\mathcal{I}$ is future complete and that $\mathfrak{D}\in\mathfrak{M}_{\mathrm{non}}$ if and only if $\Gamma$ is future complete.

In the case when $\mathcal{B}=\emptyset$ and $\mathcal{B}_0\neq\emptyset$, the spacetime has the Penrose diagram as in \Cref{fig:nakedsing}. We say that such solutions are \textit{naked singularities} and we define the corresponding subset of the moduli space as:
\begin{equation}\label{M naked}
    \mathfrak{M}_{\mathrm{naked}}=\big\{\mathfrak{D}\in\mathfrak{M}:\ \mathcal{B}=\emptyset,\ \mathcal{B}_0\neq\emptyset\big\}.
\end{equation}

Finally, we say that a spacetime is a locally naked singularity spacetime if it satisfies $\mathcal{B}_0\neq\emptyset.$ We denote the corresponding subset of the moduli space as:
\begin{equation}\label{M locally naked}
    \mathfrak{M}_{\mathrm{loc.naked}}=\big\{\mathfrak{D}\in\mathfrak{M}:\ \mathcal{B}_0\neq\emptyset\big\}.
\end{equation}
Equivalently, these are spacetimes with no trapped surfaces in a neighborhood of $b_{\Gamma}$ (this follows from \Cref{monotonicity lemma} and \Cref{extension principle away from center}). We note that the results of \cite{preparation} imply that $\mathfrak{M}_{\mathrm{loc.naked}}\neq\emptyset.$ Moreover, we can write $\mathfrak{M}_{\mathrm{loc.naked}}$ as the disjoint union:
\begin{equation}\label{M null in M locally naked}
    \mathfrak{M}_{\mathrm{loc.naked}}=\mathfrak{M}_{\mathrm{naked}}\sqcup\mathfrak{M}_{\mathrm{black}}^{\mathrm{loc.naked}}.
\end{equation}

In conclusion, our discussion shows that the moduli space of initial data can be written as the following disjoint unions:
\begin{align}\label{decomposition of M}
    \mathfrak{M}&=\mathfrak{M}_{\mathrm{black}}\sqcup\mathfrak{M}_{\mathrm{non}}\sqcup\mathfrak{M}_{\mathrm{naked}}=\mathfrak{M}_{\mathrm{black}}^{\mathrm{spacelike}}\sqcup\mathfrak{M}_{\mathrm{non}}\sqcup\mathfrak{M}_{\mathrm{loc.naked}}\\
    &=\mathfrak{M}_{\mathrm{black}}^{\mathrm{spacelike}}\sqcup\mathfrak{M}_{\mathrm{black}}^{\mathrm{loc.naked}}\sqcup\mathfrak{M}_{\mathrm{non}}\sqcup\mathfrak{M}_{\mathrm{naked}}.\notag
\end{align}

\begin{remark}\label{remark about I incomplete}
    We recall that \cite{GH78,C99-GL} define naked singularity spacetimes as having a future incomplete $\mathcal{I}$. We claim that this is a more restrictive notion than our definition of naked singularities. We define the following subset of the moduli space:
    \[\mathfrak{M}_{\mathrm{incomplete}}=\big\{\mathfrak{D}\in\mathfrak{M}:\ \mathcal{I}\text{ is future incomplete}\big\}.\]
    As explained above, the maximal developments of data in $\mathfrak{M}_{\mathrm{non}}$ have a  future complete  $\mathcal{I}$. Moreover, according to \Cref{Penrose ineq and completeness proposition}, black hole spacetimes have a future complete $\mathcal{I}$. Therefore, \eqref{decomposition of M} implies that $\mathfrak{M}_{\mathrm{incomplete}}\subset\mathfrak{M}_{\mathrm{naked}}$, so all spacetimes with a future incomplete $\mathcal{I}$ are naked singularities. We point out that, in principle, this inclusion could be strict. For example, a similar argument to \cite{Dafermos_trapped} implies that any naked singularity spacetime with $r$ uniformly bounded on $\mathcal{B}_0$ has a future complete $\mathcal{I}$.
\end{remark}

\section{The mass gap}\label{mass gap section}

The goal of this section is to establish the mass gap property. More concretely, we prove that no singularities can form in the domain of dependence of a disk bounded by a circle with Hawking mass strictly less than 1. We recall the statement of \Cref{mass gap theorem}:
\extension*

We first claim that $m\leq1-\delta$ in the domain of dependence of $C_0^+\cap\{v\leq v_0\}.$ We suppose there exists $p\in\mathcal{Q}$ in the domain of dependence of $C_0^+\cap\{v\leq v_0\},$ which satisfies $p\in\mathcal{A}$ and $J^-(p)\backslash\{p\}\subset\mathcal{R}$. \Cref{monotonicity lemma} implies $m(p)\leq1-\delta$, so \Cref{trapped region lemma} implies that $p\in\mathcal{R},$ a contradiction. As a consequence, the domain of dependence of $C_0^+\cap\{v\leq v_0\}$ is contained in the regular region $\mathcal{R}$ and $m\leq1-\delta$.

The proof of~\Cref{mass gap theorem} is by contradiction. According to \Cref{first singularities at the center corollary}, any first singularities in the boundary of the domain of dependence of $C_0^+\cap\{v\leq v_0\}$ must occur at the center. For the rest of the section, we work under the following assumption:
\begin{assumption}\label{assumption mass gap}
    We assume that for some $v_{\Gamma}<v_*\leq v_0,$ the future ingoing null cone: 
    \[C_*^-=\Big\{(u,v_*):\ u_0\leq u<u_*,\ r(u_*,v_*)=0\Big\}\] 
    starting at $(u_0,v_*)$ ends in a first singularity at the center at $(u_*,v_*)\in\overline{\Gamma}\backslash\Gamma.$ Also, we denote $r_*=r(u_0,v_*).$
\end{assumption}
The idea of the proof is to establish uniform $C^k$ estimates in the domain:
\begin{equation}\label{domain D}
    \mathcal{D}=\Big\{(u,v)\in\mathcal{Q}:\ u_0\leq u<u_*,\ v\leq v_*\Big\},
\end{equation}
which will contradict \Cref{assumption mass gap} by \Cref{lwp proposition}. We follow the strategy outlined in \Cref{mass gap intro section}.

Since the domain of dependence of $C_0^+\cap\{v\leq v_0\}$ is contained in $\mathcal{R},$ the Bondi coordinates $(u,r)$ are defined globally in $\mathcal{D}.$ It is convenient for the purpose of this section to change coordinates from the double null coordinates $(u,v)$ to the Bondi coordinates $(u,r)$, which we introduce in detail in \Cref{eqts in Bondi section}.

We outline the structure of the section. In \Cref{eqts in Bondi section}, we introduce Bondi coordinates. In \Cref{bounds on geometry section}, we prove bounds on certain geometric quantities. In \Cref{preliminary scalar field section}, we prove preliminary estimates for the scalar field. In \Cref{improved scalar field section}, we establish improved estimates for the scalar field. Finally, in \Cref{proof of mass gap section}, we complete the proof of \Cref{mass gap theorem}.

\subsection{Bondi coordinates}\label{eqts in Bondi section}
We make the change of coordinates $(u,v)\rightarrow\big(u,r(u,v)\big)$ from double null coordinates to Bondi coordinates. In Bondi coordinates $(u,r)$, the metric \eqref{metric in double null} takes the form:
\begin{equation}\label{metric in Bondi}
    g=-e^{2\nu}du^2-2e^{\nu+\lambda}dudr+r^2d\theta^2.
\end{equation}

The unknowns in Bondi coordinates are given by the functions $\big(\nu,\lambda,\phi\big).$ We also define the quantity $\overline{\gamma}=e^{\nu-\lambda}.$ The relation between the double null unknowns $\big(r,\Omega,\phi\big)$ and the Bondi unknowns is:
\begin{equation}
    \gamma=2e^{\nu+\lambda},\ -2\partial_ur=\overline{\gamma}=e^{\nu-\lambda},
\end{equation}
where we recall that $\gamma$ was defined in \eqref{definition of gamma} as $\gamma^{-1}=\Omega^{-2}\partial_vr.$ Therefore, we can regard $\big(\gamma,\overline{\gamma},\phi\big)$ as the unknowns in Bondi coordinates.

In Bondi coordinates, the Hawking mass \eqref{Hawking mass} takes the following form:
\begin{equation}\label{Hawking mass Bondi}
    m=1-2\cdot\frac{\overline{\gamma}}{\gamma}-\Lambda r^2=1-e^{-2\lambda}-\Lambda r^2.
\end{equation}

We fix the freedom to rescale the $u$ coordinate by setting $\nu=0$ on $\Gamma.$ Since $m=0$ on $\Gamma,$ we also have that $\lambda=0$ on $\Gamma,$ so we also get $\gamma=2,\ \overline{\gamma}=1$ on $\Gamma.$

We notice that $\partial_r$ is a future directed outgoing null vector. We then define the future directed ingoing null vector:
\begin{equation}\label{definition of n}
    n=\partial_u-\frac{1}{2}\overline{\gamma}\partial_r.
\end{equation}

The system of equations \eqref{wave equation phi}-\eqref{Ray v} is equivalent to:
\begin{align}
    &2\partial_u\partial_r\phi+\frac{1}{r}\partial_u\phi-\frac{1}{r}\partial_r\big(r\overline{\gamma}\partial_r\phi\big)=0\label{wave eq Bondi}\\
    &\partial_r\log\gamma=r\big(\partial_r\phi\big)^2\label{dr gamma Bondi}\\
    &\partial_r\log\overline{\gamma}=-2r\Lambda e^{2\lambda}\label{dr gamma bar Bondi}\\
    &nm=-4r\gamma^{-1}\big(n\phi\big)^2\label{n m Bondi}\\
    &\partial_rm=re^{-2\lambda}\big(\partial_r\phi\big)^2.\label{dr m Bondi}
\end{align}
For any $(u,r_0)\in\mathcal{D},$ we parametrize the future directed ingoing null curve passing through $(u,r_0)$ as:
\begin{equation}\label{definition of chi}
    \chi_{u,r_0}(s)=\bigg(s,r_0-\frac{1}{2}\int_{u}^s\overline{\gamma}\big(\chi_{u,r_0}(s')\big)ds'\bigg),
\end{equation}
with $s\in\big[u,u_{\Gamma}(u,r_0)\big]$, such that $\chi_{u,r_0}(u)=(u,r_0)$ and $\chi_{u,r_0}\big(u_{\Gamma}(u,r_0)\big)\in\Gamma.$ We note that $\chi_{u,r_0}$ has tangent vector $n.$ For example, the ingoing cone $C_*^-$ is parametrized as $\{\chi_{u_0,r_*}(s):\ s\in[u_0,u_*)\}.$ See \Cref{fig:domain D} for an illustration of this notation. 
\begin{figure}
\begin{center}
    \begin{tikzpicture}
        \draw (0,0) .. controls (0.2,2) and (0.2,3) .. (0,5);
        \draw (0,0) -- (2.5,2.5);
        \draw (0,5) -- (2.5,2.5);
        \draw (0.15,2.15) -- (1.5,3.5);
        \draw (0.15,2.85) -- (1.5,1.5);
        \filldraw[color=black, fill=black] (2.5,2.5) circle (1pt);
        \draw (2.5,2.5)  node[anchor=west] {$(u_0,r_*)$};
        \filldraw[color=black, fill=black] (1.5,3.5) circle (1pt);
        \draw (1.5,3.6)  node[anchor=west] {$\chi_{u_0,r_*}(u)$};
        \filldraw[color=black, fill=black] (1.5,1.5) circle (1pt);
        \draw (1.6,1.5)  node[anchor=west] {$(u_0,r_0)$};
        \filldraw[color=black, fill=black] (0.5,2.5) circle (1pt);
        \draw (0.6,2.5)  node[anchor=west] {$\chi_{u_0,r_0}(u)$};
        \filldraw[color=black, fill=black] (0.15,2.85) circle (1pt);
        \draw[<-] (0,2.85) -- (-2,3);
        \draw (-2,3)  node[anchor=south] {$\chi_{u_0,r_0}\big(u_{\Gamma}(u_0,r_0)\big)$};
        \filldraw[color=black, fill=white] (0,5) circle (2pt);
        \draw (0.1,1.5)  node[anchor=east] {$\Gamma$};
        \draw (-0.1,5)  node[anchor=south] {$b_{\Gamma}=(u_*,0)$};
        \draw (1.2,1)  node[anchor=north] {$C_0^+$};
        \draw (1.2,4)  node[anchor=south] {$C_*^-$};
    \end{tikzpicture}
\end{center}
    \caption{Points $(u_0,r_0),\ \chi_{u_0,r_0}(u),$ and $\chi_{u_0,r_0}\big(u_{\Gamma}(u_0,r_0)\big)$ in the domain $\mathcal{D}$.}
    \label{fig:domain D}
\end{figure}
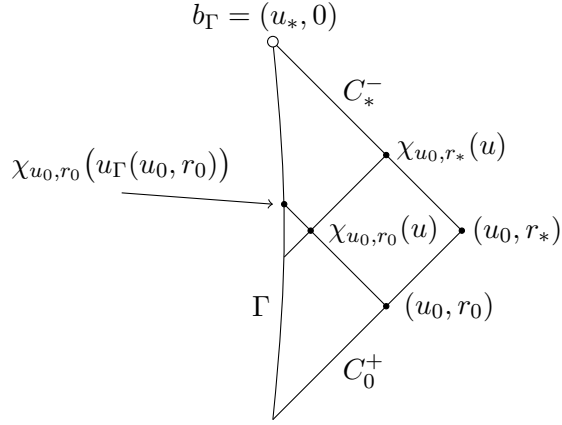
We also introduce the following notation:
\[r\big(\chi_{u,r_0}(s)\big):=r_0-\frac{1}{2}\int_{u}^s\overline{\gamma}\big(\chi_{u,r_0}(s')\big)ds'.\]
Using this notation, we can rewrite $\mathcal{D}$ as follows:
\begin{align*}
    \mathcal{D}&=\Big\{(u,r)\in\mathcal{Q}:\ u_0\leq u<u_*,\ 0\leq r\leq r\big(\chi_{u_0,r_*}(u)\big)\Big\}\\
    &=\Big\{\chi_{u_0,r}(u)\in\mathcal{Q}:\ 0\leq r\leq r_*,\ u_0\leq u\leq u_{\Gamma}(u_0,r) \Big\}.
\end{align*}
In practice, we also use truncated domains. For any $0<r_1\leq r_*$ and $u_0<u_1\leq u_{\Gamma}(u_0,r_1),$ we define:
\begin{align}
    \mathcal{D}_{u_1,r_1}&=\Big\{(u,r):\ u_0\leq u\leq u_1,\ 0\leq r\leq r\big(\chi_{u_0,r_1}(u)\big)\Big\}\label{truncated domain definition}\\
    &=\Big\{\chi_{u_0,r}(u)\in\mathcal{Q}:\ 0\leq r\leq r_1,\ u_0\leq u\leq \min\big(u_{\Gamma}(u_0,r),u_1\big)\Big\}.\notag
\end{align}
This domain represents the intersection of: the future of $C_0^+$, the past of the outgoing cone $\{u=u_1\}$, and the past of the ingoing cone passing through $(u_0,r_1).$

\subsection{Bounds on the geometry}\label{bounds on geometry section}

In this section, we derive a series of pointwise bounds on the quantities $\lambda,\nu,\gamma,\overline{\gamma}$ (together with their derivatives), as a consequence of the mass estimate $0\leq m\leq 1-\delta$ in $\mathcal{D}.$ At the end of the section, we also define the volume forms used in our argument. 

Shrinking the domain $\mathcal{D}$ if necessary (by increasing $u_0$), we can assume that $r_*$ satisfies:
\begin{equation}\label{bound on r*}
    r_*\leq\frac{\delta^3}{1000(1+|\Lambda|)(1+|u_*-u_0|)}\text{ and }r_*\leq\frac{1}{1000}.
\end{equation}

Firstly, \eqref{Hawking mass Bondi} implies that:
\begin{equation}\label{upper bound lambda}
    e^{2\lambda}\leq\frac{1}{\delta+|\Lambda|r^2}.
\end{equation}
As a consequence, \eqref{dr gamma bar Bondi} implies that:
\[0\leq\partial_r\big(\nu-\lambda\big)\leq\frac{2r|\Lambda|}{\delta+|\Lambda|r^2}=\frac{d}{dr}\log\bigg(1+\frac{|\Lambda|}{\delta}r^2\bigg).\]
We integrate this equation, using the boundary conditions at $\Gamma,$ to obtain:
\begin{equation}\label{bounds for overline gamma}
    1\leq\overline{\gamma}\leq1+\frac{|\Lambda|}{\delta}r^2.
\end{equation}
Since $\gamma=2e^{2\lambda}\overline{\gamma},$ the estimates \eqref{upper bound lambda}-\eqref{bounds for overline gamma} imply:
\begin{equation}\label{bounds for gamma}
    2\leq\gamma\leq\frac{2}{\delta}.
\end{equation}
We derive further bounds for the derivatives of $\overline{\gamma}$ using \eqref{dr gamma Bondi}, \eqref{dr gamma bar Bondi}, and \eqref{bounds for gamma}:
\begin{equation}\label{bounds for dr overline gamma}
    \partial_r\overline{\gamma}=|\Lambda|r\gamma\leq\frac{2|\Lambda|}{\delta}\cdot r
\end{equation}
\begin{equation}\label{bounds for dr^2 overline gamma}
    \partial_r^2\overline{\gamma}=|\Lambda|\gamma\Big(1+r^2(\partial_r\phi)^2\Big)\leq\frac{2|\Lambda|}{\delta}\Big(1+r^2(\partial_r\phi)^2\Big)
\end{equation}

We introduce the following notation, which will be used in the definition of volume forms:
\[\textbf{v}(u,r_0,s):=\frac{d}{dr_0}r\big(\chi_{u,r_0}(s)\big).\]
According to \eqref{definition of chi}, we have for any $s\in\big[u,u_{\Gamma}(u,r_0)\big]$:
\begin{equation}\label{volume form definition}
    \textbf{v}(u,r_0,s)=1-\frac{1}{2}\int_{u}^s\partial_r\overline{\gamma}\big(\chi_{u,r_0}(s')\big)\cdot\textbf{v}(u,r_0,s')ds'.
\end{equation}
Thus, we have the estimate:
\[\big|\textbf{v}(u,r_0,s)\big|\leq1+\frac{1}{2}\int_{u}^s\Big|\partial_r\overline{\gamma}\big(\chi_{u,r_0}(s')\big)\Big|\cdot\big|\textbf{v}(u,r_0,s')\big|ds'.\]
Using Grönwall's inequality, \eqref{bounds for dr overline gamma}, and \eqref{bound on r*}, we get that $\big|\textbf{v}(u,r_0,s)\big|<2.$ We further get that:
\begin{equation}\label{bound volume form}
\big|\textbf{v}(u,r_0,s)-1\big|\leq\int_{u}^s\Big|\partial_r\overline{\gamma}\big(\chi_{u,r_0}(s')\big)\Big|ds'\leq\frac{2|\Lambda|r_*|u_*-u_0|}{\delta}<\frac{1}{50}.
\end{equation}

For any $0<r_1\leq r_*$ and $u_0<u_1\leq u_{\Gamma}(u_0,r_1),$ we consider the truncated domain $\mathcal{D}_{u_1,r_1}$ defined in \eqref{truncated domain definition}.  For any quantity $E$ that is defined on $\mathcal{D}_{u_1,r_1}$, we denote its integral over $\mathcal{D}_{u_1,r_1}$ by:
\begin{align*}
    \iint_{\mathcal{D}_{u_1,r_1}}E\ d\mathrm{Vol}&=\int_0^{r_1}\int_{u_0}^{\min\{u_{\Gamma}(u_0,\Tilde{r}),u_1\}}E\big(\chi_{u_0,\Tilde{r}}(u)\big)\cdot\textbf{v}(u_0,\Tilde{r},u)dud\Tilde{r}\\
    &=\int_{u_0}^{u_1}\int_0^{r(\chi_{u_0,r_1}(u))}E(u,r)drdu.
\end{align*}

\subsection{Preliminary estimates for the scalar field}\label{preliminary scalar field section}
In this section, we establish preliminary estimates for the scalar field. In particular, we show in \eqref{rdrphi bound} that the scalar field cannot blow up at a rate that is faster than the self-similar rate. First, using \eqref{dr m Bondi} and \eqref{upper bound lambda}, we have the following a priori bound for the energy of the scalar field for any $(u,r)\in\mathcal{D}$:
\begin{equation}\label{preliminary bound for phi from m}
    \int_0^r\Tilde{r}\big(\partial_r\phi\big)^2(u,\Tilde{r})d\Tilde{r}\leq\frac{m(u,r)}{\delta}.
\end{equation}

We define $\psi=r^{3/2}\partial_r\phi.$ Since $k\geq2,$ \eqref{regularity conditions scalar field L2}-\eqref{regularity conditions pointwise} imply that there exists a constant $A>0$, such that:
\begin{equation}\label{assumption on initial data mass gap}
    \sup_{r\in[0,r_*]}\big|\partial_r\phi(u_0,r)\big|^2+\int_0^{r_*}\big(\partial_r\psi\big)^2(u_0,\Tilde{r})d\Tilde{r}\leq\frac{A^2}{1000}.
\end{equation}

We state the main result of the section:
\begin{proposition}\label{preliminary scalar field proposition}
     We have the following estimates for the scalar field:
    \begin{equation}\label{rdrphi bound}
        \sup_{(u,r)\in\mathcal{D}}\big|r\partial_r\phi(u,r)\big|\leq A,
    \end{equation}
    \begin{equation}\label{preliminary energy estimate psi}
        \sup_{u\in[u_0,u_*)}\int_0^{r(\chi_{u_0,r_*}(u))}\big(\partial_r\psi\big)^2(u,r)dr+\iint_{\mathcal{D}}\bigg(\frac{\big(\partial_r\psi\big)^2}{r}+\frac{\psi^2}{r^3}\bigg)\ d\mathrm{Vol}\leq A^2.
    \end{equation}
\end{proposition}

We first derive the following identity for $\psi,$ which is essential in proving energy estimates:
\begin{lemma}
    The function $\psi=r^{3/2}\partial_r\phi$ satisfies the identity:
    \begin{equation}\label{identity for psi}
        n\big(\partial_r\psi\big)^2+\frac{1}{r}\overline{\gamma}\big(\partial_r\psi\big)^2-\frac{3}{8r^2}\overline{\gamma}\partial_r\big(\psi^2\big)-2\partial_r\overline{\gamma}\big(\partial_r\psi\big)^2+\psi\partial_r\psi\cdot\bigg(\frac{1}{r}\partial_r\overline{\gamma}-\partial_r^2\overline{\gamma}\bigg)=0.
    \end{equation}
\end{lemma}
\begin{proof}
    The wave equation \eqref{wave eq Bondi} implies:
\[2\partial_u\partial_r\big(r\partial_r\phi\big)+\frac{1}{r}\partial_u\big(r\partial_r\phi\big)-\partial_r^2\big(r\overline{\gamma}\partial_r\phi\big)=0.\]
Using the definition of $\psi$, we get that:
\[2\partial_u\partial_r\psi-r^{1/2}\partial_r^2\big(r^{-1/2}\overline{\gamma}\psi\big)=0.\]
Expanding the second term, we further get that:
%\[2\partial_u\partial_r\psi-r^{1/2}\partial_r\big(r^{-1/2}\overline{\gamma}\partial_r\psi+r^{-1/2}\partial_r\overline{\gamma}\psi-1/2r^{-3/2}\overline{\gamma}\psi\big)=0\]
\[2\partial_u\partial_r\psi-\overline{\gamma}\partial_r^2\psi-2\partial_r\overline{\gamma}\partial_r\psi+\frac{1}{r}\overline{\gamma}\partial_r\psi-\frac{3}{4r^2}\overline{\gamma}\psi-\psi\partial_r^2\overline{\gamma}+\frac{1}{r}\psi\partial_r\overline{\gamma}=0.\]
Next, we use \eqref{definition of n} to get the equation:
\[2n\partial_r\psi+\frac{1}{r}\overline{\gamma}\partial_r\psi-\frac{3}{4r^2}\overline{\gamma}\psi-2\partial_r\overline{\gamma}\partial_r\psi+\frac{1}{r}\psi\partial_r\overline{\gamma}-\psi\partial_r^2\overline{\gamma}=0.\]
Multiplying by $\partial_r\psi,$ we conclude the proof of \eqref{identity for psi}.
\end{proof}

\begin{proof}[Proof of \Cref{preliminary scalar field proposition}]
    For some $0<r_1\leq r_*$ and $u_0<u_1\leq u_{\Gamma}(u_0,r_1),$ we make the bootstrap assumption for all $(u,r)\in\mathcal{D}_{u_1,r_1}$:
    \begin{equation}\label{bootstrap assumption preliminary mass gap}
        \big|\partial_r\phi(u,r)\big|\leq\frac{A}{r}.
    \end{equation}

    We note that a direct computation using \eqref{volume form definition} implies that:
    \[n\big(\partial_r\psi\big)^2\big(\chi_{u_0,\Tilde{r}}(s)\big)\cdot \textbf{v}(u_0,\Tilde{r},s)=\frac{d}{ds}\bigg(\big(\partial_r\psi\big)^2\big(\chi_{u_0,\Tilde{r}}(s)\big)\cdot \textbf{v}(u_0,\Tilde{r},s)\bigg)+\frac{1}{2}\big(\partial_r\psi\big)^2\partial_r\overline{\gamma}\big(\chi_{u_0,\Tilde{r}}(s)\big)\cdot \textbf{v}(u_0,\Tilde{r},s).\]
    For any $\chi_{u_0,\Tilde{r}}(u)\in\mathcal{D}_{u_1,r_1}$ we integrate \eqref{identity for psi} along $\chi_{u_0,\Tilde{r}}$ for $s\in[u_0,u]:$
    \[\big(\partial_r\psi\big)^2\big(\chi_{u_0,\Tilde{r}}(u)\big)\cdot \textbf{v}(u_0,\Tilde{r},u)+\int_{u_0}^u\bigg(\frac{\overline{\gamma}}{r}\big(\partial_r\psi\big)^2\big(\chi_{u_0,\Tilde{r}}(s)\big)-\frac{3\overline{\gamma}}{8r^2}\partial_r\psi^2\big(\chi_{u_0,\Tilde{r}}(s)\big)\bigg)\cdot \textbf{v}(u_0,\Tilde{r},s)ds=\]\[=\big(\partial_r\psi\big)^2(u_0,\Tilde{r})+\int_{u_0}^u\bigg(\frac{3}{2}\partial_r\overline{\gamma}\big(\partial_r\psi\big)^2\big(\chi_{u_0,\Tilde{r}}(s)\big)+\psi\partial_r\psi\cdot\bigg(\partial_r^2\overline{\gamma}-\frac{1}{r}\partial_r\overline{\gamma}\bigg)\big(\chi_{u_0,\Tilde{r}}(s)\big)\bigg)\cdot \textbf{v}(u_0,\Tilde{r},s)ds.\]
    We notice that for $k\geq2,$ \eqref{regularity conditions pointwise} implies that $\partial_r\psi=0$ on $\Gamma.$ For any $\chi_{u_0,r_0}(u)\in\mathcal{D}_{u_1,r_1},$ we integrate the identity for all $\Tilde{r}\in[0,r_0]:$
    \begin{align}\label{energy identity preliminary}
        &\int_0^{r(\chi_{u_0,r_0}(u))}\big(\partial_r\psi\big)^2(u,r)dr+\iint_{\mathcal{D}_{u,r_0}}\bigg(\frac{\overline{\gamma}}{r}\big(\partial_r\psi\big)^2-\frac{3\overline{\gamma}}{8r^2}\partial_r\psi^2\bigg)\ d\mathrm{Vol}\\
        &=\int_0^{r_0}\big(\partial_r\psi\big)^2(u_0,\Tilde{r})d\Tilde{r}+\iint_{\mathcal{D}_{u,r_0}}\bigg(\frac{3}{2}\partial_r\overline{\gamma}\big(\partial_r\psi\big)^2+\psi\partial_r\psi\bigg(\partial_r^2\overline{\gamma}-\frac{1}{r}\partial_r\overline{\gamma}\bigg)\bigg)\ d\mathrm{Vol}.\notag
    \end{align}
    We deal with the following error term by integration by parts, since $r^{-1}\psi=0$ on $\Gamma$:
    \begin{equation}\label{first error term int by parts}
        \iint_{\mathcal{D}_{u,r_0}}\frac{3\overline{\gamma}}{8r^2}\partial_r\psi^2\ d\mathrm{Vol}=\int_{u_0}^u\frac{3\overline{\gamma}}{8r^2}\psi^2\big(\chi_{u_0,r_0}(s)\big)ds+\iint_{\mathcal{D}_{u,r_0}}\bigg(\frac{3\overline{\gamma}}{4r^3}-\frac{3\partial_r\overline{\gamma}}{8r^2}\bigg)\psi^2\ d\mathrm{Vol}.
    \end{equation}
    We have the following Hardy inequality, again using $r^{-1}\psi=0$ on $\Gamma$:
    \begin{equation}\label{Hardy mass gap}
        \frac{\overline{\gamma}\psi^2}{r^2}(s,r)+\int_0^r\frac{\overline{\gamma}}{\Tilde{r}^3}\psi^2(s,\Tilde{r})d\Tilde{r}\leq\int_0^r\frac{\overline{\gamma}}{\Tilde{r}}\big(\partial_r\psi\big)^2(s,\Tilde{r})d\Tilde{r}+\int_0^r\frac{\psi^2\partial_r\overline{\gamma}}{\Tilde{r}^2}(s,\Tilde{r})d\Tilde{r}.
    \end{equation}
    We use \eqref{Hardy mass gap} for $(s,r)=\chi_{u_0,r_0}(s)$ in \eqref{first error term int by parts} to get that:
    \begin{equation}\label{first error term bound}
        \iint_{\mathcal{D}_{u,r_0}}\frac{3\overline{\gamma}}{8r^2}\partial_r\psi^2\ d\mathrm{Vol}\leq\iint_{\mathcal{D}_{u,r_0}}\bigg(\frac{3\overline{\gamma}}{4r}\big(\partial_r\psi\big)^2+\frac{3\psi^2\partial_r\overline{\gamma}}{8r^2}\bigg)\ d\mathrm{Vol}.
    \end{equation}
    As a consequence of \eqref{Hardy mass gap}, we also have the estimate:
    \begin{equation}\label{one more error term bound}
        \frac{1}{8}\iint_{\mathcal{D}_{u,r_0}}\frac{\overline{\gamma}}{r^3}\psi^2\ d\mathrm{Vol}\leq\frac{1}{8}\iint_{\mathcal{D}_{u,r_0}}\bigg(\frac{\overline{\gamma}}{r}\big(\partial_r\psi\big)^2+\frac{\psi^2\partial_r\overline{\gamma}}{r^2}\bigg)\ d\mathrm{Vol}.
    \end{equation}

    Next, we consider the following error term and use \eqref{bounds for dr^2 overline gamma}:
    \[\iint_{\mathcal{D}_{u,r_0}}\psi\partial_r\psi\cdot\partial_r^2\overline{\gamma}\ d\mathrm{Vol}=\iint_{\mathcal{D}_{u,r_0}}|\Lambda|\gamma\bigg(\psi\partial_r\psi+\frac{\partial_r\psi^4}{4r}\bigg)\ d\mathrm{Vol}.\]
We integrate by parts the second term to obtain:
\begin{align}
    \iint_{\mathcal{D}_{u,r_0}}|\Lambda|\gamma\frac{\partial_r\psi^4}{4r}\ d\mathrm{Vol}&=\iint_{\mathcal{D}_{u,r_0}}|\Lambda|\psi^4\bigg(\frac{\gamma}{4r^2}-\frac{\partial_r\gamma}{4r}\bigg)\ d\mathrm{Vol}+\int_{u_0}^u|\Lambda|\gamma\frac{\psi^4}{4r}\big(\chi_{u_0,r_0}(s)\big)ds\label{second error term bound}\\
    &\leq\iint_{\mathcal{D}_{u,r_0}}|\Lambda|\psi^4\frac{\gamma}{4r^2}\ d\mathrm{Vol}+\int_{u_0}^u|\Lambda|\gamma\frac{\psi^4}{4r}\big(\chi_{u_0,r_0}(s)\big)ds.\notag
\end{align}

The error term estimates \eqref{first error term bound}-\eqref{second error term bound} and the energy identity \eqref{energy identity preliminary} imply for any $\chi_{u_0,r_0}(u)\in\mathcal{D}_{u_1,r_1}$:
\begin{align}
    \int_0^{r(\chi_{u_0,r_0}(u))}&\big(\partial_r\psi\big)^2(u,r)dr+\iint_{\mathcal{D}_{u,r_0}}\bigg(\frac{\overline{\gamma}}{r}\big(\partial_r\psi\big)^2+\frac{\overline{\gamma}}{r^3}\psi^2\bigg)\ d\mathrm{Vol}\label{energy estimate preliminary mass gap in prop}\\
    \leq16&\int_0^{r_0}\big(\partial_r\psi\big)^2(u_0,\Tilde{r})d\Tilde{r}+16\iint_{\mathcal{D}_{u,r_0}}\bigg(\partial_r\overline{\gamma}\big(\partial_r\psi\big)^2+\frac{\psi^2\partial_r\overline{\gamma}}{r^2}+|\Lambda|\gamma\cdot|\psi\partial_r\psi|\bigg)\ d\mathrm{Vol}\notag\\
    &+16\iint_{\mathcal{D}_{u,r_0}}|\Lambda|\psi^4\frac{\gamma}{4r^2}\ d\mathrm{Vol}+16\int_{u_0}^u|\Lambda|\gamma\frac{\psi^4}{4r}\big(\chi_{u_0,r_0}(s)\big)ds.\notag
\end{align}
We use \eqref{bounds for gamma} and \eqref{bounds for dr overline gamma} to bound the error terms:
\begin{align*}
\iint_{\mathcal{D}_{u,r_0}}\bigg(\partial_r\overline{\gamma}\big(\partial_r\psi\big)^2+\frac{\psi^2\partial_r\overline{\gamma}}{r^2}+|\Lambda|\gamma\cdot|\psi\partial_r\psi|\bigg)\ d\mathrm{Vol}&\leq \frac{4|\Lambda|}{\delta}\iint_{\mathcal{D}_{u,r_0}}\bigg(r\big(\partial_r\psi\big)^2+\frac{\psi^2}{r}\bigg)\ d\mathrm{Vol}\\
    &\leq\frac{4|\Lambda|}{\delta}r_*^2\iint_{\mathcal{D}_{u,r_0}}\bigg(\frac{\big(\partial_r\psi\big)^2}{r}+\frac{\psi^2}{r^3}\bigg)\ d\mathrm{Vol},
\end{align*}
which can be absorbed on the LHS of the energy estimate using \eqref{bound on r*} and \eqref{bounds for overline gamma}. We use the bootstrap assumption \eqref{bootstrap assumption preliminary mass gap}, the preliminary estimate \eqref{preliminary bound for phi from m}, and \eqref{bound on r*}:
\[\iint_{\mathcal{D}_{u,r_0}}|\Lambda|\psi^4\frac{\gamma}{4r^2}\ d\mathrm{Vol}\leq\frac{|\Lambda|A^2}{2\delta}r_*\iint_{\mathcal{D}_{u,r_0}}\frac{\psi^2}{r^2}\ d\mathrm{Vol}\leq\frac{|\Lambda|A^2|u-u_0|\cdot m(u_0,r_0)\cdot r_*}{2\delta^2}\leq\frac{A^2}{1000}.\]

Additionally, using Cauchy-Schwarz and \eqref{preliminary bound for phi from m}, we also have:
\[\psi^4(s,r)=\bigg(2\int_0^r\psi\partial_r\psi(s,\Tilde{r})d\Tilde{r}\bigg)^2\leq\frac{4m(s,r)}{\delta}\int_0^r\Tilde{r}^2\big(\partial_r\psi\big)^2(s,\Tilde{r})d\Tilde{r}.\]
We use this for $(s,r)=\chi_{u_0,r_0}(s)$ to get for the last error term in \eqref{energy estimate preliminary mass gap in prop}:
\[\int_{u_0}^u|\Lambda|\gamma\frac{\psi^4}{4r}\big(\chi_{u_0,r_0}(s)\big)ds\leq\frac{4|\Lambda|}{\delta^2}\iint_{\mathcal{D}_{u,r_0}}r\big(\partial_r\psi\big)^2\ d\mathrm{Vol}\leq\frac{4|\Lambda|}{\delta^2}r_*^2\iint_{\mathcal{D}_{u,r_0}}\frac{\big(\partial_r\psi\big)^2}{r}\ d\mathrm{Vol},\]
which can again be absorbed on the LHS of the energy estimate. 

We bounded all the error terms in \eqref{energy estimate preliminary mass gap in prop}, so using \eqref{assumption on initial data mass gap} for the initial data, we proved that for any $\chi_{u_0,r_0}(u)\in\mathcal{D}_{u_1,r_1}$ we have the estimate:
\begin{equation}\label{energy estimate preliminary mass gap in prop 2}
    \int_0^{r(\chi_{u_0,r_0}(u))}\big(\partial_r\psi\big)^2(u,r)dr+\iint_{\mathcal{D}_{u,r_0}}\bigg(\frac{\big(\partial_r\psi\big)^2}{r}+\frac{\psi^2}{r^3}\bigg)\ d\mathrm{Vol}\leq\frac{A^2}{10}.
\end{equation}
Using the first term in \eqref{energy estimate preliminary mass gap in prop 2}, we can  improve the bootstrap assumption \eqref{bootstrap assumption preliminary mass gap} for all $(u,r)\in\mathcal{D}_{u_1,r_1}$:
\[\big|r\partial_r\phi\big|(u,r)\leq\frac{1}{r^{1/2}}\int_0^r\big|\partial_r\psi\big|(u,\Tilde{r})d\Tilde{r}\leq\bigg(\int_0^r\big|\partial_r\psi\big|^2(u,\Tilde{r})d\Tilde{r}\bigg)^{1/2}\leq\frac{A}{2}.\]
Thus, we conclude that \eqref{bootstrap assumption preliminary mass gap} and \eqref{energy estimate preliminary mass gap in prop 2} hold for all $0<r_1\leq r_*$ and $u_0<u_1\leq u_{\Gamma}(u_0,r_1),$ completing the proof of \eqref{rdrphi bound}-\eqref{preliminary energy estimate psi}. 
\end{proof}

\subsection{Improved estimates for the scalar field}\label{improved scalar field section}
The goal of this section is to improve the previous preliminary estimates. The main point is that we had additional room for improvement in the proof of \Cref{preliminary scalar field proposition}, which we exploit in \Cref{improved scalar field proposition}. The improved estimates will play a key role in \Cref{proof of mass gap section} and allow us to prove uniform bounds for the local well-posedness norm.

We fix $\rho=1/100.$ Since $k\geq2,$ \eqref{regularity conditions scalar field L2}-\eqref{regularity conditions pointwise} imply that there exists a constant $\Tilde{A}>0$, such that:
\begin{equation}\label{assumption on initial data mass gap improved}
    \int_0^{r_*}\Tilde{r}^{-\rho}\big(\partial_r\psi\big)^2(u_0,\Tilde{r})d\Tilde{r}\leq\frac{\Tilde{A}^2}{1000}.
\end{equation}

\begin{proposition}\label{improved scalar field proposition}
     We have the following estimates for the scalar field:
    \begin{equation}\label{improved rdrphi bound}
        \sup_{(u,r)\in\mathcal{D}}\big|r^{1-\frac{\rho}{2}}\partial_r\phi(u,r)\big|\leq A+\Tilde{A},
    \end{equation}
    \begin{equation}\label{improved energy estimate psi}
        \sup_{u\in[u_0,u_*)}\int_0^{r(\chi_{u_0,r_*}(u))}r^{-\rho}(\partial_r\psi\big)^2(u,r)dr+\iint_{\mathcal{D}}\bigg(\frac{\big(\partial_r\psi\big)^2}{r^{1+\rho}}+\frac{\psi^2}{r^{3+\rho}}\bigg)\ d\mathrm{Vol}\leq A^2+\Tilde{A}^2.
    \end{equation}
\end{proposition}
\begin{proof}
    The proof mirrors that of \Cref{preliminary scalar field proposition}. We multiply \eqref{identity for psi} with $r^{-\rho}$ to get:
    \[n\Big(r^{-\rho}\big(\partial_r\psi\big)^2\Big)+\frac{1-\rho/2}{r^{1+\rho}}\overline{\gamma}\big(\partial_r\psi\big)^2-\frac{3}{8r^{2+\rho}}\overline{\gamma}\partial_r\psi^2-2r^{-\rho}\partial_r\overline{\gamma}\big(\partial_r\psi\big)^2+r^{-\rho}\psi\partial_r\psi\cdot\bigg(\frac{1}{r}\partial_r\overline{\gamma}-\partial_r^2\overline{\gamma}\bigg)=0.\]
    As in the proof of \Cref{preliminary scalar field proposition}, this implies the energy identity for any $\chi_{u_0,r_0}(u)\in\mathcal{D}:$
    \begin{align}\label{energy identity improved}
        &\int_0^{r(\chi_{u_0,r_0}(u))}r^{-\rho}\big(\partial_r\psi\big)^2(u,r)dr+\iint_{\mathcal{D}_{u,r_0}}\bigg(\frac{1-\rho/2}{r^{1+\rho}}\overline{\gamma}\big(\partial_r\psi\big)^2-\frac{3\overline{\gamma}}{8r^{2+\rho}}\partial_r\psi^2\bigg)\ d\mathrm{Vol}\\
        &=\int_0^{r_0}\Tilde{r}^{-\rho}\big(\partial_r\psi\big)^2(u_0,\Tilde{r})d\Tilde{r}+\iint_{\mathcal{D}_{u,r_0}}\bigg(\frac{3}{2}r^{-\rho}\partial_r\overline{\gamma}\big(\partial_r\psi\big)^2+r^{-\rho}\psi\partial_r\psi\bigg(\partial_r^2\overline{\gamma}-\frac{1}{r}\partial_r\overline{\gamma}\bigg)\bigg)\ d\mathrm{Vol}.\notag
    \end{align}
    As before, we first deal with the following error term in \eqref{energy identity improved} by integration by parts:
    \[\iint_{\mathcal{D}_{u,r_0}}\frac{3\overline{\gamma}}{8r^{2+\rho}}\partial_r\psi^2\ d\mathrm{Vol}=\int_{u_0}^u\frac{3\overline{\gamma}}{8r^{2+\rho}}\psi^2\big(\chi_{u_0,r_0}(s)\big)ds+\iint_{\mathcal{D}_{u,r_0}}\bigg(\frac{3(2+\rho)\overline{\gamma}}{8r^{3+\rho}}-\frac{3\partial_r\overline{\gamma}}{8r^{2+\rho}}\bigg)\psi^2\ d\mathrm{Vol}.\]
    Additionally, we have the following Hardy inequality:
    \begin{equation}\label{Hardy mass gap improved}
        \frac{\overline{\gamma}\psi^2}{r^{2+\rho}}(s,r)+\int_0^r\frac{(1+\rho)\overline{\gamma}}{\Tilde{r}^{3+\rho}}\psi^2(s,\Tilde{r})d\Tilde{r}\leq\int_0^r\frac{\overline{\gamma}}{\Tilde{r}^{1+\rho}}\big(\partial_r\psi\big)^2(s,\Tilde{r})d\Tilde{r}+\int_0^r\frac{\psi^2\partial_r\overline{\gamma}}{\Tilde{r}^{2+\rho}}(s,\Tilde{r})d\Tilde{r}.
    \end{equation}
    Using \eqref{Hardy mass gap improved}, we obtain the following bound for the error terms:
    \begin{equation}\label{error term improved bound}
        \iint_{\mathcal{D}_{u,r_0}}\frac{3\overline{\gamma}}{8r^{2+\rho}}\partial_r\psi^2\ d\mathrm{Vol}+\frac{1}{8}\iint_{\mathcal{D}_{u,r_0}}\frac{\overline{\gamma}}{r^{3+\rho}}\psi^2\ d\mathrm{Vol}\leq\iint_{\mathcal{D}_{u,r_0}}\bigg(\frac{(7+3\rho)\overline{\gamma}}{8r^{1+\rho}}\big(\partial_r\psi\big)^2+\frac{\psi^2\partial_r\overline{\gamma}}{r^{2+\rho}}\bigg)\ d\mathrm{Vol}.
    \end{equation}
    Next, we consider the following error term in \eqref{energy identity improved}:
    \[\iint_{\mathcal{D}_{u,r_0}}r^{-\rho}\psi\partial_r\psi\cdot\partial_r^2\overline{\gamma}\ d\mathrm{Vol}=\iint_{\mathcal{D}_{u,r_0}}|\Lambda|\gamma\bigg(r^{-\rho}\psi\partial_r\psi+\frac{\partial_r\psi^4}{4r^{1+\rho}}\bigg)\ d\mathrm{Vol}.\]
    We integrate by parts the second term to obtain:
    \begin{align}
    \iint_{\mathcal{D}_{u,r_0}}|\Lambda|\gamma\frac{\partial_r\psi^4}{4r^{1+\rho}}\ d\mathrm{Vol}&=\iint_{\mathcal{D}_{u,r_0}}|\Lambda|\psi^4\bigg(\frac{(1+\rho)\gamma}{4r^{2+\rho}}-\frac{\partial_r\gamma}{4r^{1+\rho}}\bigg)\ d\mathrm{Vol}+\int_{u_0}^u|\Lambda|\gamma\frac{\psi^4}{4r^{1+\rho}}\big(\chi_{u_0,r_0}(s)\big)ds\label{second error term improved bound}\\
    &\leq\iint_{\mathcal{D}_{u,r_0}}|\Lambda|\psi^4\frac{(1+\rho)\gamma}{4r^{2+\rho}}\ d\mathrm{Vol}+\int_{u_0}^u|\Lambda|\gamma\frac{\psi^4}{4r^{1+\rho}}\big(\chi_{u_0,r_0}(s)\big)ds.\notag
\end{align}
Using the error term estimates \eqref{error term improved bound}-\eqref{second error term improved bound} and the energy identity \eqref{energy identity improved}, we get for any $\chi_{u_0,r_0}(u)\in\mathcal{D}$:
\begin{align*}
    \int_0^{r(\chi_{u_0,r_0}(u))}&r^{-\rho}\big(\partial_r\psi\big)^2(u,r)dr+\iint_{\mathcal{D}_{u,r_0}}\bigg(\frac{1}{r^{1+\rho}}\big(\partial_r\psi\big)^2+\frac{1}{r^{3+\rho}}\psi^2\bigg)\ d\mathrm{Vol}\\
    \leq32&\int_0^{r_0}\Tilde{r}^{-\rho}\big(\partial_r\psi\big)^2(u_0,\Tilde{r})d\Tilde{r}+32\iint_{\mathcal{D}_{u,r_0}}\bigg(r^{-\rho}\partial_r\overline{\gamma}\big(\partial_r\psi\big)^2+\frac{\psi^2\partial_r\overline{\gamma}}{r^{2+\rho}}+r^{-\rho}|\Lambda|\gamma\cdot|\psi\partial_r\psi|\bigg)\ d\mathrm{Vol}\\
    &+32\iint_{\mathcal{D}_{u,r_0}}|\Lambda|\psi^4\frac{\gamma}{4r^{2+\rho}}\ d\mathrm{Vol}+32\int_{u_0}^u|\Lambda|\gamma\frac{\psi^4}{4r^{1+\rho}}\big(\chi_{u_0,r_0}(s)\big)ds.
\end{align*}
The first term on the RHS is bounded by $\Tilde{A}^2/10$ due to \eqref{assumption on initial data mass gap improved}. We notice that the remaining error terms are bounded by $A^2/5.$ This follows since in the proof of \Cref{preliminary scalar field proposition} we had additional room left in the powers of $r$ used. Thus, we conclude that \eqref{improved rdrphi bound}-\eqref{improved energy estimate psi} hold.
\end{proof}

\subsection{Proof of \Cref{mass gap theorem}}\label{proof of mass gap section}
\begin{proof}[Proof of \Cref{mass gap theorem}]
    To contradict \Cref{assumption mass gap} and complete the proof of \Cref{mass gap theorem}, it suffices to show that for $I(u)=\{u\}\times\big[0,r(\chi_{u_0,r_*}(u))\big],$ the functions $\phi\in H^{k+1}\big(I(u)\big),\ldots,\partial_u^{k}\phi\in H^1\big(I(u)\big)$ satisfy uniform bounds for $u\in[u_0,u_*).$ Having established the improved estimates for the scalar field in \Cref{improved scalar field section}, it will be straightforward to propagate the desired Sobolev norms from initial data. We notice that the uniform $H^1$ bound for $\phi$ follows from \eqref{n m Bondi} and \eqref{preliminary bound for phi from m}. The only nontrivial part is to establish uniform bounds for $\phi\in H^{2}\big(I(u)\big)$ and $\partial_u\phi\in H^1\big(I(u)\big)$ in \eqref{bound proved in mass gap thm pf}. At higher order, the equations become linear in the top order quantities, so the remaining bounds follow by a propagation of regularity argument, which we do not carry out here for the sake of brevity. We also note that with respect to Bondi coordinates, the Kodama vector field takes the form $\mathcal{T}=\gamma^{-1}\partial_u.$ Thus, one can further show that $\phi\in H^{k+1}\big(I(u)\big),\ldots,\mathcal{T}^{k}\phi\in H^1\big(I(u)\big)$ satisfy uniform bounds for $u\in[u_0,u_*).$ 
    %Note that using k=1 bounds on I'(u):=\{u\}\times\big[0,r(\chi_{u_0,r_*+\epsilon}(u))\big], we get that partial_r\gamma is unif bounded in \mathcal{D}, so get same statement for \mathcal{T} instead of \partial_ur.

    We further shrink the domain $\mathcal{D}$ by increasing $u_0,$ so that in addition to \eqref{bound on r*}, $r_*$ also satisfies:
    \begin{equation}\label{extra bound on r*}
    r_*\leq\frac{\delta^2\rho}{1000(1+|\Lambda|)(A+\Tilde{A}+1)^2(1+|u_*-u_0|)}.
    \end{equation}
    We point out that the bounds in Sections~\ref{bounds on geometry section}-\ref{improved scalar field section} still hold in the new domain $\mathcal{D}$ (except for \eqref{assumption on initial data mass gap} and \eqref{assumption on initial data mass gap improved}, which are replaced by \eqref{rdrphi bound}-\eqref{preliminary energy estimate psi} and \eqref{improved rdrphi bound}-\eqref{improved energy estimate psi}).

    For the rest of the proof, we establish the following bound for all $u\in[u_0,u_*):$
    \begin{equation}\label{bound proved in mass gap thm pf}
    \int_0^{r(\chi_{u_0,r_*}(u))}r\Big[\big(\partial_r\partial_u\phi\big)^2+\big(\partial_r^2\phi\big)^2+\big(\partial_u\phi\big)^2+\big(\partial_r\phi\big)^2\Big](u,r)dr\leq\frac{10}{\delta}+20\int_0^{r_*}\Tilde{r}\big(\partial_r\partial_u\phi\big)^2(u_0,\Tilde{r})d\Tilde{r}.
    \end{equation}

    We first commute the wave equation \eqref{wave eq Bondi} with $\partial_u:$
    \[2\partial_u\partial_r\partial_u\phi+\frac{1}{r}\partial_u\partial_u\phi-\overline{\gamma}\partial_r^2\partial_u\phi-\frac{1}{r}\overline{\gamma}\partial_r\partial_u\phi-\partial_r\overline{\gamma}\partial_r\partial_u\phi-\partial_u\overline{\gamma}\partial_r^2\phi-\frac{1}{r}\partial_u\overline{\gamma}\partial_r\phi-\partial_u\partial_r\overline{\gamma}\partial_r\phi=0.\]
    Multiplying this equation by $2n\partial_u\phi+\partial_r\partial_u\phi,$ we derive the energy identity:
    %We rewrite the $\partial_u$ derivatives using $n$ as follows:
    %\[2\partial_rn\partial_u\phi+\frac{1}{r}n\partial_u\phi-\frac{1}{2r}\overline{\gamma}\partial_r\partial_u\phi-\partial_u\overline{\gamma}\partial_r^2\phi-\frac{1}{r}\partial_u\overline{\gamma}\partial_r\phi-\partial_u\partial_r\overline{\gamma}\partial_r\phi=0\]
    %\[2n\partial_r\partial_u\phi+\frac{1}{r}n\partial_u\phi-\partial_r\overline{\gamma}\partial_r\partial_u\phi-\frac{1}{2r}\overline{\gamma}\partial_r\partial_u\phi-\partial_u\overline{\gamma}\partial_r^2\phi-\frac{1}{r}\partial_u\overline{\gamma}\partial_r\phi-\partial_u\partial_r\overline{\gamma}\partial_r\phi=0\]
    %We combine the two to obtain the following identities:
    %\[\partial_r\big(2r(n\partial_u\phi)^2\big)-\overline{\gamma}\partial_r\partial_u\phi\cdot n\partial_u\phi-2r\partial_u\overline{\gamma}\partial_r^2\phi\cdot n\partial_u\phi-2\partial_u\overline{\gamma}\partial_r\phi\cdot n\partial_u\phi-2r\partial_u\partial_r\overline{\gamma}\partial_r\phi\cdot n\partial_u\phi=0\]
    %\[n\big(r(\partial_r\partial_u\phi)^2\big)+\partial_r\partial_u\phi\cdot n\partial_u\phi-r\partial_r\overline{\gamma}\big(\partial_r\partial_u\phi\big)^2-r\partial_u\overline{\gamma}\partial_r^2\phi\cdot\partial_r\partial_u\phi-\partial_u\overline{\gamma}\partial_r\phi\cdot\partial_r\partial_u\phi-r\partial_u\partial_r\overline{\gamma}\partial_r\phi\cdot\partial_r\partial_u\phi=0\]
    \begin{align}
        \partial_r\big(2r(n\partial_u\phi)^2\big)+n\big(r(\partial_r\partial_u\phi)^2\big)=&r\partial_u\overline{\gamma}\partial_r^2\phi\cdot \big[2n\partial_u\phi+\partial_r\partial_u\phi\big]+ r\partial_r\overline{\gamma}\big(\partial_r\partial_u\phi\big)^2\label{energy identity for du phi}\\&+(\overline{\gamma}-1)\partial_r\partial_u\phi\cdot n\partial_u\phi+\partial_r\phi\cdot\big[\partial_u\overline{\gamma}+r\partial_u\partial_r\overline{\gamma}\big] \cdot\big[2n\partial_u\phi+\partial_r\partial_u\phi\big].\notag
    \end{align}
    We remark that the only potentially dangerous term is the first term on the RHS of \eqref{energy identity for du phi}, which is cubic in second derivatives of the scalar field, since $\partial_u\overline{\gamma}$ contains a $\partial_u\partial_r\phi$ term. In our argument below, we use the improved estimates in \Cref{improved scalar field section} to deal with this term. We also notice that when commuting the wave equation \eqref{wave eq Bondi} with more $\partial_u$ derivatives, we only get linear and quadratic top order terms in the corresponding energy identity, which simplifies the argument considerably. 

    We fix any $0<r_1\leq r_*$ and $u_0<u_1\leq u_{\Gamma}(u_0,r_1).$ For any $\chi_{u_0,r_0}(u)\in\mathcal{D}_{u_1,r_1}$, we integrate the energy identity \eqref{energy identity for du phi} in $\mathcal{D}_{u,r_0}$ to obtain the energy estimate:
    \begin{align}\label{energy identity pf of mass gap thm}
        &\int_0^{r(\chi_{u_0,r_0}(u))}r\big(\partial_r\partial_u\phi\big)^2(u,r)dr+\int_{u_0}^{u}2r\big(n\partial_u\phi\big)^2\big(\chi_{u_0,r_0}(s)\big)ds\\
        &=\int_0^{r_0}\Tilde{r}\big(\partial_r\partial_u\phi\big)^2(u_0,\Tilde{r})d\Tilde{r}+\iint_{\mathcal{D}_{u,r_0}}\bigg(r\partial_u\overline{\gamma}\partial_r^2\phi\cdot \big[2n\partial_u\phi+\partial_r\partial_u\phi\big]+\frac{1}{2}r\partial_r\overline{\gamma}\big(\partial_r\partial_u\phi\big)^2\bigg)\ d\mathrm{Vol}\notag\\
        &+\iint_{\mathcal{D}_{u,r_0}}\bigg((\overline{\gamma}-1)\partial_r\partial_u\phi\cdot n\partial_u\phi+\partial_r\phi\cdot\big[\partial_u\overline{\gamma}+r\partial_u\partial_r\overline{\gamma}\big]\cdot\big[2n\partial_u\phi+\partial_r\partial_u\phi\big]\bigg)\ d\mathrm{Vol}.\notag
    \end{align}
    Motivated by \eqref{energy identity pf of mass gap thm}, we define for any $\chi_{u_0,r_0}(u)\in\mathcal{D}_{u_1,r_1}$ the energy:
    \[E^2(u,r_0)=\sup\bigg\{\int_0^{r(\chi_{u_0,r'}(u'))}r\big(\partial_r\partial_u\phi\big)^2(u',r)dr+\int_{u_0}^{u'}2r\big(n\partial_u\phi\big)^2\big(\chi_{u_0,r'}(s)\big)ds:\ \chi_{u_0,r'}(u')\in\mathcal{D}_{u,r_0}\bigg\}.\]
    Using this definition, we have for any $\chi_{u_0,r_0}(u)\in\mathcal{D}_{u_1,r_1}$:
    \begin{equation}\label{bounds for bulk mass gap}
    \iint_{\mathcal{D}_{u,r_0}}r\big(\partial_r\partial_u\phi\big)^2\ d\mathrm{Vol}\leq|u_*-u_0|E^2(u,r_0),\ \iint_{\mathcal{D}_{u,r_0}}r\big(n\partial_u\phi\big)^2\ d\mathrm{Vol}\leq r_*E^2(u,r_0).
    \end{equation}
    Additionally, we have by \eqref{dr gamma Bondi} and \eqref{bounds for dr overline gamma} for any $(s,r)\in\mathcal{D}_{u,r_0}$:
    \begin{align}
    \big|\partial_u\partial_r\overline{\gamma}\big|(s,r)&=|\Lambda|r\big|\partial_u\gamma\big|\leq2|\Lambda|r\gamma\int_0^r\Tilde{r}\big|\partial_r\phi\partial_u\partial_r\phi\big|(s,\Tilde{r})d\Tilde{r}\label{du dr gamma bar}\\
    &\leq\frac{4|\Lambda|r(A+\Tilde{A})}{\delta}\int_0^r\Tilde{r}^{\frac{\rho}{2}}\big|\partial_u\partial_r\phi\big|(s,\Tilde{r})d\Tilde{r}\leq\frac{4|\Lambda|r^{1+\frac{\rho}{2}}(A+\Tilde{A})}{\delta\sqrt{\rho}}E(u,r_0),\notag\\
    \big|\partial_u\overline{\gamma}\big|(s,r)&\leq\frac{4|\Lambda|r^{2+\frac{\rho}{2}}(A+\Tilde{A})}{\delta\sqrt{\rho}}E(u,r_0),\label{du gamma bar}
\end{align}
where we used the fact that $\partial_u\gamma(u,0)=\partial_u\overline{\gamma}(u,0)=0.$

We now bound the error terms on the RHS of the energy estimate \eqref{energy identity pf of mass gap thm} one by one. For the second error term in \eqref{energy identity pf of mass gap thm}, we use \eqref{bounds for dr overline gamma} and \eqref{bounds for bulk mass gap}:
\[\iint_{\mathcal{D}_{u,r_0}}\frac{1}{2}r\partial_r\overline{\gamma}\big(\partial_r\partial_u\phi\big)^2\ d\mathrm{Vol}\leq\frac{|\Lambda|}{\delta}\iint_{\mathcal{D}_{u,r_0}}r^2\big(\partial_r\partial_u\phi\big)^2\ d\mathrm{Vol}\leq\frac{|\Lambda|}{\delta}|u_*-u_0|r_*E^2(u,r_0).\]
For the third error term in \eqref{energy identity pf of mass gap thm}, we use \eqref{bounds for overline gamma} and \eqref{bounds for bulk mass gap}:
\[\iint_{\mathcal{D}_{u,r_0}}(\overline{\gamma}-1)\big|\partial_r\partial_u\phi\cdot n\partial_u\phi\big|\ d\mathrm{Vol}\leq\frac{|\Lambda|}{\delta}\iint_{\mathcal{D}_{u,r_0}}r^2\big|\partial_r\partial_u\phi\cdot n\partial_u\phi\big|\ d\mathrm{Vol}\leq\frac{|\Lambda|}{\delta}\sqrt{|u_*-u_0|}r_*^{3/2}E^2(u,r_0).\]
For the fourth error term in \eqref{energy identity pf of mass gap thm}, we use \eqref{improved rdrphi bound} and \eqref{bounds for bulk mass gap}-\eqref{du gamma bar}:
\begin{align*}
    \iint_{\mathcal{D}_{u,r_0}}&\big|\partial_r\phi\big|\cdot\Big[\big|\partial_u\overline{\gamma}\big|+r\big|\partial_u\partial_r\overline{\gamma}\big|\Big]\cdot\Big[2\big|n\partial_u\phi\big|+\big|\partial_r\partial_u\phi\big|\Big]\ d\mathrm{Vol}\\
    &\leq\frac{8|\Lambda|(A+\Tilde{A})^2}{\delta\sqrt{\rho}}E(u,r_0)\iint_{\mathcal{D}_{u,r_0}}r^{1+\rho}\Big(2\big|n\partial_u\phi\big|+\big|\partial_r\partial_u\phi\big|\Big)\ d\mathrm{Vol}\\
    &\leq\frac{32|\Lambda|(A+\Tilde{A})^2}{\delta\sqrt{\rho}}\big(r_*+|u_*-u_0|\big)r_*^{1+\rho}E^2(u,r_0).
\end{align*}
For the first error term in \eqref{energy identity pf of mass gap thm}, we use the fact that $r|\partial_r^2\phi|\leq r^{-1/2}|\partial_r\psi|+2|\partial_r\phi|.$ The second term has just been dealt with above. For the first term we have by \eqref{du gamma bar} and \eqref{improved energy estimate psi}:
\begin{align*}
    &\iint_{\mathcal{D}_{u,r_0}}r^{-1/2}\big|\partial_u\overline{\gamma}\big|\cdot\big|\partial_r\psi\big|\cdot \Big[2\big|n\partial_u\phi\big|+\big|\partial_r\partial_u\phi\big|\Big]\ d\mathrm{Vol}\\
    &\leq\frac{4|\Lambda|(A+\Tilde{A})}{\delta\sqrt{\rho}}r_*^{\rho/2}E(u,r_0)\iint_{\mathcal{D}_{u,r_0}}r^{3/2}\big|\partial_r\psi\big|\Big[2\big|n\partial_u\phi\big|+\big|\partial_r\partial_u\phi\big|\Big]\ d\mathrm{Vol}\\
    &\leq\frac{16|\Lambda|(A+\Tilde{A})^2}{\delta\sqrt{\rho}}\big(r_*+|u_*-u_0|\big)r_*^{1+\rho}E^2(u,r_0).
\end{align*}

Using our error term estimates in \eqref{energy identity pf of mass gap thm}, we proved that for any $\chi_{u_0,r_0}(u)\in\mathcal{D}_{u_1,r_1}$:
\[E^2(u,r_0)\leq\int_0^{r_0}\Tilde{r}\big(\partial_r\partial_u\phi\big)^2(u_0,\Tilde{r})d\Tilde{r}+\frac{100|\Lambda|(A+\Tilde{A}+1)^2}{\delta\sqrt{\rho}}\big(|u_*-u_0|+1\big)r_*E^2(u,r_0).\]
In view of \eqref{extra bound on r*}, we can absorb the error term on the LHS, so we get that for any $u\in[u_0,u_*):$
\begin{equation}\label{H1 dot norm for Tphi mass gap}
    \int_0^{r(\chi_{u_0,r_*}(u))}r\big(\partial_r\partial_u\phi\big)^2(u,r)dr\leq2\int_0^{r_*}\Tilde{r}\big(\partial_r\partial_u\phi\big)^2(u_0,\Tilde{r})d\Tilde{r}.
\end{equation}

    In order to prove the uniform $H^2$ bound for $\phi$, we use \eqref{wave eq Bondi}, \eqref{bounds for overline gamma}, and \eqref{bounds for dr overline gamma} to get:
    \begin{equation}\label{pointwise bound for dr2 phi}
        \big|\partial_r^2\phi\big|\leq2\big|\partial_u\partial_r\phi\big|+\frac{\big|\partial_u\phi-\overline{\gamma}\partial_r\phi\big|}{r}+\frac{2|\Lambda|}{\delta}r\big|\partial_r\phi\big|.
    \end{equation}
    We define $\xi=\partial_u\phi-\overline{\gamma}\partial_r\phi.$ The boundary conditions at $\Gamma$ in \Cref{asympt ads data def} imply that $\xi(u,0)=0.$ Moreover, we compute that $\partial_r(r\xi)=-r\partial_u\partial_r\phi$. Using the Hardy inequality \eqref{Hardy mass gap}, we get that for any $(u,r)\in\mathcal{D}$:
    \begin{equation}\label{estimate for xi}
        \int_0^r\frac{\xi^2(u,\Tilde{r})}{\Tilde{r}}d\Tilde{r}=\int_0^r\frac{(r\xi)^2(u,\Tilde{r})}{\Tilde{r}^3}d\Tilde{r}\leq\int_0^r\frac{\big(\partial_r(r\xi)\big)^2(u,\Tilde{r})}{\Tilde{r}}d\Tilde{r}=\int_0^r\Tilde{r}\big(\partial_u\partial_r\phi\big)^2(u,\Tilde{r})d\Tilde{r}.
    \end{equation}
    Using \eqref{bound on r*}, \eqref{preliminary bound for phi from m}, and \eqref{H1 dot norm for Tphi mass gap}-\eqref{estimate for xi}, we get for any $u\in[u_0,u_*):$
\begin{equation}\label{H2 dot norm for phi mass gap}
    \int_0^{r(\chi_{u_0,r_*}(u))}r\big(\partial_r^2\phi\big)^2(u,r)dr\leq1+10\int_0^{r_*}\Tilde{r}\big(\partial_r\partial_u\phi\big)^2(u_0,\Tilde{r})d\Tilde{r}.
\end{equation}
Finally, using \eqref{preliminary bound for phi from m}, \eqref{H1 dot norm for Tphi mass gap}, and \eqref{estimate for xi}, we get for any $u\in[u_0,u_*):$
\begin{equation}\label{L2 norm for phi mass gap}
    \int_0^{r(\chi_{u_0,r_*}(u))}r\big(\partial_u\phi\big)^2(u,r)dr\leq\frac{4}{\delta}+\int_0^{r_*}\Tilde{r}\big(\partial_r\partial_u\phi\big)^2(u_0,\Tilde{r})d\Tilde{r}.
\end{equation}

According to \eqref{preliminary bound for phi from m}, \eqref{H1 dot norm for Tphi mass gap}, \eqref{H2 dot norm for phi mass gap}, and \eqref{L2 norm for phi mass gap}, we completed the proof of \eqref{bound proved in mass gap thm pf}. Thus, we established uniform bounds for $\phi\in H^{2}\big(I(u)\big)$ and $\partial_u\phi\in H^1\big(I(u)\big)$. As explained previously, our argument also works for higher order derivatives, where the structure of the error terms simplifies significantly. Therefore, one can also prove uniform bounds for $\phi\in H^{k+1}\big(I(u)\big),\ldots,\partial_u^{k}\phi\in H^1\big(I(u)\big),$ which contradicts \Cref{assumption mass gap}, completing the proof of \Cref{mass gap theorem}.
\end{proof}

\section{Trapped surface formation criterion}\label{trapped surface formation section}

We prove a criterion for trapped surface formation following the strategy of Christodoulou in \cite{ChrBH}. We recall the statement of \Cref{trapped surface formation theorem}:
\trapped*
\begin{proof}
    Throughout the proof we use the indices $1$ and $2$ for quantities evaluated on $C_1$ respectively $C_2.$ We first claim that there exists $u_0<u_1<0$ such that $r_2(u_1)=\frac{5\delta_0/4}{1+\delta_0}r_{2}(u_0).$ Supposing this is not true, we have by \Cref{global structure section} that $\overline{C_2}\cap\mathcal{B}_0\neq\emptyset$ (since $r=0$ on $\mathcal{B}$ and $C_1\cap\Gamma=\emptyset$). By equation \eqref{wave equation r}, we have that for any $u_0<u<0:$
    \[\frac{5\delta_0/4}{1+\delta_0}r_{2}(u_0)<r_2(u)\leq r_1(u)+\delta_0r_1(u_0).\]
    Taking $u\rightarrow0^-,$ we obtain the contradiction $r_{2}(u_0)<r_1(u_0).$

    We fix $u_0<u_1<0$ such that $r_2(u_1)=\frac{5\delta_0/4}{1+\delta_0}r_{2}(u_0)$. We assume the conclusion is false, so there exist no marginally trapped surfaces on $C_2$ for $u\in[u_0,u_1].$ At first, we also assume that $m_2(u)-m_1(u)>0$ for all $u\in[u_0,u_1]$.

    A key role in the proof is played by the quantity $\eta:=|\Lambda|^{-1}r_2^{-2}(m_2-m_1),$ for which we aim to derive a differential inequality. Using equation \eqref{du m}, we compute for any $u\in[u_0,u_1]$:
    \[\partial_u\big(m_2-m_1\big)=-\frac{4}{\gamma_2}\cdot\bigg(r_2\big(\partial_u\phi_2\big)^2-\frac{\gamma_2}{\gamma_1}r_1\big(\partial_u\phi_1\big)^2\bigg)=-\frac{4}{\gamma_2}\cdot\bigg(\zeta^2+2\zeta\sqrt{r_1}\partial_u\phi_1-\bigg(\frac{\gamma_2}{\gamma_1}-1\bigg)r_1\big(\partial_u\phi_1\big)^2\bigg),\]
    where we denoted $\zeta=\sqrt{r_2}\partial_u\phi_2-\sqrt{r_1}\partial_u\phi_1.$ By the assumption $m_2(u)-m_1(u)>0$, equations \eqref{dv m} and \eqref{dv log gamma} imply $\gamma_2>\gamma_1.$ We further get:
    \begin{equation}\label{prelim lower bound du m}
        \partial_u\big(m_2-m_1\big)\geq-\frac{4}{\gamma_2}\cdot\bigg(1+\frac{1}{\gamma_2/\gamma_1-1}\bigg)\zeta^2\geq-\frac{4}{\gamma_2}\cdot\bigg(1+\frac{1}{\log(\gamma_2/\gamma_1)}\bigg)\zeta^2.
    \end{equation}

    Using $m\geq1$ and the equations \eqref{dv log gamma}, \eqref{dv m}, \eqref{wave equation phi}, we have the bounds for any $u\in[u_0,u_1]$:
    \begin{align}
        \log\frac{\gamma_2}{\gamma_1}&=\int_{v_1}^{v_2}r\frac{(\partial_v\phi)^2}{\partial_vr}(u,v')dv'=\int_{v_1}^{v_2}\frac{\partial_vm}{-4\Omega^{-2}\partial_ur\partial_vr}(u,v')dv'\geq\frac{m_2-m_1}{|\Lambda|r_2^2}=\eta>0,\label{prelim trap bound 1}\\
        \zeta^2&\leq\bigg(\int_{v_1}^{v_2}\frac{|\partial_ur|\cdot|\partial_v\phi|}{2\sqrt{r}}\bigg)^2\leq\big(m_2-m_1\big)\int_{v_1}^{v_2}\frac{\Omega^{2}|\partial_ur|}{16r^2}\leq\frac{(m_2-m_1)\cdot \gamma_2\cdot|\partial_ur_2|}{16r_2}\bigg(\frac{r_2}{r_1}-1\bigg).\label{prelim trap bound 2}
    \end{align}
    We use \eqref{prelim trap bound 1}-\eqref{prelim trap bound 2} in \eqref{prelim lower bound du m}:
    \[\partial_u\big(m_2-m_1\big)\geq-\frac{(m_2-m_1)\cdot|\partial_ur_2|}{4r_2}\bigg(1+\frac{1}{\eta}\bigg)\cdot\bigg(\frac{r_2}{r_1}-1\bigg).\]
    Using the fact that $\eta=|\Lambda|^{-1}r_2^{-2}(m_2-m_1),$ we can rewrite this bound for any $u\in[u_0,u_1]$:
    \begin{equation}\label{prelim trap bound eta u}
        \partial_u\eta\geq\eta\cdot\frac{|\partial_ur_2|}{4r_2}\cdot\bigg[8-\bigg(1+\frac{1}{\eta}\bigg)\cdot\bigg(\frac{r_2}{r_1}-1\bigg)\bigg].
    \end{equation}
    We introduce the variable $x=r_2(u)/r_{2}(u_0)$ along $C_2,$ and we notice that \eqref{wave equation r} implies:
    \[\frac{r_2}{r_1}-1\leq\frac{\delta_0}{x(1+\delta_0)-\delta_0}.\]
    We rewrite inequality \eqref{prelim trap bound eta u} for any $x\in\big[\frac{5\delta_0/4}{1+\delta_0},1\big]$:
    \[-x\partial_x\eta\geq\eta\bigg[2-\frac{1}{4}\frac{\delta_0}{x(1+\delta_0)-\delta_0}\bigg(1+\frac{1}{\eta}\bigg)\bigg]=\mathbf{g}\eta-\mathbf{f},\]
    \[\mathbf{f}=\frac{1}{4}\frac{\delta_0}{x(1+\delta_0)-\delta_0},\ \mathbf{g}=2-\mathbf{f}.\]
    We integrate to get that for any $x\in\big[\frac{5\delta_0/4}{1+\delta_0},1\big]$:
    \[\eta_0\leq\mathbf{F}(x)+\eta(x)e^{-\mathbf{G}(x)},\]
    \[\mathbf{F}(x)=\int_x^1e^{-\mathbf{G}(x')}f(x')\frac{dx'}{x'},\ \mathbf{G}(x)=\int_x^1g(x')\frac{dx'}{x'}.\]
    By assumption, the circle $(u_1,v_2)$ is not trapped, so for $x_1:=r_2(u_1)/r_{2}(u_0):$
    \[\eta(x_1)=\frac{m_2(u_1)-m_1(u_1)}{|\Lambda|r_2^2(u_1)}\leq\frac{m_2(u_1)-1}{|\Lambda|r_2^2(u_1)}\leq1,\]
    which implies that $\eta_0\leq\mathbf{F}(x_1)+e^{-\mathbf{G}(x_1)}.$ We recall that we chose $u_1$ such that $x_1=\frac{5\delta_0/4}{1+\delta_0}.$ While not essential to the proof, we remark that the minimum of $\mathbf{F}(x)+e^{-\mathbf{G}(x)}$ for $x\in\big[\frac{5\delta_0/4}{1+\delta_0},1\big]$ is attained exactly at $x_1,$ where we have that $\mathbf{f}(x_1)=\mathbf{g}(x_1)=1.$ Finally, we compute that:
    \[e^{-\mathbf{G}(x)}=x^2\bigg(\frac{x}{x(1+\delta_0)-\delta_0}\bigg)^{1/4},\ \mathbf{F}(x)=\frac{\delta_0}{4}\int_x^1\bigg(\frac{x'}{x'(1+\delta_0)-\delta_0}\bigg)^{5/4}dx',\]
    \[\eta_0\leq\mathbf{F}(x_1)+e^{-\mathbf{G}(x_1)}=\frac{\delta_0}{4}+O\big(\delta_0^2|\log\delta_0|\big).\]
    However, we have that $\eta_0>\delta_0,$ which leads to a contradiction for $\delta_0>0$ small enough.

    We recall that in the above proof we assumed that for all $u\in[u_0,u_1]$, we have $m_2(u)-m_1(u)>0$. To deal with the remaining case, we denote $u_*=\inf_{[u_0,u_1]}\{u:\ m_2(u)=m_1(u)\},$ and we have that $u_*\in(u_0,u_1].$ We can repeat the above proof on the interval $[u_0,u_*)$ instead of $[u_0,u_1).$ We denote $x_*:=r_2(u_*)/r_{2}(u_0),$ and we get that $\eta_0\leq\mathbf{F}(x_*)+\eta(x_*)e^{-\mathbf{G}(x_*)}=\mathbf{F}(x_*)$, since $m_2(u_*)=m_1(u_*)$. However, we have $\mathbf{F}(x)<\delta_0/2$ for all $x\in[x_1,1],$ which again contradicts $\eta_0>\delta_0.$
\end{proof}

\section{Locally naked singularity spacetimes}\label{naked singularity section}
In this section, we provide a detailed description of locally naked singularity spacetimes in a neighborhood of the first singularity at the center. Throughout this section, we fix $k\geq 2$ and we consider $\big(\mathcal{M},g,\phi\big)$ to be the maximal globally hyperbolic development of a $C^k$ characteristic data set $\mathfrak{D}$ on the outgoing null cone $C_0^+.$ We assume that $b_{\Gamma}$ is a first singularity at the center $\Gamma$ and that the ingoing null cone $C_0^-$ passing through $b_{\Gamma}$ intersects $C_0^+$ (rather than $\mathcal{I}$). We choose double null coordinates $(u,v)$ such that: \[b_{\Gamma}=(0,0),\ C_0^-=\big\{v=0,\ 2r=-u\big\},\ C_0^+=\big\{u=u_0,\ 2r=v-u_0\big\}.\]
Moreover, we assume that $\big(\mathcal{M},g\big)$ is a locally naked singularity spacetime, as defined in \eqref{M locally naked}, so $\mathcal{B}_0\neq\emptyset.$ In particular, according to \Cref{global structure section}, there exists $v_1>0$ such that $[u_0,0)\times[0,v_1]\subset\mathcal{R}$.

As a consequence of \Cref{mass gap theorem} and \Cref{trapped region lemma}, we have that for any $(u,v)\in[u_0,0)\times[0,v_1]:$
\begin{equation}\label{bounds for mass locally naked}
    1\leq m(u,v)\leq1+|\Lambda|r^2(u,v).
\end{equation}
This further implies a lower bound for $\gamma$ on $C_0^-$:
\begin{equation}\label{gamma lower bound on C0-}
    \gamma(u,0)=\frac{-4\partial_ur(u,0)}{|\Lambda|r^2(u,0)+1-m(u,0)}\geq8|\Lambda|^{-1}|u|^{-2}.
\end{equation}
Moreover, we notice that \eqref{Ray u} implies that $\partial_u\Omega^{-2}\geq0,$ so $\Omega^2$ is decreasing along $C_0^-$. 

Next, we introduce the blueshift, which plays an essential role in our analysis:
\begin{definition}\label{blueshift definition}
    We define the blueshift as the function $\mathbf{b}:[u_0,0)\times[0,v_1]\rightarrow[0,\infty):$
    \begin{equation}\label{blueshift definition eq}
        \mathbf{b}(u,v)=\int_{u_0}^u\frac{|\Lambda|}{2}r\gamma(u',v)du'.
    \end{equation}
    We denote $\mathbf{b}(u)=\mathbf{b}(u,0).$ We also define the function $\mathbf{s}:[u_0,0)\rightarrow[0,\infty):$
    \begin{equation}\label{blueshift strength definition}
        \mathbf{s}(u)=\int_{u_0}^u\frac{|\Lambda|}{4}r^2\gamma(u',0)du'.
    \end{equation}
\end{definition}

Using this notation, we can integrate \eqref{wave equation r} to obtain that:
\begin{equation}\label{dv r in terms of blueshift}
        \partial_vr(u,v)=\frac{1}{2}e^{-\mathbf{b}(u,v)}.
\end{equation}
In particular, we have $e^{-\mathbf{b}(u)}\gamma(u,0)=2\partial_vr(u,0)\gamma(u,0)=2\Omega^2(u,0)\leq2\Omega^2(u_0,0).$ We use this bound for the rest of the section without explicit reference.

We prove that all first singularities at the center have unbounded blueshift on $C_0^-$:
\begin{proposition}\label{blueshift lower bound proposition}
    The blueshift $\mathbf{b}:[u_0,0)\rightarrow[0,\infty)$ is unbounded on $C_0^-$ and satisfies the quantitative estimate:
    \begin{equation}\label{blueshift lower bound eq}
        \mathbf{b}(u)\geq2\log\frac{u_0}{u}.
    \end{equation}
    The lower bound \eqref{blueshift lower bound eq} holds more generally for any solution $\big(\mathcal{M},g,\phi\big)$ with a first singularity $b_{\Gamma}\in\overline{\Gamma}\backslash\Gamma$  (which is not necessarily locally naked).
\end{proposition}
\begin{proof}
    The proof is immediate by using \eqref{gamma lower bound on C0-} and \eqref{blueshift definition eq}. We note this only relies on \eqref{bounds for mass locally naked} on $C_0^-,$ which indeed holds for any first singularity at the center, without assuming that the solution is a locally naked singularity spacetime.
\end{proof}

In order to state our next result, we introduce notation similar to that of \Cref{trapped surface formation section}. For any $(u,v)\in[u_0,0)\times[0,v_1],$ we define:
\[\eta(u,v)=\frac{m(u,v)-m(u,0)}{|\Lambda|r^2(u,v)},\ \delta(u,v)=\frac{r(u,v)}{r(u,0)}-1.\]

We make the notation convention for the rest of this section that we write $x\lesssim y$ for any quantities $x,y>0$ if there is a constant $C\big(k,|\Lambda|,|u_0|,\Omega^2(u_0,0)\big)>0$ such that $x\leq Cy.$ 

We prove further a priori bounds for the solution $\big(\mathcal{M},g,\phi\big):$
\begin{proposition}\label{bounds for locally naked singularity proposition}
    For any $\epsilon>0$ sufficiently small depending on $|\Lambda|,|u_0|,\Omega^2(u_0,0)$, and any $v_0\in(0,v_1)$ sufficiently small depending on $\epsilon,|\Lambda|,|u_0|,\Omega^2(u_0,0),\gamma|_{C_0^+\cap\{0\leq v\leq v_1\}},$ we define the region:
    \begin{equation}
        \mathcal{P}=\Big\{(u,v)\in[u_0,0)\times[0,v_0]:\ v^{1/2}\mathbf{s}(u)\leq \epsilon\Big\},
    \end{equation}
    represented schematically in \Cref{fig:locallynakedsing}. Then, for any $(u,v)\in\mathcal{P}$ we have that:
    \begin{equation}\label{eta less than delta}
        \eta(u,v)\leq\delta(u,v),
    \end{equation}
    \begin{equation}
    \partial_vr(u,v)\leq 100e^{-\mathbf{b}(u)},
\end{equation}
\begin{equation}
    \gamma(u,0)\leq \gamma(u,v)\leq 100\gamma(u,0),
\end{equation}
    \begin{equation}\label{background energy bound 1}
    \int_0^v\theta^2\partial_vr(u,v')dv'\leq10^4|\Lambda|\cdot v|u|e^{-\mathbf{b}(u)}\gamma(u,0),
    \end{equation}
    \begin{equation}\label{background energy bound 2}
    \int_u^0|u'|\gamma^{-1}(u',0)\cdot\mathbf{1}_{\mathcal{P}}(u',v)\big|\partial_u\phi\big|^2(u',v)du'\lesssim|u|^2,
    \end{equation}
    where $\mathbf{1}_{\mathcal{P}}$ represents the characteristic function of $\mathcal{P}.$
\end{proposition}
\begin{proof}
    We notice that $\mathbf{s}$ is increasing, so if $(\Tilde{u},\Tilde{v})\in\mathcal{P}$ then $[u_0,\Tilde{u}]\times[0,\Tilde{v}]\subset\mathcal{P}.$ We make the following bootstrap assumptions for any $(u,v)\in[u_0,\Tilde{u}]\times[0,\Tilde{v}]\subset\mathcal{P}$:
\begin{equation}\label{A1}
    \partial_vr(u,v)\leq 100e^{-\mathbf{b}(u)},
\end{equation}
\begin{equation}\label{A2}
    \gamma(u,v)\leq 100\gamma(u,0).
\end{equation}
We first check that \eqref{A1}-\eqref{A2} hold initially with significantly better constants. In view of \eqref{dv r in terms of blueshift} and $2r(u_0,v)=v-u_0,$ we have for $(u,v)\in\big(C_0^-\cup C_0^+\big)\cap\mathcal{P}$ that $\partial_vr(u,v)\leq e^{-\mathbf{b}(u)}.$ By continuity, we have that $\gamma(u_0,v)\leq2\gamma(u_0,0)$ for all $v\in[0,v_0],$ where $v_0>0$ is small enough depending on $\gamma|_{C_0^+\cap\{0\leq v\leq v_1\}}$. 

Next, we have by \eqref{A1}, \eqref{blueshift lower bound eq} for any $(u,v)\in[u_0,\Tilde{u}]\times[0,\Tilde{v}]$:
\begin{equation}\label{bound for delta locally naked}
    \delta(u,v)=\frac{\int_0^v\partial_vr(u,v')dv'}{r(u,0)}\leq\frac{200e^{-\mathbf{b}(u)}v}{|u|}\leq\frac{200|u|v}{|u_0|^2}.
\end{equation}
For $v_0$ sufficiently small in terms of $|u_0|,$ we have for all $v\in(0,v_0]$ that $\delta(u,v)$ is small enough, as in the statement of \Cref{trapped surface formation theorem}. Since $[u_0,0)\times[0,v_0]\subset\mathcal{R},$ \Cref{trapped surface formation theorem} implies that $\eta(u,v)\leq\delta(u,v)$ for all $(u,v)\in[u_0,\Tilde{u}]\times[0,\Tilde{v}]$.

We prove some preliminary estimates for all $(u,v)\in[u_0,\Tilde{u}]\times[0,\Tilde{v}]$. By monotonicity, we have that $r\geq|u|/2,\ \partial_ur\leq-1/2,$ and $\gamma(u,v)\geq\gamma(u,0).$ Also, for $v_0$ sufficiently small in terms of $|u_0|,$ we have that \eqref{A1} and \eqref{blueshift lower bound eq} imply $r\leq|u|.$ Moreover, we have by \eqref{wave equation r}, \eqref{dv r in terms of blueshift}, \eqref{A1}, and \eqref{A2}:
\begin{align}\label{bound for du r locally naked}
    |\partial_ur|&\leq\frac{1}{2}+\int_0^v|\Lambda|\cdot|u|\gamma(u,v')\partial_vr(u,v')dv' \leq\frac{1}{2}+10^4|\Lambda|\cdot v|u|e^{-\mathbf{b}(u)}\gamma(u,0)\\
    &\leq\frac{1}{2}+2\cdot10^4|\Lambda|\cdot v|u|\Omega^2(u,0)\leq\frac{1}{2}+2\cdot10^4|\Lambda|\cdot v|u_0|\Omega^2(u_0,0)<1,\notag
\end{align}
for $v_0$ sufficiently small in terms of $|\Lambda|,|u_0|,\Omega^2(u_0,0).$

We now improve the bootstrap assumptions. Since $\eta(u,v)\leq\delta(u,v),$ we get by \eqref{dv m}, \eqref{dv log gamma}, \eqref{dv r in terms of blueshift}, \eqref{A2}, and \eqref{bound for delta locally naked} that for all $(u,v)\in[u_0,\Tilde{u}]\times[0,\Tilde{v}]$:
\begin{align*}
    \log\frac{\gamma(u,v)}{\gamma(u,0)}&=\int_0^v\frac{r(\partial_v\phi)^2}{\partial_vr}(u,v')dv'=\int_0^v\frac{\gamma\partial_vm}{-4\partial_ur}(u,v')dv'\leq50|\Lambda|\gamma(u,0)|u|^2\eta(u,v)\\
    &\leq50|\Lambda|\gamma(u,0)|u|^2\delta(u,v)\leq10^4|\Lambda|\cdot v|u|e^{-\mathbf{b}(u)}\gamma(u,0)\leq2\cdot10^4|\Lambda|\cdot v|u_0|\Omega^2(u_0,0)<1,
\end{align*}
for $v_0$ sufficiently small in terms of $|\Lambda|,|u_0|,\Omega^2(u_0,0).$ Thus, we improved the bootstrap assumption \eqref{A2}. We notice that in the above we also proved that for all $(u,v)\in[u_0,\Tilde{u}]\times[0,\Tilde{v}]$:
\begin{equation}\label{bound for theta^2 locally naked}
    \int_0^v\theta^2\partial_vr(u,v')dv'=\int_0^v\frac{r(\partial_v\phi)^2}{\partial_vr}(u,v')dv'\leq10^4|\Lambda|\cdot v|u|e^{-\mathbf{b}(u)}\gamma(u,0).
\end{equation}

Using \eqref{dv log gamma}, \eqref{dv r in terms of blueshift}, \eqref{A1}, \eqref{A2}, and \eqref{bound for theta^2 locally naked}, we have the following estimate for the blueshift for all $(u,v)\in[u_0,\Tilde{u}]\times[0,\Tilde{v}]$:
\begin{align*}
    |\mathbf{b}&(u,v)-\mathbf{b}(u)|\lesssim\int_0^v\int_{u_0}^u\partial_v(r\gamma)(u',v')du'dv'\lesssim\int_0^v\int_{u_0}^u\Big(\gamma\partial_vr+|u'|\gamma\theta^2\partial_vr\Big)(u',v')du'dv'\\
    &\lesssim v\int_{u_0}^u\gamma(u',0)e^{-\mathbf{b}(u')}du'+\int_{u_0}^u|u'|\gamma(u',0)\bigg(\int_0^v\theta^2\partial_vr(u',v')dv'\bigg)du'\\
    &\lesssim v|u_0|\Omega^2(u_0,0)+v\int_{u_0}^u|u'|^2\gamma^2(u',0)e^{-\mathbf{b}(u')}du'\lesssim v|u_0|\Omega^2(u_0,0)+v\Omega^2(u_0,0)\cdot\mathbf{s}(u)\lesssim\epsilon,
\end{align*}
for $v_0$ sufficiently small in terms of $\epsilon,|\Lambda|,|u_0|,\Omega^2(u_0,0).$ We use this estimate in \eqref{dv r in terms of blueshift}:
\begin{equation}\label{dv r in terms of b(u)}
    \partial_vr(u,v)=\frac{1}{2}e^{-\mathbf{b}(u,v)}=\frac{1}{2}e^{-\mathbf{b}(u)}\big(1+O(\epsilon)\big)\leq 10e^{-\mathbf{b}(u)},
\end{equation}
provided that $\epsilon>0$ is sufficiently small in terms of $|\Lambda|,|u_0|,\Omega^2(u_0,0).$ This allows us to improve the assumption \eqref{A1} as well.

We improved the bootstrap assumptions \eqref{A1}-\eqref{A2}, so we proved \eqref{eta less than delta}-\eqref{background energy bound 1} for all $(u,v)\in\mathcal{P}.$ As a consequence of \eqref{bounds for mass locally naked}, we have that for any $u_0\leq u_1<u_2<0$ and $v\in[0,v_1]:$
\[0\leq m(u_1,v)-m(u_2,v)\leq|\Lambda|r^2(u_1,v).\]
Using this inequality, together with \eqref{du m} and the bounds established in $\mathcal{P}$ completes the proof of \eqref{background energy bound 2}.
\end{proof}

\begin{remark}
    The function $\mathbf{s}(u)$ can be interpreted as a measurement for the strength of the blueshift, which affects the size of the region $\mathcal{P}.$ In the case when $\mathbf{s}(u)$ is bounded, we can take $\mathcal{P}=[u_0,0)\times[0,v_0]$ and we obtain a priori estimates all the way to $\mathcal{B}_0.$ In particular, \eqref{blueshift rate for CR} implies this is the case for the examples of locally naked singularities constructed in \cite{preparation}.
\end{remark}

\section{Instability to trapped surface formation}\label{instability section}
In this section, we prove that all locally naked singularity spacetimes are unstable to trapped surface formation arbitrarily close to $b_{\Gamma}$. Let $\big(\mathcal{M},g,\phi\big)$ be a locally naked singularity spacetime satisfying the set-up in \Cref{naked singularity section}. In particular, $\big(\mathcal{M},g,\phi\big)$ is the maximal globally hyperbolic development of a $C^k$ characteristic data set $\mathfrak{D}$, for some $k\geq 2$. We also assume that the ingoing null cone $C_0^-$ passing through $b_{\Gamma}$ intersects $C_0^+$, and we choose double null coordinates $(u,v)$ such that: \[b_{\Gamma}=(0,0),\ C_0^-=\big\{v=0,\ 2r=-u\big\},\ C_0^+=\big\{u=u_0,\ 2r=v-u_0\big\}.\]
Moreover, we take $v_1>0$ such that $[u_0,0)\times[0,v_1]\subset\mathcal{R}$, and we fix $v_0,\epsilon>0$ as required in \Cref{naked singularity section}.

We fix any integer $n>k+1$ and $\lambda\in\mathbb{R}\backslash\{0\}$. We define characteristic data on $\big\{u=u_0,\ v\in[u_0,v_1]\big\}:$
\begin{equation}\label{characteristic data perturbation}
    r_{\lambda}=r,\ \phi_{\lambda}=\phi,\ \Omega_{\lambda}=\Omega\ \text{on }\big\{u=u_0,\ v\in[u_0,0]\big\},\ r_{\lambda}=r,\ \theta_{\lambda}=\theta+\lambda v^{n}\ \text{on }\big\{u=u_0,\ v\in[0,v_1]\big\},
\end{equation}
where $\phi_{\lambda}$ and $\Omega_{\lambda}$ are defined on $\big\{u=u_0,\ v\in(0,v_1]\big\}$ according to \eqref{definition of theta} and \eqref{Ray v}. Since $n>k+1,$ the characteristic data $\big(r_{\lambda},\Omega_{\lambda},\phi_{\lambda}\big)$ determines uniquely a $C^k$ characteristic partial data set $\mathfrak{D}_{\lambda}$, as in \Cref{partial data set definition}, which coincides with $\mathfrak{D}$ for $\big\{u=u_0,\ v\in[u_0,0]\big\}.$ 

We denote by $\big(\mathcal{M}_{\lambda},g_{\lambda},\phi_{\lambda}\big)$ the maximal globally hyperbolic development of the $C^k$ characteristic partial data set $\mathfrak{D}_{\lambda}$. By domain of dependence, we have that $r_{\lambda}=r,\ \phi_{\lambda}=\phi,\ \Omega_{\lambda}=\Omega$ for $v\leq 0.$ The main result of this section is the instability of locally naked singularity spacetimes to trapped surface formation:

\begin{proposition}\label{instability to trapped surfaces proposition}
    For any $\lambda\in\mathbb{R}\backslash\{0\},$ the spacetime $\big(\mathcal{M}_{\lambda},g_{\lambda}\big)$ contains a sequence of trapped surfaces converging to $b_{\Gamma}.$
\end{proposition}

The proof of \Cref{instability to trapped surfaces proposition} consists of two parts. First, in \Cref{estimates for perturbed spacetime section} we prove existence of the solution $\big(\mathcal{M}_{\lambda},g_{\lambda},\phi_{\lambda}\big)$ and quantitative bounds in a suitable region $\mathcal{P}_{\lambda}.$ A key aspect is that the estimates proved are consistent with the fact that the initial data perturbation in \eqref{characteristic data perturbation} vanishes to order $n$ at $C_0^-.$ Second, in \Cref{proof of instability thm section} we define a curve $\mathcal{C}_{\lambda}$ which converges to $b_{\Gamma}$ and is contained in $\mathcal{P}_{\lambda}$ sufficiently close to $b_{\Gamma}$. We prove that the blueshift effect causes the amplification of $\theta_{\lambda}$ on $\mathcal{C}_{\lambda}$, with the main term coming from the $\lambda v^n$ perturbation in \eqref{characteristic data perturbation}. We can then apply the formation of trapped surfaces criterion in \Cref{trapped surface formation theorem} for symmetry circles on $\mathcal{C}_{\lambda}$ to complete the proof of \Cref{instability to trapped surfaces proposition}.

\subsection{Estimates for the perturbed spacetime}\label{estimates for perturbed spacetime section}
The main result of this section establishes existence and quantitative estimates in the region $\mathcal{P}_{\lambda}:$
\begin{proposition}\label{bounds for perturbed spacetime proposition}
    For any $\epsilon'\in(0,\epsilon)$ sufficiently small depending on $\lambda,n,|\Lambda|,|u_0|,\Omega^2(u_0,0)$, and any $v_0'\in(0,v_0)$ sufficiently small depending on $\lambda,n,\epsilon',|\Lambda|,|u_0|,\Omega^2(u_0,0),\gamma|_{C_0^+\cap\{0\leq v\leq v_1\}},$ we define the region:
    \begin{equation}\label{definition of P lambda}
    \mathcal{P}_{\lambda}=\bigg\{u\in[u_0,0),\ 0\leq v\leq\min\big(v_0',\lambda^4,|u|^{\frac{1}{n}}\big),\ |\lambda| v^{n+\frac{1}{4}}\leq e^{-\mathbf{b}(u)}|u|^{\frac{1}{2}+\frac{1}{10n}}\sqrt{\gamma(u,0)},\ v^{\frac{1}{10n}}\mathbf{s}(u)\leq\epsilon'\bigg\},
    \end{equation}
    represented schematically in \Cref{fig:instability}. The solution $\big(\mathcal{M}_{\lambda},g_{\lambda},\phi_{\lambda}\big)$ exists and satisfies the quantitative~bounds in the region $\mathcal{P}_{\lambda}$:
    \begin{equation}\label{B1 blueshift instability}
    \partial_vr_{\lambda}(u,v)\leq 10e^{-\mathbf{b}(u)},
\end{equation}
\begin{equation}\label{B2 blueshift instability}
    \big|\gamma_{\lambda}(u,v)-\gamma(u,0)\big|\leq 100e^{-\mathbf{b}(u)}v^{1/4}|u|\gamma^2(u,0),
\end{equation}
\begin{equation}\label{B3 blueshift instability}
    \int_0^v\big|\theta_{\lambda}-\theta\big|^2(u,v')dv'\leq 100\lambda^2v^{2n+1}e^{2\mathbf{b}(u)}.
\end{equation}
\begin{equation}\label{B4 blueshift instability}
    \int_0^v\Big|\big(\theta_{\lambda}-\theta\big)(u,v')-e^{\mathbf{b}_{\lambda}(u,v')}\big(\theta_{\lambda}-\theta\big)(u_0,v')\Big|^2dv'\leq\lambda^2v^{2n+\frac{4}{3}}e^{2\mathbf{b}(u)}.
\end{equation}
\end{proposition}

We prove \Cref{bounds for perturbed spacetime proposition} for the rest of \Cref{estimates for perturbed spacetime section}. We recall that along $C_0^-$ we have that $\mathbf{b}$ and $\mathbf{s}$ are increasing, while $\Omega^2$ is decreasing. Thus, if $(\Tilde{u},\Tilde{v})\in\mathcal{P}_{\lambda}$ then $[u_0,\Tilde{u}]\times[0,\Tilde{v}]\subset\mathcal{P}_{\lambda}.$ We make the bootstrap assumptions in for any $(u,v)\in[u_0,\Tilde{u}]\times[0,\Tilde{v}]\subset\mathcal{P}_{\lambda}$:
\begin{equation}\label{A1 blueshift instability}
    \partial_vr_{\lambda}(u,v)\leq 10e^{-\mathbf{b}(u)},
\end{equation}
\begin{equation}\label{A2 blueshift instability}
    \big|\gamma_{\lambda}(u,v)-\gamma(u,0)\big|\leq 100e^{-\mathbf{b}(u)}v^{1/4}|u|\gamma^2(u,0),
\end{equation}
\begin{equation}\label{A3 blueshift instability}
    \int_0^v\big|\theta_{\lambda}-\theta\big|^2(u,v')dv'\leq 100\lambda^2v^{2n+1}e^{2\mathbf{b}(u)}.
\end{equation}
We first check that \eqref{A1 blueshift instability}-\eqref{A3 blueshift instability} hold initially with significantly better constants. By \eqref{dv r in terms of blueshift} and \eqref{characteristic data perturbation}, we have that $\partial_vr_{\lambda}(u,v)=\partial_vr(u,v)\leq e^{-\mathbf{b}(u)}$ for any $(u,v)\in \big(C_0^-\cup C_0^+\big)\cap\mathcal{P}_{\lambda}.$ Moreover, \eqref{characteristic data perturbation} also implies:
\begin{equation}\label{bound for theta initial data pert}
    \int_0^v\big|\theta_{\lambda}-\theta\big|^2(u_0,v')dv'=\frac{\lambda^2}{2n+1}v^{2n+1}\leq\lambda^2v^{2n+1}e^{2\mathbf{b}(u_0)}.
\end{equation}
Next, we have by \eqref{dv log gamma}, \eqref{characteristic data perturbation}, \eqref{definition of P lambda}, and the proof of \Cref{bounds for locally naked singularity proposition} that for any $(u_0,v)\in C_0^+\cap\mathcal{P}_{\lambda}:$
\begin{align*}
    0\leq\frac{\gamma_{\lambda}(u_0,v)}{\gamma(u_0,0)}-1&\leq2\int_0^v\theta^2_{\lambda}\partial_vr_{\lambda}(u_0,v')dv'\leq4\int_0^v\theta^2\partial_vr(u_0,v')dv'+\frac{2\lambda^2}{2n+1}v^{2n+1}\\
    &\leq4|\Lambda|\gamma(u_0,0)\cdot v|u_0|+\gamma(u_0,0)\cdot v^{\frac{1}{2}}|u_0|^{1+\frac{1}{5n}}\leq e^{-\mathbf{b}(u_0)}v^{1/4}|u_0|\gamma(u_0,0),
\end{align*}
for $v_0'$ sufficiently small in terms of $|\Lambda|,|u_0|,\Omega^2(u_0,0).$

In order to prove \Cref{bounds for perturbed spacetime proposition}, it suffices to improve the bootstrap assumptions \eqref{A1 blueshift instability}-\eqref{A3 blueshift instability} in $[u_0,\Tilde{u}]\times[0,\Tilde{v}]$. In \Cref{A1 A2 section} we improve the bootstrap assumptions \eqref{A1 blueshift instability}-\eqref{A2 blueshift instability}. The most difficult part of the proof is improving \eqref{A3 blueshift instability}. In particular, we point out that in \eqref{A3 blueshift instability} we need to capture the same order of vanishing as on the initial data in \eqref{bound for theta initial data pert}. We establish further estimates for the solution in \Cref{refined estimates section}. We then improve the bootstrap assumption \eqref{A3 blueshift instability} in \Cref{A3 section}, completing the proof of \Cref{bounds for perturbed spacetime proposition}.

We make the notation convention for the rest of this section that we write $x\lesssim y$ for any quantities $x,y>0$ if there is a constant $C\big(k,n,|\Lambda|,|u_0|,\Omega^2(u_0,0)\big)>0$ such that $x\leq Cy.$ Additionally, for any quantity $\psi$ we denote $\widetilde{\psi}=\psi_{\lambda}-\psi.$

\subsubsection{Improving the bootstrap assumptions \eqref{A1 blueshift instability}-\eqref{A2 blueshift instability}}\label{A1 A2 section}
We begin the proof of \Cref{bounds for perturbed spacetime proposition} by improving the bootstrap assumptions \eqref{A1 blueshift instability}-\eqref{A2 blueshift instability}. Firstly, \eqref{A1 blueshift instability} implies that for any $(u,v)\in[u_0,\Tilde{u}]\times[0,\Tilde{v}]:$
\[\big|r_{\lambda}(u,v)-r(u,0)\big|\leq 10ve^{-\mathbf{b}(u)}.\]
As a result, $|u|/2\leq r_{\lambda}(u,v)\leq|u|.$ By taking $v_0'$ sufficiently small in terms of $|u_0|,\Omega^2(u_0,0),$ the bootstrap assumption \eqref{A2 blueshift instability} implies that $\gamma(u,0)\leq\gamma_{\lambda}(u,v)\leq2\gamma(u,0).$ Similarly to the proof of \eqref{bound for du r locally naked}, the bootstrap assumption \eqref{A1 blueshift instability} also implies that $1/2\leq|\partial_ur_{\lambda}|\leq1,$ provided that $v_0'$ is sufficiently small in terms of $|\Lambda|,|u_0|,\Omega^2(u_0,0).$ We use these bounds for $r_{\lambda},\gamma_{\lambda},\partial_ur_{\lambda}$ for the rest of the proof without explicit reference.

According to \eqref{background energy bound 1},\eqref{dv r in terms of b(u)}, \eqref{A1 blueshift instability}, and \eqref{A3 blueshift instability}:
\begin{align}\label{intermediate bound for theta lambda^2}
    \int_0^v\theta_{\lambda}^2\partial_vr_{\lambda}(u,v')dv'&\leq20e^{-\mathbf{b}(u)}\int_0^v\big|\theta_{\lambda}-\theta\big|^2(u,v')dv'+20e^{-\mathbf{b}(u)}\int_0^v\theta^2(u,v')dv'\\
    &\leq10^4\lambda^2v^{2n+1}e^{\mathbf{b}(u)}+100\int_0^v\theta^2\partial_vr(u,v')dv'\notag\\
    &\leq10^4\lambda^2v^{2n+1}e^{\mathbf{b}(u)}+10^6|\Lambda|\cdot v|u|e^{-\mathbf{b}(u)}\gamma(u,0)\leq e^{-\mathbf{b}(u)}v^{\frac{1}{3}}|u|\gamma(u,0),\notag
\end{align}
for $v_0'$ sufficiently small in terms of $|\Lambda|,|u_0|,\Omega^2(u_0,0).$ We note that in the derivation of \eqref{intermediate bound for theta lambda^2} we used the fact that \eqref{definition of P lambda} implies $\lambda^2v^{2n+1}e^{\mathbf{b}(u)}\leq v^{\frac{1}{2}}e^{-\mathbf{b}(u)}|u|\cdot|u_0|^{\frac{1}{5n}}\gamma(u,0).$ 

The RHS of \eqref{intermediate bound for theta lambda^2} is bounded by $2v^{\frac{1}{3}}|u_0|\Omega^2(u_0,0),$ so it is small. Thus, \eqref{dv log gamma} and \eqref{intermediate bound for theta lambda^2} imply:
\[\big|\gamma_{\lambda}(u,v)-\gamma(u,0)\big|\leq2\gamma(u,0)\int_0^v\theta_{\lambda}^2\partial_vr_{\lambda}(u,v')dv'\leq e^{-\mathbf{b}(u)}v^{\frac{1}{4}}|u|\gamma^2(u,0),\]
improving the bootstrap assumptions \eqref{A2 blueshift instability}.

Next, we estimate the blueshift difference using a similar argument to \Cref{bounds for locally naked singularity proposition}. Using \eqref{dv log gamma}, \eqref{dv r in terms of blueshift}, \eqref{A1 blueshift instability}, and \eqref{intermediate bound for theta lambda^2}, we have for all $(u,v)\in[u_0,\Tilde{u}]\times[0,\Tilde{v}]$:
\begin{align*}
    |\mathbf{b}_{\lambda}&(u,v)-\mathbf{b}(u)|\lesssim\int_0^v\int_{u_0}^u\Big(\gamma_{\lambda}\partial_vr_{\lambda}+|u'|\gamma_{\lambda}\theta_{\lambda}^2\partial_vr_{\lambda}\Big)(u',v')du'dv'\\
    &\lesssim v\int_{u_0}^u\gamma(u',0)e^{-\mathbf{b}(u')}du'+\int_{u_0}^u|u'|\gamma(u',0)\bigg(\int_0^v\theta_{\lambda}^2\partial_vr_{\lambda}(u',v')dv'\bigg)du'\\
    &\lesssim v|u_0|\Omega^2(u_0,0)+v^{\frac{1}{3}}\int_{u_0}^u|u'|^2\gamma^2(u',0)e^{-\mathbf{b}(u')}du'\lesssim v|u_0|\Omega^2(u_0,0)+v^{\frac{1}{3}}\Omega^2(u_0,0)\cdot\mathbf{s}(u).
\end{align*}
Thus, by \eqref{definition of P lambda} we have for $v_0'$ sufficiently small in terms of $n,|\Lambda|,|u_0|,\Omega^2(u_0,0):$
\begin{equation}\label{blueshift difference bound}
    |\mathbf{b}_{\lambda}(u,v)-\mathbf{b}(u)|\leq v^{\frac{1}{4}}.
\end{equation}
As a result, we have that for any $(u,v)\in[u_0,\Tilde{u}]\times[0,\Tilde{v}]$:
\begin{equation}\label{dv r lambda in terms of b(u)}
    \partial_vr_{\lambda}(u,v)=\frac{1}{2}e^{-\mathbf{b}_{\lambda}(u,v)}=\frac{1}{2}e^{-\mathbf{b}(u)}\big(1+O(v^{\frac{1}{4}})\big),
\end{equation}
which improves the bootstrap assumption \eqref{A1 blueshift instability}.

\subsubsection{Refined estimates}\label{refined estimates section}
The final step in the proof of \Cref{bounds for perturbed spacetime proposition} is to improve the bootstrap assumption \eqref{A3 blueshift instability}. For this we first need to establish more precise bounds which capture higher order vanishing properties at $v=0.$ We also recall that we are using the notation $\widetilde{\psi}=\psi_{\lambda}-\psi.$
\begin{lemma}
    Under the bootstrap assumptions \eqref{A1 blueshift instability}-\eqref{A3 blueshift instability}, the following bounds hold for all $(u,v)\in[u_0,\Tilde{u}]\times[0,\Tilde{v}]$:
    \begin{align}
        \big|\widetilde{\gamma}\big|(u,v)\lesssim& ve^{-\mathbf{b}(u)}|u|\gamma^2(u,0)\big[v\mathbf{s}(u)\big]^{n}+|\lambda| v^{n+1}\gamma(u,0)\sqrt{|u|\gamma(u,0)}+\lambda^2v^{2n+1}\gamma(u,0)e^{\mathbf{b}(u)},\label{blueshift lemma bound 1}\\
        \big|\widetilde{\partial_vr}\big|(u,v)\lesssim& e^{-\mathbf{b}(u)}\big[v\mathbf{s}(u)\big]^{n+1}+|\lambda|v^{n+1}e^{-\mathbf{b}(u)}\mathbf{x}(u)+\lambda^2v^{2n+1},\label{blueshift lemma bound 2}\\
        \big|\widetilde{r}\big|(u,v)\lesssim & e^{-\mathbf{b}(u)}v\big[v\mathbf{s}(u)\big]^{n+1}+|\lambda| v^{n+2}e^{-\mathbf{b}(u)}\mathbf{x}(u)+\lambda^2v^{2n+2},\label{blueshift lemma bound 3}\\
        \big|\widetilde{\partial_ur}\big|(u,v)\lesssim& v|u|\big[v\mathbf{s}(u)\big]^{n}+|\lambda| v^{n+1}e^{-\mathbf{b}(u)/2}|u|^{1/2}\gamma(u,0)+\lambda^2 v^{2n+1}\gamma(u,0),\label{blueshift lemma bound 4}\\
        \big|\widetilde{\sqrt{r}\partial_u\phi}\big|(u,v)\lesssim& v|u|^{-1/2}e^{-\mathbf{b}(u)}\sqrt{\gamma(u,0)}\big[v\mathbf{s}(u)\big]^{n}+|\lambda|\cdot|u|^{-1}v^{n+1},\label{blueshift lemma bound 5}
    \end{align}
    where we use the notation:
    \[\mathbf{x}(u)=\int_{u_0}^u\big(|u'|\gamma(u',0)\big)^{3/2}du'.\]
\end{lemma}
\begin{proof}
    We notice that we have the inequality:
    \[\big|\widetilde{\theta^2\partial_vr}\big|\lesssim\theta^2\big|\widetilde{\partial_vr}\big|+e^{-\mathbf{b}(u)}|\theta|\cdot\big|\widetilde{\theta}\big|+e^{-\mathbf{b}(u)}\big|\widetilde{\theta}\big|^2.\]
    We integrate this inequality, and use \eqref{background energy bound 1}, \eqref{dv r in terms of b(u)}, \eqref{dv r lambda in terms of b(u)} for the first term, \eqref{background energy bound 1}, \eqref{A3 blueshift instability} for the second term, and \eqref{A3 blueshift instability} for the third term. Thus, we get that for all $(u,v)\in[u_0,\Tilde{u}]\times[0,\Tilde{v}]$:
    \begin{align}\label{preliminary bound difference of theta}
        \int_0^v\big|\widetilde{\theta^2\partial_vr}\big|(u,v')dv'&\lesssim v|u|\gamma(u,0)\sup_{v'\in[0,v]}\big|\widetilde{\partial_vr}\big|(u,v')+|\lambda| v^{n+1}\sqrt{|u|\gamma(u,0)}+\lambda^2v^{2n+1}e^{\mathbf{b}(u)}\\
        &\lesssim v|u|\gamma(u,0)e^{-\mathbf{b}(u)}+|\lambda|v^{n+1}\sqrt{|u|\gamma(u,0)}+\lambda^2v^{2n+1}e^{\mathbf{b}(u)}.\notag
    \end{align}
    According to \eqref{definition of P lambda}, the RHS of \eqref{preliminary bound difference of theta} is bounded up to a constant by $v^{\frac{1}{2}}|u_0|^{1+\frac{1}{5n}}\Omega^2(u_0,0),$ so it is small for $v_0'$ sufficiently small in terms of $|\Lambda|,|u_0|,\Omega^2(u_0,0).$ As a result, \eqref{dv log gamma} implies that:
    \begin{align}\label{preliminary bound difference of gamma}
        \big|\widetilde{\gamma}\big|(u,v)&\lesssim \gamma(u,v)\cdot\int_0^v\big|\widetilde{\theta^2\partial_vr}\big|(u,v')dv'\\
        &\lesssim v|u|\gamma^2(u,0)e^{-\mathbf{b}(u)}+|\lambda|v^{n+1}\gamma(u,0)\sqrt{|u|\gamma(u,0)}+\lambda^2v^{2n+1}\gamma(u,0)e^{\mathbf{b}(u)}.\notag
    \end{align}
    
    By induction we prove that for any $0\leq i\leq n$ and $(u,v)\in[u_0,\Tilde{u}]\times[0,\Tilde{v}]$:
    \begin{align}
        \big|\widetilde{\gamma}\big|(u,v)\lesssim & ve^{-\mathbf{b}(u)}|u|\gamma^2(u,0)\big[v\mathbf{s}(u)\big]^{i}+|\lambda|v^{n+1}\gamma(u,0)\sqrt{|u|\gamma(u,0)}+\lambda^2v^{2n+1}\gamma(u,0)e^{\mathbf{b}(u)},\label{induction g}\\
        \big|\widetilde{\partial_vr}\big|(u,v)\lesssim & e^{-\mathbf{b}(u)}\big[v\mathbf{s}(u)\big]^{i}+|\lambda|v^{n+1}e^{-\mathbf{b}(u)}\mathbf{x}(u)+\lambda^2v^{2n+1}.\label{induction r}
    \end{align}
    According to \eqref{dv r in terms of b(u)}, \eqref{dv r lambda in terms of b(u)}, and \eqref{preliminary bound difference of gamma}, we already established the case $i=0$. We assume the induction hypothesis holds for $0\leq i\leq n-1$ and prove it for $i+1.$ We first have the preliminary estimate:
    \begin{equation}\label{preliminary bound difference of r}
        \big|\widetilde{r}\big|(u,v)\lesssim e^{-\mathbf{b}(u)}v\big[v\mathbf{s}(u)\big]^{i}+|\lambda| v^{n+2}e^{-\mathbf{b}(u)}\mathbf{x}(u)+\lambda^2v^{2n+2}.
    \end{equation}
    As a result, we get by combining \eqref{induction g} with \eqref{preliminary bound difference of r} and using \eqref{gamma lower bound on C0-} that for any $(u,v)\in[u_0,\Tilde{u}]\times[0,\Tilde{v}]$:
    \begin{align}\label{bound for tilde rgamma}
        \big|\widetilde{r\gamma}\big|(u,v)\lesssim & ve^{-\mathbf{b}(u)}|u|^2\gamma^2(u,0)\big[v\mathbf{s}(u)\big]^{i}+|\lambda| v^{n+1}\big(|u|\gamma(u,0)\big)^{3/2}+\lambda^2v^{2n+1}|u|\gamma(u,0)e^{\mathbf{b}(u)}\\
        & +|\lambda|\Omega^2(u_0,0) v^{n+2}\mathbf{x}(u)+\lambda^2\Omega^2(u_0,0)v^{2n+2}e^{\mathbf{b}(u)}.\notag
    \end{align}
    We notice that by \Cref{blueshift definition}, we have that:
    \begin{equation}\label{bound for the old y}
        \int_{u_0}^u|u'|\gamma(u',0)e^{\mathbf{b}(u')}du'\lesssim e^{\mathbf{b}(u)}.
    \end{equation}
    We integrate \eqref{bound for tilde rgamma}, using \eqref{bound for the old y} and the fact that $\mathbf{s},\mathbf{x}$ are non-decreasing to get that:
    \begin{equation}\label{preliminary bound difference of r gamma}
        \int_{u_0}^u\big|\widetilde{r\gamma}\big|(u',v)du'\lesssim \big[v\mathbf{s}(u)\big]^{i+1}+|\lambda| v^{n+1}\mathbf{x}(u)+\lambda^2v^{2n+1}e^{\mathbf{b}(u)}.
    \end{equation}
    We bound the last two terms on the RHS of \eqref{preliminary bound difference of r gamma} using \eqref{definition of P lambda}:
    \begin{align}
        |\lambda| v^{n+1}\mathbf{x}(u)&\lesssim\int_{u_0}^u|u'|^{\frac{3}{2}}\gamma(u',0)\cdot v^{\frac{3}{4}}|u'|^{\frac{1}{2}+\frac{1}{10n}}e^{-\mathbf{b}(u')}\gamma(u',0)du'\lesssim |u_0|^{\frac{1}{10n}}v^{\frac{3}{4}}\mathbf{s}(u),\label{preliminary bound x}\\
        \lambda^2v^{2n+1}e^{\mathbf{b}(u)}&\lesssim v^{\frac{1}{2}}e^{-\mathbf{b}(u)}|u|^{1+\frac{1}{5n}}\gamma(u,0)\lesssim|u_0|^{1+\frac{1}{5n}}v^{\frac{1}{2}}.\label{preliminary bound y}
    \end{align}
    In particular, \eqref{preliminary bound x}-\eqref{preliminary bound y} imply that the RHS of \eqref{preliminary bound difference of r gamma} is small, so using \eqref{wave equation r} we get:
    \begin{equation}\label{preliminary bound difference of dv r}
        \big|\widetilde{\partial_vr}\big|(u,v)\lesssim\partial_vr(u,v)\cdot\bigg|\exp\bigg(\int_{u_0}^u\frac{\Lambda}{2}\widetilde{r\gamma}(u',v)du'\bigg)-1\bigg|\lesssim e^{-\mathbf{b}(u)}\cdot\int_{u_0}^u\big|\widetilde{r\gamma}\big|(u',v)du'.
    \end{equation}
    Together with \eqref{preliminary bound difference of r gamma}, this implies \eqref{induction r} for $i+1.$ Similarly to the derivation of \eqref{preliminary bound difference of theta}, we have that:
    \begin{align}
        \int_0^v\big|\widetilde{\theta^2\partial_vr}\big|(u,v')dv'\lesssim&ve^{-\mathbf{b}(u)}|u|\gamma(u,0)\big[v\mathbf{s}(u)\big]^{i+1}+|\lambda|\cdot|u|\gamma(u,0)v^{n+2}e^{-\mathbf{b}(u)}\mathbf{x}(u)\label{mid preliminary bound difference of theta}\\
        &+\lambda^2|u|\gamma(u,0)v^{2n+2}+|\lambda| v^{n+1}\sqrt{|u|\gamma(u,0)}+\lambda^2v^{2n+1}e^{\mathbf{b}(u)},\notag
    \end{align}
    where the first three terms follow by using \eqref{induction r} for $i+1,$ and the last two terms are the same as in \eqref{preliminary bound difference of theta}. We bound the second term on the RHS of \eqref{mid preliminary bound difference of theta} using \eqref{definition of P lambda} as follows:
    \begin{align}\label{some bound I need later}
        |\lambda|\cdot|u|\gamma(u,0)v^{n+2}e^{-\mathbf{b}(u)}\mathbf{x}(u)&\lesssim|\lambda| v^{n+1}\sqrt{|u|\gamma(u,0)}\cdot v\sqrt{e^{-\mathbf{b}(u)}\gamma(u,0)}\int_{u_0}^u|u'|^2\gamma(u',0)\sqrt{e^{-\mathbf{b}(u')}\gamma(u',0)}du'\notag\\
        &\lesssim|\lambda| v^{n+1}\sqrt{|u|\gamma(u,0)}\cdot v\mathbf{s}(u)\lesssim|\lambda| v^{n+1}\sqrt{|u|\gamma(u,0)}.
    \end{align}
    Additionally, the third term on the RHS of \eqref{mid preliminary bound difference of theta} is bounded by the fifth term. Therefore, we can further bound the terms on the RHS of \eqref{mid preliminary bound difference of theta} to get for any $(u,v)\in[u_0,\Tilde{u}]\times[0,\Tilde{v}]$:
    \begin{equation}\label{better preliminary bound difference of theta}
        \int_0^v\big|\widetilde{\theta^2\partial_vr}\big|(u,v')dv'\lesssim ve^{-\mathbf{b}(u)}|u|\gamma(u,0)\big[v\mathbf{s}(u)\big]^{i+1}+|\lambda| v^{n+1}\sqrt{|u|\gamma(u,0)}+\lambda^2v^{2n+1}e^{\mathbf{b}(u)}.
    \end{equation}
    Arguing as we did for \eqref{preliminary bound difference of theta} above, the RHS of \eqref{better preliminary bound difference of theta} is small. We can repeat the derivation of \eqref{preliminary bound difference of gamma} and use \eqref{better preliminary bound difference of theta} instead of \eqref{preliminary bound difference of theta}. This implies that \eqref{induction g} holds for $i+1,$ completing our induction argument. 

    In particular, the estimate \eqref{induction g} for $i=n$ implies that \eqref{blueshift lemma bound 1} holds. Iterating the argument used in proving \eqref{preliminary bound difference of r}-\eqref{preliminary bound difference of dv r}, while starting with \eqref{induction g}-\eqref{induction r} for $i=n$, we also get \eqref{blueshift lemma bound 2} and the desired bound for $\widetilde{r}$ in \eqref{blueshift lemma bound 3}. Next, we have by \eqref{wave equation r}, \eqref{definition of P lambda}, and \eqref{blueshift lemma bound 1}-\eqref{blueshift lemma bound 3} that for any $(u,v)\in[u_0,\Tilde{u}]\times[0,\Tilde{v}]$:
    \begin{align*}
        \big|\widetilde{\partial_ur}\big|(u,v)&\lesssim\int_{0}^v\big|\widetilde{r\gamma\partial_vr}\big|(u,v')dv'\lesssim\int_{0}^v\bigg[\gamma(u,0)e^{-\mathbf{b}(u)}\big|\widetilde{r}\big|+|u|e^{-\mathbf{b}(u)}\big|\widetilde{\gamma}\big|+|u|\gamma(u,0)\big|\widetilde{\partial_vr}\big|\bigg](u,v')dv'\\
        &\lesssim v|u|\big[v\mathbf{s}(u)\big]^{n}+|\lambda|v^{n+2}e^{-\mathbf{b}(u)}|u|\gamma(u,0)\Big[\mathbf{x}(u)+\sqrt{|u|\gamma(u,0)}\Big]+\lambda^2 v^{2n+2}|u|\gamma(u,0)\\
        &\lesssim v|u|\big[v\mathbf{s}(u)\big]^{n}+|\lambda|v^{n+1}e^{-\mathbf{b}(u)/2}|u|^{1/2}\gamma(u,0)+\lambda^2 v^{2n+1}\gamma(u,0),
    \end{align*}
    which completes the proof of \eqref{blueshift lemma bound 4}. 

    In order to prove \eqref{blueshift lemma bound 5}, we first derive the following equation as a consequence of \eqref{wave equation phi}:
    \[\partial_v\big(\sqrt{r}\partial_u\phi\big)=-\frac{1}{2\sqrt{r}}\partial_ur\partial_v\phi=-\frac{1}{2r}\partial_ur\partial_vr\cdot\theta.\]
    We use this to get an equation for $\partial_v\big(\widetilde{\sqrt{r}\partial_u\phi}\big)$. Thus, we start with the following bound:
    \begin{equation}\label{preliminary bound difference du phi}
        \big|\widetilde{\sqrt{r}\partial_u\phi}\big|(u,v)\lesssim\int_0^v\bigg(|u|^{-1}\big|\widetilde{\partial_vr}\cdot\theta\big|+e^{-\mathbf{b}(u)}|u|^{-2}\big|\widetilde{r}\cdot\theta\big|+e^{-\mathbf{b}(u)}|u|^{-1}\Big(\big|\widetilde{\partial_ur}\cdot\theta\big|+\big|\widetilde{\theta}\big|\Big)\bigg)(u,v')dv'.
    \end{equation}
    We bound the terms on the RHS of \eqref{preliminary bound difference du phi} one by one. For the first term we have using \eqref{background energy bound 1}, \eqref{definition of P lambda}, \eqref{blueshift lemma bound 2}, and \eqref{some bound I need later}:
    \begin{align*}
        \int_0^v|u|^{-1}\big|\widetilde{\partial_vr}\cdot\theta\big|(u,v')dv'\lesssim &v|u|^{-1/2}e^{-\mathbf{b}(u)}\sqrt{\gamma(u,0)}\big[v\mathbf{s}(u)\big]^{n}+|u|^{-1}|\lambda|\sqrt{|u|\gamma(u,0)}v^{n+2}e^{-\mathbf{b}(u)}\mathbf{x}(u)\\
        &+\lambda^2v^{2n+2}|u|^{-1/2}\sqrt{\gamma(u,0)}\\
        \lesssim &v|u|^{-1/2}e^{-\mathbf{b}(u)}\sqrt{\gamma(u,0)}\big[v\mathbf{s}(u)\big]^{n}+|u|^{-1}|\lambda|v^{n+1}.
    \end{align*}
    Additionally, the second term on the RHS of \eqref{preliminary bound difference du phi} is bounded by the same quantity. This follows since $e^{-\mathbf{b}(u)}|u|^{-2}\lesssim1$ and the RHS of \eqref{blueshift lemma bound 3} is bounded by the RHS of \eqref{blueshift lemma bound 2}. For the third term on the RHS of \eqref{preliminary bound difference du phi} we use \eqref{background energy bound 1} and \eqref{blueshift lemma bound 4}:
    \begin{equation*}
        \int_0^ve^{-\mathbf{b}(u)}|u|^{-1}\big|\widetilde{\partial_ur}\cdot\theta\big|(u,v')dv'\lesssim v^2e^{-\mathbf{b}(u)}\sqrt{|u|\gamma(u,0)}\big[v\mathbf{s}(u)\big]^{n}+|\lambda|v^{n+2}+\lambda^2v^{2n+2}|u|^{-1/2}\sqrt{\gamma(u,0)}.
    \end{equation*}
    Finally, for the fourth term on the RHS of \eqref{preliminary bound difference du phi} we use the bootstrap assumption \eqref{A3 blueshift instability}:
    \[\int_0^ve^{-\mathbf{b}(u)}|u|^{-1}\big|\widetilde{\theta}\big|(u,v')dv'\lesssim\sqrt{v}e^{-\mathbf{b}(u)}|u|^{-1}\bigg(\int_0^v\big|\widetilde{\theta}\big|^2(u,v')dv'\bigg)^{1/2}\lesssim|u|^{-1}|\lambda|v^{n+1}.\]
    Combining the estimates for the terms on the RHS of \eqref{preliminary bound difference du phi}, we completed the proof of \eqref{blueshift lemma bound 5}.
    %For the fourth term on the RHS of \eqref{preliminary bound difference du phi} we use \eqref{background energy bound 1} and \eqref{blueshift lemma bound 3}:\begin{align*}\int_0^ve^{-\mathbf{b}(u)}|u|^{-2}\big|\widetilde{r}\cdot\theta\big|(u,v')dv'\lesssim &v^2e^{-\mathbf{b}(u)}\sqrt{|u|\gamma(u,0)}\big[v\mathbf{s}(u)\big]^{n}+|\lambda|\sqrt{|u|\gamma(u,0)}v^{n+3}e^{-\mathbf{b}(u)}\mathbf{x}(u)\\ &+\lambda^2v^{2n+3}\sqrt{|u|\gamma(u,0)}\\ \lesssim &v^2e^{-\mathbf{b}(u)}\sqrt{|u|\gamma(u,0)}\big[v\mathbf{s}(u)\big]^{n}+\lambda|v^{n+2}.\end{align*}
\end{proof}

\subsubsection{Improving the bootstrap assumption \eqref{A3 blueshift instability}}\label{A3 section}
The goal of this section is to improve the bootstrap assumption \eqref{A3 blueshift instability} and complete the proof of \Cref{bounds for perturbed spacetime proposition}. We first notice that the wave equations \eqref{wave equation phi} and \eqref{wave equation r} imply:
\[\partial_u\theta=-\frac{\Lambda}{2}r\gamma\theta-\frac{1}{2\sqrt{r}}\partial_u\phi.\]
As a result, we can derive the following equation for $\widetilde{\theta}:$
\begin{align}
    \partial_u\widetilde{\theta}=&-\frac{\Lambda}{2}r_{\lambda}\gamma_{\lambda}\widetilde{\theta}+Err,\label{du tilde theta}\\
    Err=&-\frac{\Lambda}{2}\widetilde{r\gamma}\cdot\theta-\frac{1}{2}\widetilde{r^{-1/2}\partial_u\phi}. \label{def of err instability}
\end{align}
We integrate equation \eqref{du tilde theta} to obtain for all $(u,v)\in[u_0,\Tilde{u}]\times[0,\Tilde{v}]$:
\begin{equation}\label{tilde theta integrated}
    \widetilde{\theta}(u,v)=e^{\mathbf{b}_{\lambda}(u,v)}\cdot\bigg(\widetilde{\theta}(u_0,v)+\int_{u_0}^uErr\cdot e^{-\mathbf{b}_{\lambda}(u',v)}(u',v)du'\bigg).
\end{equation}

The essential point that allows us to improve the bootstrap assumption \eqref{A3 blueshift instability} is that the main contribution in \eqref{tilde theta integrated} comes from the first term, determined by the initial data perturbation \eqref{characteristic data perturbation}. We deal with the second term in \eqref{tilde theta integrated} using \Cref{bound for Err instability lemma} below. We assume for now the bound \eqref{bound for Err instability} and complete the proof of \Cref{bounds for perturbed spacetime proposition}.

Combining \eqref{tilde theta integrated} and \eqref{bound for Err instability}, we get using \eqref{blueshift difference bound} that for all $(u,v)\in[u_0,\Tilde{u}]\times[0,\Tilde{v}]$:
\[\int_0^v\Big|\widetilde{\theta}(u,v')-e^{\mathbf{b}_{\lambda}(u,v')}\widetilde{\theta}(u_0,v')\Big|^2dv'\lesssim\int_0^ve^{2\mathbf{b}_{\lambda}(u,v')}\bigg(\int_{u_0}^u\big|Err(u',v')\big|e^{-\mathbf{b}(u')}du'\bigg)^2dv'\lesssim\lambda^2v^{2n+\frac{3}{2}}e^{2\mathbf{b}(u)}.\]
We denote the implicit constant in this estimate by $C>0,$ which depends on $k,n,|\Lambda|,|u_0|,$ and $\Omega^2(u_0,0).$

On the other hand, we have by the definition \eqref{characteristic data perturbation} that $\big|\widetilde{\theta}(u_0,v)\big|=|\lambda|v^n$. We also use \eqref{blueshift difference bound} to get for all $(u,v)\in[u_0,\Tilde{u}]\times[0,\Tilde{v}]$:
\begin{align*}
    \int_0^v\big|\widetilde{\theta}(u,v')\big|^2dv'\leq&2\int_0^v\Big|\widetilde{\theta}(u,v')-e^{\mathbf{b}_{\lambda}(u,v')}\widetilde{\theta}(u_0,v')\Big|^2dv'+2\int_0^v\Big|e^{\mathbf{b}_{\lambda}(u,v')}\widetilde{\theta}(u_0,v')\Big|^2dv'\\
    \leq & 2C\lambda^2v^{2n+\frac{3}{2}}e^{2\mathbf{b}(u)}+2\int_0^v\Big|e^{\mathbf{b}(u)}\big(1+2v'^{\frac{1}{4}}\big)\lambda v'^n\Big|^2dv'\\
    \leq & 2(C+8)\lambda^2v^{2n+\frac{3}{2}}e^{2\mathbf{b}(u)}+4\lambda^2v^{2n+1}e^{2\mathbf{b}(u)}\leq10\lambda^2v^{2n+1}e^{2\mathbf{b}(u)},
\end{align*}
provided that $v_0'>0$ is small enough in terms of $k,n,|\Lambda|,|u_0|,$ and $\Omega^2(u_0,0).$ Conditioned on \eqref{bound for Err instability}, this improves the bootstrap assumption \eqref{A3 blueshift instability}, concluding the bootstrap argument. We notice that in the above argument we also established \eqref{B4 blueshift instability} (again, by taking $v_0'>0$ sufficiently small), which was the remaining bound in \Cref{bounds for perturbed spacetime proposition}.

In the above, we completed the proof of \Cref{bounds for perturbed spacetime proposition}, conditioned on \eqref{bound for Err instability}. We now establish this bound to complete our argument.
\begin{lemma}\label{bound for Err instability lemma}
    Under the bootstrap assumptions \eqref{A1 blueshift instability}-\eqref{A3 blueshift instability}, the 
    error term \eqref{def of err instability} satisfies the following bound for all $(u,v)\in[u_0,\Tilde{u}]\times[0,\Tilde{v}]$:
    \begin{equation}\label{bound for Err instability}
        \int_0^v\bigg(\int_{u_0}^u\big|Err(u',v')\big|e^{-\mathbf{b}(u')}du'\bigg)^2dv'\lesssim\lambda^2v^{2n+\frac{3}{2}}.
    \end{equation}
\end{lemma}
\begin{proof}
    Using the definition \eqref{def of err instability}, we bound $Err$ as follows:
    \begin{equation}\label{bound for Err using I II III}
        \big|Err\big|(u,v)\lesssim|u|^{-1}\big|\widetilde{\sqrt{r}\partial_u\phi}\big|(u,v)+|u|^{-\frac{3}{2}}\big|\widetilde{r}\cdot\partial_u\phi\big|(u,v)+\big|\widetilde{r\gamma}\cdot\theta\big|(u,v).
    \end{equation}
    Upon integration as on the LHS of \eqref{bound for Err instability}, the bound \eqref{bound for Err using I II III} gives rise to three error terms, which we denote by $I,II,$ and $III.$ In order to bound $I,$ we use \eqref{blueshift lower bound eq}, \eqref{definition of P lambda}, and \eqref{blueshift lemma bound 5}:
    \begin{align*}
        \int_{u_0}^u|u'|^{-1}e^{-\mathbf{b}(u')}\big|\widetilde{\sqrt{r}\partial_u\phi}\big|(u',v)du'&\lesssim\int_{u_0}^u\bigg(|u'|^{-\frac{3}{2}}e^{-2\mathbf{b}(u')}\sqrt{\gamma(u',0)}v\big[v\mathbf{s}(u')\big]^n+|\lambda|v^{n+1}|u'|^{-2}e^{-\mathbf{b}(u')}\bigg)du'\\
        &\lesssim v\big[v\mathbf{s}(u)\big]^n+|\lambda|v^{n+1}\lesssim(\epsilon')^nv^{n+\frac{9}{10}}+|\lambda|v^{n+1}\lesssim|\lambda|v^{n+\frac{1}{2}},
    \end{align*}
    where we also used the fact that $\epsilon'\lesssim1$ and $v\lesssim\lambda^4.$ We square this relation and integrate in $v$ to get that $I$ is bounded by the RHS of \eqref{bound for Err instability}.

    In order to bound $II,$ we first notice that \eqref{background energy bound 2} implies for any $(u,v)\in[u_0,\Tilde{u}]\times[0,\Tilde{v}]$:
    \begin{align*}
        \bigg(\int_{u_0}^u|u'|^{-\frac{3}{2}}e^{-\mathbf{b}(u')}\big|\widetilde{r}\partial_u\phi\big|(u',v)du'\bigg)^2&\lesssim\int_{u_0}^u|u'|\gamma^{-1}(u',0)\big|\partial_u\phi\big|^2(u',v)du'\cdot\int_{u_0}^u|u'|^{-4}\gamma(u',0)e^{-2\mathbf{b}(u')}\big|\widetilde{r}\big|^2du'\\
        &\lesssim\int_{u_0}^u|u'|^{-4}\gamma(u',0)e^{-2\mathbf{b}(u')}\big|\widetilde{r}\big|^2(u',v)du'.
    \end{align*}
    We use \eqref{blueshift lemma bound 3}, which creates three error terms that we bound using \eqref{definition of P lambda}, \eqref{some bound I need later} as follows:
    \[\int_{u_0}^u|u'|^{-4}\gamma(u',0)e^{-4\mathbf{b}(u')}v^2\big[v\mathbf{s}(u')\big]^{2n+2}du'\lesssim v^2\big[v\mathbf{s}(u)\big]^{2n+2}\lesssim\lambda^2v^{2n+2},\]
    \[\lambda^2v^{2n+4}\int_{u_0}^u|u'|^{-4}\gamma(u',0)e^{-4\mathbf{b}(u')}\mathbf{x}^2(u')du'\lesssim\lambda^2v^{2n+2}\int_{u_0}^u|u'|^{-5}e^{-2\mathbf{b}(u')}du'\lesssim\lambda^2v^{2n+2}\big|\log|u|\big|\lesssim\lambda^2v^{2n+\frac{3}{2}},\]
    \[\int_{u_0}^u|u'|^{-4}\gamma(u',0)e^{-2\mathbf{b}(u')}\lambda^4v^{4n+4}du'\lesssim\lambda^2v^{2n+3}\int_{u_0}^u|u'|^{-3}\gamma^2(u',0)e^{-4\mathbf{b}(u')}du'\lesssim\lambda^2v^{2n+3}.\]
    Integrating in $v$ and summing these terms, we get that $II$ is bounded by the RHS of \eqref{bound for Err instability}.

    The last step of the proof if to bound $III.$ We have by Cauchy-Schwarz that:
    \begin{equation}\label{bound for III}
        III\lesssim\int_{u_0}^u\int_0^v|u'|\theta^2(u',v')dv'du'\cdot\sup_{v'\in[0,v]}\int_{u_0}^u|u'|^{-1}e^{-2\mathbf{b}(u')}\big|\widetilde{r\gamma}\big|^2(u',v')du'.
    \end{equation}
    For the first term, we have according to \eqref{background energy bound 1} and \eqref{definition of P lambda}:
    \begin{equation}\label{basic bound for theta instability}
        \int_{u_0}^u|u'|\int_0^v\theta^2(u',v')dv'du'\lesssim v\int_{u_0}^u|u'|^2\gamma(u',0)du'\lesssim v\mathbf{s}(u)\lesssim v^{1/2}.
    \end{equation}
    In order to deal with the second term in \eqref{bound for III}, we use \eqref{bound for tilde rgamma} for $i=n$ (which follows by \eqref{blueshift lemma bound 1} and \eqref{blueshift lemma bound 3}) to get for $(u,v)\in[u_0,\Tilde{u}]\times[0,\Tilde{v}]$:
    \begin{align*}
        e^{-2\mathbf{b}(u)}\big|\widetilde{r\gamma}\big|^2(u,v)\lesssim & v^2|u|^4\big[v\mathbf{s}(u)\big]^{2n}+\lambda^2 v^{2n+2}|u|^3\gamma(u,0)+\lambda^4v^{4n+2}|u|^2\gamma^2(u,0)+\lambda^2v^{2n+4}e^{-2\mathbf{b}(u)}\mathbf{x}^2(u)\\
        \lesssim & v^2|u|^4\big[v\mathbf{s}(u)\big]^{2n}+\lambda^2 v^{2n+2}|u|^3\gamma(u,0)+\lambda^4v^{4n+2}|u|^2\gamma^2(u,0)+\lambda^2v^{2n+4}|u|\mathbf{s}^2(u).
    \end{align*}
    Additionally, we use \eqref{definition of P lambda} as well to get for all $(u,v)\in[u_0,\Tilde{u}]\times[0,\Tilde{v}]$:
    \begin{equation}\label{final bound for tilde r gamma}
        |u|^{-1}e^{-2\mathbf{b}(u)}\big|\widetilde{r\gamma}\big|^2(u,v)\lesssim \lambda^2v^{2n+1}+\lambda^2 v^{2n+\frac{3}{2}}|u|^2\gamma(u,0).
    \end{equation}
    We integrate \eqref{final bound for tilde r gamma}, using \eqref{blueshift strength definition} and \eqref{definition of P lambda}. We plug this bound into \eqref{bound for III}, together with \eqref{basic bound for theta instability}, to conclude that $III$ is bounded by the RHS of \eqref{bound for Err instability}. This completes the proof of \eqref{bound for Err instability}.
\end{proof}

\subsection{Proof of \Cref{instability to trapped surfaces proposition}}\label{proof of instability thm section}
In this section, we complete the proof of \Cref{instability to trapped surfaces proposition}. The main role is played by the curve $\mathcal{C}_{\lambda}$:
\begin{equation}\label{definition of C lambda}
    \mathcal{C}_{\lambda}=\Big\{\big(u,v_{\sharp}(u)\big):\ u\in[u_0,0),\ v_{\sharp}(u)>0,\ |\lambda|\big[v_{\sharp}(u)\big]^{n}=e^{-\mathbf{b}(u)}|u|^{\frac{1}{2}-\frac{1}{10n}}\sqrt{\gamma(u,0)}\Big\},
\end{equation} 
represented schematically in \Cref{fig:instability}. Below, we shall see that the blueshift instability drives the growth of certain quantities and allows us to  apply the formation of trapped surfaces criterion in \Cref{trapped surface formation theorem} for points on $\mathcal{C}_{\lambda}.$

We first show that the curve $\mathcal{C}_{\lambda}$ is indeed contained in the domain $\mathcal{P}_{\lambda},$ for which we proved quantitative estimates in \Cref{bounds for perturbed spacetime proposition}.

\begin{lemma}\label{curve C lemma}
    The curve $\mathcal{C}_{\lambda}$ converges to $b_{\Gamma}=(0,0)$ as $u\rightarrow0$. Moreover, for all $|u|$ small enough, $\mathcal{C}_{\lambda}$ is contained in the region $\mathcal{P}_{\lambda}$ defined in \Cref{bounds for perturbed spacetime proposition}.
\end{lemma}
\begin{proof}
    The curve $\mathcal{C}_{\lambda}$ converges to $b_{\Gamma}=(0,0)$ as $u\rightarrow0$ since:
    \[|\lambda|\big[v_{\sharp}(u)\big]^{n}=e^{-\mathbf{b}(u)}|u|^{\frac{1}{2}-\frac{1}{10n}}\sqrt{\gamma(u,0)}\leq2\Omega(u_0,0)\cdot e^{-\mathbf{b}(u)/2}|u|^{\frac{1}{2}-\frac{1}{10n}}\leq2\Omega(u_0,0)\cdot|u|^{\frac{3}{2}-\frac{1}{10n}}|u_0|^{-1}.\]
    The same bound also implies that for $|u|$ small enough we get $v_{\sharp}(u)\leq\min\big(v_0',\lambda^4,|u|^{\frac{1}{n}}\big).$ In order to check the next condition in the definition of $\mathcal{P}_{\lambda}$ in \eqref{definition of P lambda}, we notice that for all $|u|$ small enough:
    \begin{align*}
        |\lambda|\big[v_{\sharp}(u)\big]^{n+\frac{1}{4}}&=\big[v_{\sharp}(u)\big]^{\frac{1}{4}}\cdot e^{-\mathbf{b}(u)}|u|^{\frac{1}{2}-\frac{1}{10n}}\sqrt{\gamma(u,0)}\\
        &\leq e^{-\mathbf{b}(u)}|u|^{\frac{1}{2}+\frac{1}{4n}-\frac{1}{10n}}\sqrt{\gamma(u,0)}\leq e^{-\mathbf{b}(u)}|u|^{\frac{1}{2}+\frac{1}{10n}}\sqrt{\gamma(u,0)},
    \end{align*}
    so the second condition in \eqref{definition of P lambda} holds as well. Finally, for the third condition in \eqref{definition of P lambda}, we compute that:
    \begin{align*}
        \big[v_{\sharp}(u)\big]^{\frac{1}{10n}}\cdot\mathbf{s}(u)&=|\lambda|^{-\frac{1}{10n^2}}e^{-\frac{\mathbf{b}(u)}{10n^2}}|u|^{\frac{1}{20n^2}-\frac{1}{100n^3}}\big[\gamma(u,0)\big]^{\frac{1}{20n^2}}\cdot\mathbf{s}(u)\\
        &\lesssim|\lambda|^{-\frac{1}{10n^2}}|u_0|\cdot\big[\Omega^2(u_0,0)\big]^{\frac{1}{20n^2}}\cdot e^{-\frac{\mathbf{b}(u)}{20n^2}}\mathbf{b}(u)\cdot|u|^{\frac{1}{20n^2}-\frac{1}{100n^3}}\rightarrow0\text{ as }u\rightarrow0,
    \end{align*}
    where we used the fact that $\mathbf{s}(u)\leq|u_0|\mathbf{b}(u)$ and $\mathbf{b}(u)\rightarrow\infty$ as $u\rightarrow0.$ This implies that $\big[v_{\sharp}(u)\big]^{\frac{1}{10n}}\cdot\mathbf{s}(u)\lesssim\epsilon'$ for all $|u|$ sufficiently small, so the curve $\mathcal{C}_{\lambda}$ is contained in the region $\mathcal{P}_{\lambda}$ for $|u|$ small enough.
\end{proof}

We now have all the necessary ingredients to complete the proof of \Cref{instability to trapped surfaces proposition}, establishing the instability of locally naked singularity spacetimes to trapped surface formation arbitrarily close to $b_{\Gamma}$:
\begin{proof}[Proof of \Cref{instability to trapped surfaces proposition}]
    According to \Cref{curve C lemma}, the curve $\mathcal{C}_{\lambda}$ is contained in the region $\mathcal{P}_{\lambda}$ for all $|u|$ small enough. Thus, we fix $u_{\mathcal{C}}\in(u_0,0)$ such that $\big(u,v_{\sharp}(u)\big)\in\mathcal{P}_{\lambda}$ for all $u\in[u_{\mathcal{C}},0).$

    Our goal is to use \Cref{trapped surface formation theorem}. Firstly, we have by \eqref{B1 blueshift instability} that for any $u\in[u_{\mathcal{C}},0):$
    \begin{equation}\label{bound for delta lambda}
        \delta_{\lambda}\big(u,v_{\sharp}(u)\big)=\frac{r_{\lambda}\big(u,v_{\sharp}(u)\big)}{r_{\lambda}(u,0)}-1\lesssim v_{\sharp}(u)|u|^{-1}e^{-\mathbf{b}(u)}.
    \end{equation}
    Next, we want to obtain a lower bound for $\eta_{\lambda}\big(u,v_{\sharp}(u)\big).$ We use \eqref{characteristic data perturbation}, \eqref{blueshift difference bound}, \eqref{B4 blueshift instability}, and \eqref{background energy bound 1} to get that for $u\in[u_{\mathcal{C}},0):$
    \begin{align*}
        \lambda^2\big[v_{\sharp}(u)\big]^{2n+1}e^{2\mathbf{b}(u)}&\lesssim\int_0^{v_{\sharp}(u)}\Big|e^{\mathbf{b}_{\lambda}(u,v')}\widetilde{\theta}(u_0,v')\Big|^2dv'\\&\lesssim\int_0^{v_{\sharp}(u)}\big|\widetilde{\theta}(u,v')\big|^2dv'+\int_0^{v_{\sharp}(u)}\Big|\widetilde{\theta}(u,v')-e^{\mathbf{b}_{\lambda}(u,v')}\widetilde{\theta}(u_0,v')\Big|^2dv'\\
        &\lesssim\int_0^{v_{\sharp}(u)}\big|\theta_{\lambda}(u,v')\big|^2dv'+\int_0^{v_{\sharp}(u)}\big|\theta(u,v')\big|^2dv'+\lambda^2v^{2n+\frac{4}{3}}e^{2\mathbf{b}(u)}\\
        &\lesssim\int_0^{v_{\sharp}(u)}\big|\theta_{\lambda}(u,v')\big|^2dv'+v_{\sharp}(u)\cdot|u|\gamma(u,0)+\lambda^2\big[v_{\sharp}(u)\big]^{2n+\frac{4}{3}}e^{2\mathbf{b}(u)}.
    \end{align*}
    By the construction of $\mathcal{C}_{\lambda},$ we have that $v_{\sharp}(u)\cdot|u|\gamma(u,0)=\lambda^2\big[v_{\sharp}(u)\big]^{2n+1}e^{2\mathbf{b}(u)}\cdot|u|^{\frac{1}{5n}}$. As a result, there exists $u_{\mathcal{C}}'\in[u_{\mathcal{C}},0)$ such that for all $u\in[u_{\mathcal{C}}',0):$
    \begin{equation}\label{lower bound for theta lambda}
        \int_0^{v_{\sharp}(u)}\big|\theta_{\lambda}(u,v')\big|^2dv'\gtrsim\lambda^2\big[v_{\sharp}(u)\big]^{2n+1}e^{2\mathbf{b}(u)}.
    \end{equation}
    This lower bound confirms our expectation that the leading order term in $\theta_{\lambda}$ comes from the amplification of the initial data perturbation $\lambda v^n$ by the blueshift $e^{\mathbf{b}(u)}.$

    We use \eqref{dv m} and \eqref{lower bound for theta lambda} to derive a lower bound for $\eta_{\lambda}\big(u,v_{\sharp}(u)\big)$ for any $u\in[u_{\mathcal{C}}',0):$
    \begin{align*}
        \eta_{\lambda}\big(u,v_{\sharp}(u)\big)&=\frac{m_{\lambda}\big(u,v_{\sharp}(u)\big)-m_{\lambda}(u,0)}{|\Lambda|r_{\lambda}^2\big(u,v_{\sharp}(u)\big)}\gtrsim|u|^{-2}\int_0^{v_{\sharp}(u)}\gamma_{\lambda}^{-1}|\partial_ur_{\lambda}|\theta_{\lambda}^2\partial_vr_{\lambda}(u,v')dv'\\
        &\gtrsim|u|^{-2}e^{-\mathbf{b}(u)}\gamma^{-1}(u,0)\int_0^{v_{\sharp}(u)}\big|\theta_{\lambda}(u,v')\big|^2dv'\gtrsim\lambda^2\big[v_{\sharp}(u)\big]^{2n+1}|u|^{-2}e^{\mathbf{b}(u)}\gamma^{-1}(u,0).
    \end{align*}
    According to the definition of $\mathcal{C}_{\lambda}$ we further have that for all $u\in[u_{\mathcal{C}}',0):$
    \begin{equation}\label{lower bound for eta lambda}
        \eta_{\lambda}\big(u,v_{\sharp}(u)\big)\gtrsim v_{\sharp}(u)|u|^{-1-\frac{1}{5n}}e^{-\mathbf{b}(u)}.
    \end{equation}

    Combining \eqref{bound for delta lambda} and \eqref{lower bound for eta lambda}, we get that there exists $u_{\mathcal{C}}''\in[u_{\mathcal{C}}',0)$ such that for all $u\in[u_{\mathcal{C}}'',0)$ we have $\eta_{\lambda}\big(u,v_{\sharp}(u)\big)>\delta_{\lambda}\big(u,v_{\sharp}(u)\big)$ and $\delta_{\lambda}\big(u,v_{\sharp}(u)\big)>0$ is sufficiently small as required in \Cref{trapped surface formation theorem}. Thus, the trapped surface formation criterion in \Cref{trapped surface formation theorem} implies that for all $u\in[u_{\mathcal{C}}'',0)$ a trapped surface forms on the ingoing cone $\big\{v=v_{\sharp}(u)\big\}$ to the future of $\big(u,v_{\sharp}(u)\big).$
\end{proof}

\section{Proof of \Cref{main theorem}}\label{WCC section}

We combine the results in \Cref{set up section}, \Cref{instability section}, and Appendix~\ref{appendix}, to complete the proof of \Cref{main theorem}, establishing the weak cosmic censorship conjecture for circularly
symmetric solutions of \eqref{Einstein equations}:
\begin{proof}[Proof of \Cref{main theorem}]

    The most important part of the proof is showing \eqref{loc naked for main theorem}:
    \begin{equation}\label{M null unstable to M black}
        \mathfrak{M}_{\mathrm{loc.naked}}\subset\mathrm{cl}\big(\mathfrak{M}_{\mathrm{black}}^{\mathrm{spacelike}}\big)\backslash\mathfrak{M}_{\mathrm{black}}^{\mathrm{spacelike}},
    \end{equation}
    establishing the fact that all locally naked singularities are unstable to the formation of spacelike-only singularities.

    In order to prove \eqref{M null unstable to M black}, we consider any $\mathfrak{D}\in\mathfrak{M}_{\mathrm{loc.naked}}.$ For any $\epsilon>0,$ we prove that there exists $\mathfrak{D}'\in\mathfrak{M}_{\mathrm{black}}^{\mathrm{spacelike}}$ such that:
    \begin{equation}\label{D-D' less than epsilon}
        d_{X([0,\infty))}\big(\mathfrak{D},\mathfrak{D}'\big)<\epsilon.
    \end{equation}
    
    We denote by $\big(\mathcal{M},g,\phi\big)$ the maximal globally hyperbolic development of the asymptotically AdS characteristic data set $\mathfrak{D}$ on the initial outgoing null cone $C_0^+.$ Since $\mathfrak{D}\in\mathfrak{M}_{\mathrm{loc.naked}},$ we have by definition that the solution has a first singularity $b_{\Gamma}$ at the center $\Gamma.$ We denote by $C_0^-$ the ingoing null cone that ends at $b_{\Gamma}$. By global hyperbolicity, we have that either $C_0^-\cap C_0^+\neq\emptyset$ or $C_0^-$ has a limit endpoint on $\mathcal{I}.$ If the latter holds, we fix an outgoing null cone $C_1^+$ such that $C_0^-\cap C_1^+\neq\emptyset$ and $\overline{C_1^+}\cap\mathcal{I}\neq\emptyset$. On the other hand, if we already have $C_0^-\cap C_0^+\neq\emptyset$, we simply apply the argument below for $C_1^+=C_0^+.$

    We denote by $\mathfrak{D}_*$ the asymptotically AdS characteristic data set induced by $\big(\mathcal{M},g,\phi\big)$ on $C_1^+$. We note that we can also apply the results in \Cref{well posedness section} and \Cref{global structure section} towards the past. By Cauchy stability, there exists $\delta>0$ such that for any asymptotically AdS characteristic data set $\mathfrak{D}_*'$ on $C_1^+$ which satisfies:
    \begin{equation}\label{D1-D1' less than delta}
        d_{X([0,\infty))}\big(\mathfrak{D}_*,\mathfrak{D}_*'\big)<\delta,
    \end{equation}
    the maximal globally hyperbolic (past) development of $\mathfrak{D}_*'$ induces an asymptotically AdS characteristic data set $\mathfrak{D}'$ on $C_0^+$ which satisfies $d_{X([0,\infty))}\big(\mathfrak{D},\mathfrak{D}'\big)<\epsilon.$

    In view of the above argument, to prove \eqref{D-D' less than epsilon} it suffices to show that for any $\delta>0$ there exists $\mathfrak{D}_*'\in\mathfrak{M}_{\mathrm{black}}^{\mathrm{spacelike}}$ satisfying \eqref{D1-D1' less than delta}. The strategy is to use the results in \Cref{instability section} to construct locally a perturbation of $\mathfrak{D}_*$ which in evolution forms trapped surfaces arbitrarily close to $b_{\Gamma}$, and then extend this to a global perturbation of $\mathfrak{D}_*$ using the results in Appendix~\ref{appendix}.
    
    In order to use the results in \Cref{instability section}, we choose double null coordinates $(u,v)$ on $\mathcal{Q}=\mathcal{M}/U(1),$ such that: \[b_{\Gamma}=(0,0),\ C_0^-=\big\{v=0,\ 2r=-u\big\},\ C_1^+=\big\{u=u_1,\ 2r=v-u_1\big\}.\]
    Let $n>k+1$ be an integer. We fix $R_0>0$ large enough depending on $\mathfrak{D}_*$ and $k$, as required in \Cref{local perturbation extension proposition}, and we set $v_1=2R_0+u_1.$ We carry out the construction in \Cref{instability section}. For any $\lambda\in\mathbb{R}\backslash\{0\},$ we define characteristic data as in \eqref{characteristic data perturbation}, which determines uniquely on $\big\{u=u_1,\ v\in[u_1,v_1]\big\}$ a $C^k$ characteristic partial data set $\mathfrak{D}_{\lambda}$, as in \Cref{partial data set definition}, which coincides with $\mathfrak{D}_*$ for $\big\{u=u_1,\ v\in[u_1,0]\big\}.$ 
    
    We want to apply \Cref{local perturbation extension proposition} to extend $\mathfrak{D}_{\lambda}$ to $C_1^+$ as a small perturbation of $\mathfrak{D}_*.$ We point out that this is a nontrivial difficulty that is not present in the asymptotically flat setting. We denote by $C_0>0$ the constant in \Cref{local perturbation extension proposition}, which depends only on $\mathfrak{D}_*, R_0$ and $k.$ We choose $\varepsilon>0$ to be small enough in terms of $\mathfrak{D}_*,R_0,$ and $k,$ as required in \Cref{local perturbation extension proposition}, and to satisfy $C_0\varepsilon<\delta.$ 
    
    By uniform continuity in $\lambda$ on $[0,R_0],$ there exists $\lambda\neq0$ sufficiently small, such that:
    \begin{equation}\label{D-Dlamdba R0}
        d_{X([0,R_0])}\big(\mathfrak{D}_*,\mathfrak{D}_{\lambda}\big)<\varepsilon.
    \end{equation}
    We apply \Cref{local perturbation extension proposition} with $r_0=R_0.$ We get that $\mathfrak{D}_{\lambda}$ can be extended as a $C^k$ asymptotically AdS characteristic data set on $C_1^+$ such that:
    \begin{equation}\label{D-Dlamdba}
        d_{X([0,\infty))}\big(\mathfrak{D}_*,\mathfrak{D}_{\lambda}\big)<C_0\varepsilon<\delta.
    \end{equation}
    We relabel the extended characteristic data set $\mathfrak{D}_{\lambda}$ by $\mathfrak{D}'_*,$ so $\mathfrak{D}'_*$ satisfies \eqref{D1-D1' less than delta}. According to our construction, \Cref{instability to trapped surfaces proposition} implies that the maximal globally hyperbolic development of $\mathfrak{D}'_*$ contains trapped surfaces arbitrarily close to $b_{\Gamma}$, so $\mathfrak{D}'_*\in\mathfrak{M}_{\mathrm{black}}^{\mathrm{spacelike}}.$ This completes the proof of \eqref{M null unstable to M black}.

    We define $\mathfrak{M}_{\mathrm{generic}}=\mathfrak{M}_{\mathrm{black}}\sqcup\mathrm{int}\big(\mathfrak{M}_{\mathrm{non}}\big).$ According to \eqref{M null in M locally naked}, \eqref{decomposition of M}, and \eqref{M null unstable to M black}, we have that \eqref{inclusions for main theorem} holds. To complete the proof of \Cref{main theorem}, we show that $\mathfrak{M}_{\mathrm{generic}}$ is open and dense in $\mathfrak{M}$. According to \Cref{global structure section}, $\mathfrak{D}\in\mathfrak{M}_{\mathrm{black}}$ if and only if $\mathcal{T}\neq\emptyset.$ By Cauchy stability, having trapped surfaces is an open condition, so $\mathfrak{M}_{\mathrm{black}}$ is open. Since $\mathrm{int}\big(\mathfrak{M}_{\mathrm{non}}\big)$ is open by definition, we get that $\mathfrak{M}_{\mathrm{generic}}$ is open in $\mathfrak{M}.$ Finally, we prove that $\mathfrak{M}_{\mathrm{generic}}$ is dense in $\mathfrak{M}$ by showing that $\mathrm{cl}\big(\mathfrak{M}_{\mathrm{generic}}\big)=\mathfrak{M}.$ By the definition of $\mathfrak{M}_{\mathrm{generic}}$, we have that $\mathrm{cl}\big(\mathfrak{M}_{\mathrm{generic}}\big)=\mathrm{cl}\big(\mathfrak{M}_{\mathrm{black}}\big)\cup\mathrm{cl}\big(\mathrm{int}\big(\mathfrak{M}_{\mathrm{non}}\big)\big).$ Then, \eqref{inclusions for main theorem} implies that:
    \[\mathfrak{M}_{\mathrm{naked}}\cup\mathfrak{M}_{\mathrm{black}}\cup\mathrm{int}\big(\mathfrak{M}_{\mathrm{non}}\big)\subset\mathrm{cl}\big(\mathfrak{M}_{\mathrm{generic}}\big).\]
    This equation and \eqref{decomposition of M} also imply that: \[\mathfrak{M}_{\mathrm{non}}\backslash\mathrm{int}\big(\mathfrak{M}_{\mathrm{non}}\big)\subset\mathrm{cl}\big(\mathfrak{M}\backslash\mathfrak{M}_{\mathrm{non}}\big)=\mathrm{cl}\big(\mathfrak{M}_{\mathrm{naked}}\cup\mathfrak{M}_{\mathrm{black}}\big)\subset\mathrm{cl}\big(\mathfrak{M}_{\mathrm{generic}}\big),\]
    so we conclude that $\mathfrak{M}_{\mathrm{generic}}$ is dense in $\mathfrak{M}$, completing the proof of \Cref{main theorem}.
\end{proof}

\begin{remark}
    Since \eqref{loc naked for main theorem} shows that locally naked singularities are unstable to the formation of spacelike-only singularities, one might wonder whether $\mathfrak{M}_{\mathrm{loc.naked}}$ is non-generic as well. However, it is not clear if $\mathfrak{M}_{\mathrm{black}}^{\mathrm{spacelike}}$ is open or if it even has a nonempty interior.
\end{remark}

\appendix
\section{Appendix}\label{appendix}
The goal of this section is to prove that local perturbations at the level of initial data can be extended to global perturbations. While this is straightforward for asymptotically flat characteristic data, the proof in the case of asymptotically AdS characteristic data is more involved due to the timelike asymptotic boundary $\mathcal{I}.$

We first define the notion of a characteristic partial data set:
\begin{definition}\label{partial data set definition}
    Let $k\geq2$. For any $r_0>0,$ we say that the tuple of functions on $[0,r_0]$:
    \[\mathfrak{D}=\big(r,\ldots,\partial_U^{k+1}r,\Omega,\ldots,\partial_U^k\Omega,m,\ldots,\partial_U^km,\phi,\ldots,\partial_U^k\phi,\mathcal{T}\phi,\ldots,\mathcal{T}^k\phi\big)\]
    represents a $C^k$ characteristic \textit{partial} data set if it satisfies conditions $1.$ to $4.$ in \Cref{asympt ads data def} on the interval $\{r\in[0,r_0]\}$.
\end{definition}

In this section, we establish the following extension result for characteristic partial data sets, which plays a key role in the proof of \Cref{main theorem} in \Cref{WCC section}.

\begin{proposition}\label{local perturbation extension proposition}
    Let $k\geq2$ and let  $\mathfrak{D}$ be a $C^k$ asymptotically AdS characteristic data set. There exists $R_0>0$ large enough (depending on $\mathfrak{D},\Lambda,$ and $k$), such that for any $r_0\geq R_0$ and any $\varepsilon>0$ sufficiently small in terms of $\mathfrak{D},r_0,\Lambda,$ and $k$, the following holds: any $C^k$ characteristic partial data set on $[0,r_0]$: \[\underline{\mathfrak{D}}=\big(\underline{r},\ldots,\underline{\partial_U^{k+1}r},\underline{\Omega},\ldots,\underline{\partial_U^k\Omega},\underline{m},\ldots,\underline{\partial_U^km},\underline{\phi},\ldots,\underline{\partial_U^k\phi},\underline{\mathcal{T}\phi},\ldots,\underline{\mathcal{T}^k\phi}\big),\]  
    which satisfies:
    \begin{equation}\label{smallness in a finite region}
        d_{X([0,r_0])}\big(\mathfrak{D},\underline{\mathfrak{D}}\big)<\varepsilon,
    \end{equation}
    can be extended to an asymptotically AdS characteristic data set on $[0,\infty)$ such that:
    \begin{equation}\label{smallness of extension}
        d_{X([0,\infty))}\big(\mathfrak{D},\underline{\mathfrak{D}}\big)<C_0\varepsilon,
    \end{equation}
    where the constant $C_0>0$ depends only on $\mathfrak{D},r_0,\Lambda,$ and $k$. 
\end{proposition}

\begin{remark}
    We briefly explain the connection of \Cref{local perturbation extension proposition} to the characteristic gluing literature. The problem of characteristic gluing consists of constructing characteristic data on a truncated null hypersurface connecting prescribed data on two codimension-2 spheres. This was initiated in the series of works \cite{gluing1,gluing2,gluing3} for the Einstein vacuum equations in a perturbative regime around Minkowski space, up to linear obstructions (see also \cite{Lgluing1%,Lgluing2,Lgluing3,Lgluing4
    } for the non-vanishing cosmological constant case). These results were later extended to the obstruction-free case in \cite{gluing4}. More recently, the works \cite{gluing5,gluing6} proved non-perturbative characteristic gluing results. \Cref{local perturbation extension proposition} can be interpreted as a characteristic gluing result, in the sense that we glue a characteristic partial data set $\underline{\mathfrak{D}}$ to an asymptotic BTZ sphere at~$\mathcal{I}.$ We note that our result is perturbative around any given asymptotically AdS characteristic data set~$\mathfrak{D}$. Moreover, we expect that our methods can also be used to prove that $\underline{\mathfrak{D}}$ can be glued to a finite BTZ sphere. We also note that in contrast to the characteristic gluing results referenced above, which prescribe seed data at zeroth-order, we prescribe data at the level of the top-order transversal derivative of $\phi.$
\end{remark}

We prove this result for the rest of the section. For the sake of exposition, we restrict to the case $k=2$ for the proof, as it is clear that our methods can be adapted to the general case as explained in \Cref{remark on general k}. We briefly outline the idea of the proof and the structure of the Appendix.

The main obstructions to extending $\underline{\mathfrak{D}}$ are given by the compatibility conditions in \Cref{asympt ads data def} and the global regularity requirements \eqref{regularity conditions scalar field L2}. In order to enforce the needed compatibility conditions, we want to extend $\underline{\mathcal{T}^2\phi}$ to $(r_0,\infty),$ then compute the remaining components of $\underline{\mathfrak{D}}$ uniquely by solving transport equations with initial data at $r_0$. We set up the system of null constraint equations in \Cref{null constraint section}. We point out that except for the lapse and its derivatives, the elements of $\underline{\mathfrak{D}}$ are independent of the choice of a $V$ coordinate, so we carry out most of the argument using the coordinate $r$. However, it is clear that not every choice of $\underline{\mathcal{T}^2\phi}\in H^1((r_0,\infty))$ will lead to $\underline{\phi}\in H^3([0,\infty))$ and $\underline{\mathcal{T}\phi}\in H^2([0,\infty)).$ For this to hold, we need certain vanishing conditions which guarantee the suitable decay of $\underline{\phi}$ and $\underline{\mathcal{T}\phi},$ which we prove in \Cref{vanishing section}. The idea is that the vanishing conditions hold for $\mathfrak{D},$ so we prove the existence of a suitable $\underline{\mathcal{T}^2\phi}$ as a perturbation of $\mathcal{T}^2\phi$ using the Implicit Function Theorem. Finally, we prove in \Cref{extending local perturbations section} that under the vanishing conditions of \Cref{vanishing section}, the smallness bound \eqref{smallness of extension} holds.

\subsection{Null constraint equations}\label{null constraint section}

In this section, we derive the system of null constraint equations \eqref{dr theta}-\eqref{dr T phi} satisfied by the elements of a characteristic data set $\mathfrak{D}$ and we set up these equations as a system of ODEs \eqref{dr A_1}-\eqref{ODE system data}.

We define the renormalized outgoing derivative:
\[\partial_r=\frac{\partial_V}{\partial_Vr},\]
which is independent of the choice of the outgoing  coordinate $V.$ We recall the definition of $\gamma$ in \eqref{definition of gamma} and we express the Kodama vector field $\mathcal{T}$ defined in \eqref{definition of T} as:
\[\mathcal{T}=\Omega^{-2}\partial_Vr\partial_U-\Omega^{-2}\partial_Ur\partial_V=\frac{1}{\gamma}\partial_U-\frac{1}{\gamma}\partial_Ur\partial_r.\]

Firstly, we notice that equations \eqref{wave equation phi}, \eqref{Ray v}, and \eqref{wave equation r} are equivalent to:
\begin{align}
    \partial_r\theta+\bigg(\frac{1}{2r}+\frac{\Lambda\gamma r}{2\partial_Ur}\bigg)\theta=&\frac{\partial_r\big(\sqrt{r}\gamma\mathcal{T\phi}\big)}{-\partial_Ur}\label{dr theta}\\
    \partial_r\log\gamma=&\theta^2\label{dr gamma}\\
    \partial_r\partial_Ur=&\frac{\Lambda}{2}\gamma r.\label{dr du r}
\end{align}
In order to derive the higher order equations, we compute the following commutation relation:
\begin{equation}\label{commutation of T and dr}
    \big[\mathcal{T},\partial_r\big]=\bigg[\mathcal{T},\frac{\partial_V}{\partial_Vr}\bigg]=\theta^2\mathcal{T}.
\end{equation}
Commuting equations \eqref{dr theta}-\eqref{dr du r} with $\mathcal{T},$ we obtain the equations:
\begin{align}
    \partial_r\mathcal{T}\theta+\bigg(\frac{1}{2r}+\frac{\Lambda\gamma r}{2\partial_Ur}+\theta^2\bigg)\mathcal{T}\theta=&\frac{\partial_r\mathcal{T}\big(\sqrt{r}\gamma\mathcal{T}\phi\big)}{-\partial_Ur}+\frac{\theta^2\mathcal{T}\big(\sqrt{r}\gamma\mathcal{T}\phi\big)}{-\partial_Ur}+\frac{\mathcal{T}\partial_Ur\cdot\partial_r\big(\sqrt{r}\gamma\mathcal{T\phi}\big)}{(\partial_Ur)^2}\label{dr T theta}\\
    &+\frac{\Lambda}{2}r\theta\frac{\mathcal{T}\gamma}{-\partial_Ur}+\frac{\Lambda}{2}r\theta\frac{\gamma\mathcal{T}\partial_Ur}{(\partial_Ur)^2},\notag\\
    \partial_r\mathcal{T}\gamma=&2\gamma\theta\mathcal{T}\theta,\label{dr T gamma}\\
    \partial_r\mathcal{T}\partial_Ur=& -\theta^2\mathcal{T}\partial_Ur+\frac{\Lambda}{2}r\mathcal{T}\gamma.\label{dr T du r}
\end{align}
Moreover, we also compute that:
\begin{equation}\label{dr T phi}
    \partial_r\big(\sqrt{r}\gamma\mathcal{T\phi}\big)=\gamma\mathcal{T}\theta+\frac{\gamma\mathcal{T}\phi}{2\sqrt{r}},\ \partial_r\big(\gamma\mathcal{T}\phi\big)=\frac{\gamma\mathcal{T}\theta}{\sqrt{r}}.
\end{equation}

We refer to equations \eqref{dr theta}-\eqref{dr T phi} as the \textit{null constraint equations}. We note that one can derive equations with similar structure for any order $k\geq2,$ see \Cref{remark on general k}. 

\begin{remark}\label{solving compatibility conditions remark}
    Given a tuple $\mathfrak{D}$ at $r_0,$ solutions of the null constraint equations on $[r_0,\infty)$ determine the tuple $\mathfrak{D}$ uniquely on $[r_0,\infty)$ (up to the choice of a coordinate $V$); moreover, the elements of $\mathfrak{D}$ satisfy the compatibility conditions in \Cref{asympt ads data def}. To see this, we note that equations \eqref{dr theta}-\eqref{dr du r} imply \eqref{wave equation phi}, \eqref{wave equation r}, and \eqref{Ray v}. Also, equations \eqref{dr theta}-\eqref{dr du r}, \eqref{dr T gamma}, and \eqref{dr T phi} imply \eqref{wave equation Omega}, while using \eqref{dr T du r} as well we obtain \eqref{Ray u}. At higher order, we notice that \eqref{dr T theta} implies the equation obtained by taking a $\partial_U$ derivative of \eqref{wave equation phi}. Finally, the quantities $\partial_U^2\Omega$ and $\partial_U^3r$ are not computed directly from the solutions of the null constraints, so we define them using the equations obtained by taking a $\partial_U$ derivative of \eqref{wave equation Omega} and \eqref{Ray u}.
\end{remark}

In view of \Cref{solving compatibility conditions remark}, the goal in \Cref{local perturbation extension proposition} is to construct solutions of the null constraint equations \eqref{dr theta}-\eqref{dr T phi} on $(r_0,\infty).$ We notice that all the quantities are determined at $r_0$ by the partial data set on $[0,r_0]$. Moreover, the only quantity that appears on the RHS in \eqref{dr theta}-\eqref{dr T phi} but does not satisfy a $\partial_r$ equation itself is the top order quantity $\partial_r\mathcal{T}\big(\sqrt{r}\gamma\mathcal{T}\phi\big).$ Thus, in order to construct solutions of the null constraint equations, we prescribe $\partial_r\mathcal{T}\big(\sqrt{r}\gamma\mathcal{T}\phi\big)$ on $(r_0,\infty)$ and solve equations \eqref{dr theta}-\eqref{dr T phi} as a system of ODEs with initial data at $r_0.$

Our discussion motivates us to study separately the following system of ODEs:
\begin{align}
    \partial_rA_1+\bigg(\frac{1}{2r}+\frac{\Lambda B_0 r}{2D_0}+A_0^2\bigg)A_1=&\frac{F_1}{-D_0}+\frac{A_0^2P_2}{-D_0}+\frac{D_1F_0}{D_0^2}\label{dr A_1}\\&+\frac{\Lambda}{2}r\frac{A_0B_1}{-D_0}+\frac{\Lambda}{2}r\frac{A_0B_0D_1}{D_0^2},\notag\\
    \partial_rB_1=&2B_0 A_0 A_1,\label{dr B_1}\\
    \partial_rD_1=& -A_0^2D_1+\frac{\Lambda}{2}rB_1,\label{dr D_1}\\
    \partial_rP_2=&F_1,\label{dr P_2}\\
    \partial_rA_0+\bigg(\frac{1}{2r}+\frac{\Lambda B_0 r}{2D_0}\bigg)A_0=& \frac{F_0}{-D_0},\label{dr A_0}\\
    F_0=&B_0A_1+\frac{P_1}{2\sqrt{r}},\label{F_0}\\
    \partial_rB_0=& A_0^2B_0,\label{dr B_0}\\
    \partial_rD_0=& \frac{\Lambda}{2}B_0 r,\label{dr D_0}\\
    \partial_rP_1=&\frac{B_0A_1}{\sqrt{r}},\label{dr P_1}\\
    A_i(r_0)=a_i,\ B_i(r_0)=b_i,\ D_i(r_0)=&d_i,\ P_{i+1}(r_0)=p_{i+1},\ i=0,1.\label{ODE system data}
\end{align}
We study the system \eqref{dr A_1}-\eqref{ODE system data} in detail in \Cref{vanishing section}. We also point out that the null constraint equations \eqref{dr theta}-\eqref{dr T phi} are equivalent to the system \eqref{dr A_1}-\eqref{ODE system data} for:
\[A_0=\theta,A_1=\mathcal{T}\theta,B_0=\gamma,B_1=\mathcal{T}\gamma,D_0=\partial_Ur,D_1=\mathcal{T}\partial_Ur,\]\[P_1=\gamma\mathcal{T}\phi, P_2=\mathcal{T}(\sqrt{r}\gamma\mathcal{T\phi}),\ F_0=\partial_r(\sqrt{r}\gamma\mathcal{T\phi}),\ F_1=\partial_r\mathcal{T}(\sqrt{r}\gamma\mathcal{T\phi}).\]

\subsection{Vanishing conditions}\label{vanishing section}

According to \Cref{null constraint section}, the strategy to construct characteristic data sets is to prescribe $\partial_r\mathcal{T}\big(\sqrt{r}\gamma\mathcal{T}\phi\big)$ on $(r_0,\infty)$ and treat the null constraint equations as a system of ODEs. However, prescribing $\partial_r\mathcal{T}\big(\sqrt{r}\gamma\mathcal{T}\phi\big)$ freely on $(r_0,\infty)$ will not lead in general to $\phi,\mathcal{T}\phi,$ and $\mathcal{T}^2\phi$ decaying sufficiently fast, unless additional vanishing conditions hold. We know these hold for the elements of an asymptotically AdS characteristic data set $\mathfrak{D}.$ More specifically, we have that \Cref{asympt ads data def} (more specifically, equation \eqref{regularity conditions scalar field L2}) and equations \eqref{dr theta}-\eqref{dr T phi} imply:
\begin{equation}\label{L2 condition for F1}
    \sqrt{r}\partial_r\mathcal{T}\big(\sqrt{r}\gamma\mathcal{T\phi}\big)\in L^2([r_0,\infty)),
\end{equation}
\begin{equation}\label{vanishing conditions for D}
    \int_{r_0}^{\infty}\partial_r\mathcal{T}\big(\sqrt{r}\gamma\mathcal{T\phi}\big)dr=-\mathcal{T}\big(\sqrt{r}\gamma\mathcal{T\phi}\big)(r_0),\ \int_{r_0}^{\infty}\frac{\gamma\mathcal{T}\theta}{\sqrt{r}}dr=-\gamma\mathcal{T\phi}(r_0),\ \int_{r_0}^{\infty}\frac{\theta}{\sqrt{r}}dr=-\phi(r_0),
\end{equation}
where we also used the fact that $\mathcal{T}\gamma$ is uniformly bounded due to \eqref{dr T gamma} and \eqref{regularity conditions scalar field L2}.

\begin{remark}\label{nonlocal nonlinear remark}
    One could take an alternative approach and prescribe instead the seed data $\phi$ on $(r_0,\infty)$, similarly to the asymptotically flat case of \cite{ChrWCC}, see also \Cref{differences from HS remark}. Using the null constraint equations \eqref{dr theta}-\eqref{dr T phi}, one recovers all the other quantities in the system in terms of $\phi$. However, the vanishing conditions \eqref{vanishing conditions for D} represent nonlinear and non-local constraints for $\phi,$ so we do not expect to be able to express the space of seed data as a vector space. We choose not to pursue this approach, and instead prescribe $\partial_r\mathcal{T}\big(\sqrt{r}\gamma\mathcal{T}\phi\big)$ on $(r_0,\infty)$ as outlined in our above strategy.
\end{remark}

Given an asymptotically AdS data set $\mathfrak{D},$ our goal is to prove that there exists $\partial_r\underline{\mathcal{T}}\big(\sqrt{r}\underline{\gamma\mathcal{T}\phi}\big)$ such that:
\begin{equation}\label{vanishing conditions F1 bar L2 conditions}
    \sqrt{r}\partial_r\underline{\mathcal{T}}\big(\sqrt{r}\underline{\gamma\mathcal{T}\phi}\big)\in L^2([r_0,\infty)),\ \big\|\sqrt{r}\partial_r\underline{\mathcal{T}}\big(\sqrt{r}\underline{\gamma\mathcal{T}\phi}\big)-\sqrt{r}\partial_r\mathcal{T}\big(\sqrt{r}\gamma\mathcal{T\phi}\big)\big\|_{L^2([r_0,\infty))}<C_0\varepsilon,
\end{equation}
\begin{equation}\label{vanishing conditions F1 bar}
    \int_{r_0}^{\infty}\partial_r\underline{\mathcal{T}}\big(\sqrt{r}\underline{\gamma\mathcal{T}\phi}\big)dr=-\underline{\mathcal{T}}\big(\sqrt{r}\underline{\gamma}\underline{\mathcal{T}\phi}\big)(r_0),\ \int_{r_0}^{\infty}\frac{\underline{\gamma \mathcal{T}\theta}}{\sqrt{r}}dr=-\underline{\gamma}\underline{\mathcal{T}\phi}(r_0),\ \int_{r_0}^{\infty}\frac{\underline{\theta}}{\sqrt{r}}dr=-\underline{\phi}(r_0),
\end{equation}
where the constant $C_0>0$ is independent of $\varepsilon$ and $\partial_r\underline{\mathcal{T}}\big(\sqrt{r}\underline{\gamma\mathcal{T}\phi}\big)$.

We prove for the rest of the section the existence of $\partial_r\underline{\mathcal{T}}\big(\sqrt{r}\underline{\gamma\mathcal{T}\phi}\big)$ satisfying \eqref{vanishing conditions F1 bar L2 conditions}-\eqref{vanishing conditions F1 bar}. This requires proving a series of results for the ODE system \eqref{dr A_1}-\eqref{ODE system data}. We show that the system \eqref{dr A_1}-\eqref{ODE system data} has global solutions in \Cref{S is well defined lemma}, and that the solution map is $C^1$ in \Cref{dot S is well defined lemma}. We then prove a non-degeneracy condition in \Cref{dC is surjective lemma}, which allow us to use the Implicit Function Theorem in \Cref{IFT corollary} and conclude the existence of a suitable $\partial_r\underline{\mathcal{T}}\big(\sqrt{r}\underline{\gamma\mathcal{T}\phi}\big)$ in \Cref{existence of F1 proposition}.
    
We first define the suitable weighted spaces used in our argument. We fix a constant $M>0$ and a large constant $R_0>0$ (to be determined in the course of the proof). For any $r_0\geq R_0,$ we define:
\begin{align*}
    TY=\bigg\{F:[r_0,\infty)\rightarrow\mathbb{R}:\ &\|F\|_{TY}^2=\int_{r_0}^{\infty}r^2F^2(r)dr<\infty\bigg\},\\
    Y=\bigg\{F\in TY:\ &\|F\|_Y=\|F\|_{TY}<M\bigg\},\\
    Z_1=\bigg\{A:[r_0,\infty)\rightarrow\mathbb{R}:\ &\|A\|_{Z_1}=\sup_{[r_0,\infty)}r^{3/2}|A|<\infty\bigg\},\\
    Z_2=\bigg\{B:[r_0,\infty)\rightarrow\mathbb{R}:\ &\|B\|_{Z_2}=\sup_{[r_0,\infty)}|B|<\infty\bigg\},\\
    Z_3=\bigg\{D:[r_0,\infty)\rightarrow\mathbb{R}:\ &\|D\|_{Z_2}=\sup_{[r_0,\infty)}r^{-2}|D|<\infty\bigg\},\\
    Z_4=\bigg\{P:[r_0,\infty)\rightarrow\mathbb{R}:\ &\|P\|_{Z_4}=\sup_{[r_0,\infty)}|P|<\infty\bigg\},\\
    Z=Z_1^2\times Z_2^2&\times Z_3^2\times Z_4^2,\\
    \|w\|_W=r_0^{3/2}|a_0|+r_0^{3/2}|a_1|+|b_0|+&|b_1|+r_0^{-2}|d_0|+r_0^{-2}|d_1|+r_0|p_1|+r_0^{1/2}|p_2|,\\
    W=\bigg\{w=(a_0,a_1,b_0,b_1,d_0,d_1,p_1,p_2)\in\mathbb{R}^8:\ &|b_0-2|<1,\ \bigg|d_0-\frac{\Lambda}{2}r_0^2\bigg|<\frac{|\Lambda|}{4}r_0^2,\ \|w\|_W<M\bigg\}.
\end{align*}

We prove that the system \eqref{dr A_1}-\eqref{ODE system data} has global solutions:
\begin{lemma}\label{S is well defined lemma}
        The solution map of the system \eqref{dr A_1}-\eqref{ODE system data} given by $\mathcal{S}:W\times Y\rightarrow Z$: \[(w,F_1)=(a_0,a_1,b_0,b_1,d_0,d_1,p_1,p_2,F_1)\mapsto(A_0,A_1,B_0,B_1,D_0,D_1,P_1,P_2)\] is well-defined.
\end{lemma}
    \begin{proof}
        Since $F_1\in Y$, equation \eqref{dr P_2} implies:
        \[|P_2|\leq|p_2|+r_0^{-1/2}\|F_1\|_Y<2r_0^{-1/2}M.\]
        We make the following bootstrap assumption for some $R>r_0$:
        \begin{equation}\label{boostrap assumption A0 A1}
            \int_{r_0}^RA_0^2(r)+rA_1^2(r)dr<\log2.
        \end{equation}
        As a consequence, we get from \eqref{dr B_0}, \eqref{dr D_0}, and \eqref{dr P_1} that for all $r\in[r_0,R]$:
        \begin{equation}\label{bounds for B0 D0 P1}
            b_0\leq B_0(r)\leq 2b_0,\ \frac{|\Lambda|}{4}r^2<-D_0(r)<4|\Lambda|r^2,\ |P_1|<8r_0^{-1/2}.
        \end{equation}
        We integrate the equation \eqref{dr A_0} for $A_0:$
        \begin{equation}\label{integrated eq for A0}
            A_0(r)=\frac{a_0d_0\sqrt{r_0}}{D_0\sqrt{r}}+\frac{1}{-D_0\sqrt{r}}\int_{r_0}^r\sqrt{r'}F_0(r')dr'.
        \end{equation}
        Using \eqref{F_0}, \eqref{boostrap assumption A0 A1}, and \eqref{bounds for B0 D0 P1}, we get that \eqref{integrated eq for A0} implies for some constant $C(\Lambda)>0$:
        \begin{equation}\label{pointwise bound for A0}
            |A_0|<8Mr^{-5/2}r_0+64|\Lambda|^{-1}r^{-3/2}r_0^{-1/2}\leq (M+1)C(\Lambda)r^{-3/2}.
        \end{equation}
        As a result, we get from \eqref{dr B_1} and \eqref{dr D_1} that for $R_0$ sufficiently large in terms of $\Lambda$:
        \[|B_1|<2(M+1),\ |D_1|<2(M+1)(1+|\Lambda|)r^2.\]
        Similarly to the derivation of \eqref{integrated eq for A0}, we integrate the equation \eqref{dr A_1} for $A_1:$
        \begin{align*}
            A_1(r)=&\frac{a_1d_0\sqrt{r_0}}{D_0\sqrt{r}}\cdot\exp\bigg(-\int_{r_0}^rA_0^2dr'\bigg)+\frac{1}{-D_0\sqrt{r}}\int_{r_0}^r\sqrt{r'}F_1(r')\exp\bigg(-\int_{r'}^rA_0^2dr''\bigg)dr'\\
            &+\frac{1}{-D_0\sqrt{r}}\int_{r_0}^r\sqrt{r'}\bigg[A_0^2P_2+\frac{D_1F_0}{-D_0}+\frac{\Lambda}{2}r'A_0B_1+\frac{\Lambda}{2}r'\frac{A_0B_0D_1}{-D_0}\bigg]\exp\bigg(-\int_{r'}^rA_0^2dr''\bigg)dr'.
        \end{align*}
        Using the previous pointwise bounds and the bootstrap assumption \eqref{boostrap assumption A0 A1}, we obtain for some constant $C(\Lambda,M)>0$:
        \begin{equation}\label{pointwise bound for A1}
            |A_1|\leq C(\Lambda,M)r^{-3/2}.
        \end{equation}
        
        The pointwise bounds \eqref{pointwise bound for A0} and \eqref{pointwise bound for A1} allow us to improve the bootstrap assumption \eqref{boostrap assumption A0 A1}, provided that $R_0$ is sufficiently large in terms of $\Lambda$ and $M$. Thus, the ODE system \eqref{dr A_1}-\eqref{ODE system data} has global solutions. Moreover, the proof implies that $A_0,A_1\in Z_1,$ $B_0,B_1\in Z_2,$ $D_0,D_1\in Z_3,$ and $P_1,P_2\in Z_4,$ so the map $\mathcal{S}:W\times Y\rightarrow Z$ is well-defined.
    \end{proof}

We prove that the map $\mathcal{S}:W\times Y\rightarrow Z$ is $C^1.$ Let $(A_0,A_1,B_0,B_1,D_0,D_1,P_1,P_2)=\mathcal{S}(w,F_1).$ We consider the system of ODEs satisfied by the formal Fréchet derivatives of $A_0,\ldots,P_2$ with respect to $a_0,\ldots,p_2,$ and $F_1$:
%\partial_r\dot{A}_1+\bigg(\frac{1}{2r}+\frac{\Lambda B_0 r}{2D_0}+A_0^2\bigg)\dot{A}_1=&\frac{\dot{F}_1}{-D_0}+\frac{F_1\dot{D}_0}{D_0^2}+\frac{A_0^2\dot{P}_2}{-D_0}+\frac{2A_0\dot{A}_0P_2}{-D_0}+\frac{A_0^2P_2\dot{D}_0}{D_0^2}+\frac{\dot{D}_1F_0}{D_0^2}\label{dr dot A1}\\&+\frac{D_1\dot{F}_0}{D_0^2}-\frac{2D_1F_0\dot{D}_0}{D_0^3}-\frac{\Lambda}{2}r\bigg(\frac{\dot{A}_0B_1}{D_0}+\frac{A_0\dot{B}_1}{D_0}-\frac{A_0B_1\dot{D}_0}{D_0^2}\bigg)\notag\\&+\frac{\Lambda}{2}r\bigg(\frac{\dot{A}_0B_0D_1}{D_0^2}+\frac{A_0\dot{B}_0D_1}{D_0^2}+\frac{A_0B_0\dot{D}_1}{D_0^2}-\frac{2A_0B_0D_1\dot{D}_0}{D_0^3}\bigg)\notag\\ &-\frac{\Lambda A_1 \dot{B}_0 r}{2D_0}+\frac{\Lambda A_1 B_0 \dot{D}_0 r}{2D_0^2}-2A_0A_1\dot{A}_0,\notag\\
\begin{align}
    \partial_r\dot{A}_1+\bigg(\frac{1}{2r}+\frac{\Lambda B_0 r}{2D_0}+A_0^2\bigg)\dot{A}_1=&\frac{\dot{F}_1}{-D_0}+\frac{A_0^2\dot{P}_2}{-D_0}+\frac{D_1\dot{F}_0}{D_0^2}+\dot{A}_0\frac{\Lambda r}{2}\bigg(\frac{B_0D_1}{D_0^2}-\frac{B_1}{D_0}\bigg)+\label{dr dot A1}\\&+2\dot{A}_0\bigg(\frac{A_0P_2}{-D_0}-A_0A_1\bigg)+\dot{B}_0\frac{\Lambda r}{2}\bigg(\frac{A_0D_1}{D_0^2}-\frac{A_1}{D_0}\bigg)-\dot{B}_1\frac{\Lambda r}{2}\frac{A_0}{D_0}\notag\\&+\dot{D}_0\frac{\Lambda r}{2}\bigg(\frac{A_1 B_0}{D_0^2}-\frac{2A_0B_0D_1}{D_0^3}+\frac{A_0B_1}{D_0^2}\bigg)+\dot{D}_1\frac{\Lambda r}{2}\frac{A_0B_0}{D_0^2}\notag\\
    &+\dot{D}_0\bigg(\frac{F_1}{D_0^2}+\frac{A_0^2P_2}{D_0^2}-\frac{2D_1F_0}{D_0^3}\bigg)+\dot{D}_1\frac{F_0}{D_0^2},\notag\\
    \partial_r\dot{B}_1=&2\dot{B}_0 A_0 A_1+2B_0\dot{A}_0 A_1+2B_0 A_0\dot{A}_1,\label{dr dot B1}\\
    \partial_r\dot{D}_1=& -A_0^2\dot{D}_1-2A_0D_1\dot{A}_0+\frac{\Lambda}{2}r\dot{B}_1,\label{dr dot D1}\\
    \partial_r\dot{P}_2=&\dot{F}_1,\label{dr dot P2}\\
    \partial_r\dot{A}_0+\bigg(\frac{1}{2r}+\frac{\Lambda B_0 r}{2D_0}\bigg)\dot{A}_0=&\frac{\dot{F}_0}{-D_0}+\frac{F_0\dot{D}_0}{D_0^2}-\frac{\Lambda A_0 \dot{B}_0 r}{2D_0}+\frac{\Lambda A_0 B_0 \dot{D}_0 r}{2D_0^2}\label{dr dot A0}\\
    \dot{F}_0=&\dot{B}_0A_1+B_0\dot{A}_1+\frac{\dot{P}_1}{2\sqrt{r}},\label{dot F0}\\
    \partial_r\dot{B}_0=&\dot{B}_0A_0^2+2A_0B_0\dot{A}_0\label{dr dot B0}\\
    \partial_r\dot{D}_0=&\frac{\Lambda}{2}\dot{B}_0 r\label{dr dot D0}\\
    \partial_r\dot{P}_1=&\frac{\dot{B}_0A_1}{\sqrt{r}}+\frac{B_0\dot{A}_1}{\sqrt{r}},\label{dr dot P1}\\
    \dot{A}_i(r_0)=\dot{a}_i,\ \dot{B}_i(r_0)=&\dot{b}_i,\ \dot{D}_i(r_0)=\dot{d}_i,\ \dot{P}_{i+1}(r_0)=\dot{p}_{i+1},\ i=0,1.\label{dot data}
\end{align}
\begin{lemma}\label{dot S is well defined lemma}
    The solution map of the system \eqref{dr dot A1}-\eqref{dot data} given by $\dot{\mathcal{S}}_{(w,F_1)}:\mathbb{R}^8\times TY\rightarrow Z$:
    \[(\dot{a}_0,\dot{a}_1,\dot{b}_0,\dot{b}_1,\dot{d}_0,\dot{d}_1,\dot{p}_1,\dot{p}_2,\dot{F}_1)\mapsto(\dot{A}_0,\dot{A}_1,\dot{B}_0,\dot{B}_1,\dot{D}_0,\dot{D}_1,\dot{P}_1,\dot{P}_2)\] 
    is a well-defined bounded linear functional. Moreover, the map $\dot{\mathcal{S}}_{(w,F_1)}$ is continuous in $(w,F_1)$ and $\dot{\mathcal{S}}_{(w,F_1)}=d\mathcal{S}_{(w,F_1)}$. Therefore, the map $\mathcal{S}:W\times Y\rightarrow Z$ is $C^1.$
\end{lemma}
\begin{proof}
    We make the notation convention for the purpose of this proof that we write $x\lesssim y$ for any quantities $x,y>0$ if there is a constant $C(\Lambda,M)>0$ such that $x\leq Cy.$ We first have that \eqref{dr dot P2} implies the bound:
    \[|\dot{P}_2|\leq|\dot{p}_2|+r_0^{-1/2}\|\dot{F}_1\|_{TY}.\]
    
    We make the following bootstrap assumption, for some constant $K>0$ to be determined in the course of the proof:
    \begin{equation}\label{boostrap assumption dot A0 A1}
        \int_{r_0}^R\dot{A}_0^2(r)+r\dot{A}_1^2(r)dr<K^2.
    \end{equation}
    As a consequence, we get from \eqref{dr dot B0}, \eqref{dr dot D0}, \eqref{dr dot P1}, and \eqref{dot F0} that for all $r\in[r_0,R]$:
    \[|\dot{B}_0|\lesssim |\dot{b}_0|+Kr_0^{-1},\ |\dot{D}_0|\lesssim|\dot{d}_0|+|\dot{b}_0|r^2+Kr^2r_0^{-1},\]\[|\dot{P}_1|\lesssim|\dot{p}_1|+|\dot{b}_0|+Kr_0^{-1/2},\ |\dot{F}_0|\lesssim|\dot{p}_1|r^{-1/2}+|\dot{b}_0|r^{-1/2}+Kr^{-1/2}r_0^{-1/2}+|\dot{A}_1|.\] 
    Similarly to the derivation of \eqref{integrated eq for A0}, we integrate the equation \eqref{dr dot A0} for $\dot{A}_0$:
    \[\dot{A}_0(r)=\frac{\dot{a}_0d_0\sqrt{r_0}}{D_0\sqrt{r}}+\frac{1}{-D_0\sqrt{r}}\int_{r_0}^r\sqrt{r'}\bigg(\dot{F}_0+F_0\frac{\dot{D}_0}{-D_0}+\frac{\Lambda A_0 \dot{B}_0 r'}{2}+\frac{\Lambda A_0 B_0 \dot{D}_0 r'}{-2D_0}\bigg)dr'.\]
    Using the above pointwise estimates and the bootstrap assumption \eqref{boostrap assumption dot A0 A1}, we get:
    \begin{equation}\label{pointwise bound for dot A0}
        |\dot{A}_0|\lesssim r_0^{5/2}r^{-5/2}|\dot{a}_0|+r^{-3/2}\big(|\dot{b}_0|+|\dot{d}_0|+|\dot{p}_1|\big)+r^{-3/2}r_0^{-1/2}K.
    \end{equation}
    
    Next, we integrate equations \eqref{dr dot B1} and \eqref{dr dot D1}, using \eqref{pointwise bound for dot A0}, the previous pointwise estimates, and the bootstrap assumption \eqref{boostrap assumption dot A0 A1}:
    \[|\dot{B}_1|\lesssim|\dot{a}_0|+|\dot{b}_0|+|\dot{b}_1|+|\dot{d}_0|+|\dot{p}_1|+Kr_0^{-3/2}.\]
    \[|\dot{D}_1|\lesssim|\dot{d}_1|+r^2\big(|\dot{a}_0|+|\dot{b}_0|+|\dot{d}_0|+|\dot{b}_1|+|\dot{p}_1|+Kr_0^{-3/2}\big).\]
    Similarly to the derivation of \eqref{integrated eq for A0}, we can integrate equation \eqref{dr dot A1} for $\dot{A}_1$. Using the previously established bounds and the bootstrap assumption \eqref{boostrap assumption dot A0 A1}, one eventually obtains the following pointwise estimate for $\dot{A}_1$:
    \begin{align*}
        |\dot{A}_1|\lesssim & r_0^{5/2}r^{-5/2}\big(|\dot{a}_1|+|\dot{a}_0|\log r\big)+\big\|\dot{F}_1\big\|_{TY}r^{-5/2}\log r\\ &+r^{-3/2}\sum_{i=0}^1\big(|\dot{a}_i|+|\dot{b}_i|+|\dot{d}_i|+|\dot{p}_{i+1}|\big)+Kr^{-3/2}r_0^{-1/2}.
    \end{align*}
    We combine the pointwise bounds for $\dot{A}_0$ and $\dot{A}_1$ to bound the LHS of \eqref{boostrap assumption dot A0 A1}:
    \[\int_{r_0}^R\dot{A}_0^2(r)+r\dot{A}_1^2(r)dr\lesssim\big(|\dot{a}_1|^2+|\dot{a}_0|^2\big)r_0^3+\sum_{i=0}^1\big(|\dot{b}_i|^2+|\dot{d}_i|^2+|\dot{p}_{i+1}|^2\big)r_0^{-1}+\big\|\dot{F}_1\big\|_{TY}^2r_0^{-2}+K^2r_0^{-2}.\]
    
    By taking the constant $K$ to be:
    \[K^2=\big(|\dot{a}_1|^2+|\dot{a}_0|^2\big)r_0^5+\sum_{i=0}^1\big(|\dot{b}_i|^2+|\dot{d}_i|^2+|\dot{p}_{i+1}|^2\big)+\big\|\dot{F}_1\big\|_{TY}^2,\]
    we improve the bootstrap assumption \eqref{boostrap assumption dot A0 A1}, provided that $R_0$ is sufficiently large in terms of $\Lambda$ and $M$. We obtain that the ODE system \eqref{dr dot A1}-\eqref{dot data} has global solutions. Moreover, the proof implies that $\dot{A}_0,\dot{A}_1\in Z_1$, $\dot{B}_0,\dot{B}_1\in Z_2,$ $\dot{D}_0,\dot{D}_1\in Z_3,$ and $\dot{P}_1,\dot{P}_2\in Z_4.$ Thus, the linear map $\dot{\mathcal{S}}_{(w,F_1)}:\mathbb{R}^8\times TY\rightarrow Z$ is well-defined and bounded, since we have:
    \[\big\|\dot{\mathcal{S}}_{(w,F_1)}(\dot{w},\dot{F}_1)\big\|_Z\lesssim\big(|\dot{a}_1|+|\dot{a}_0|\big)r_0^2+\sum_{i=0}^1\big(|\dot{b}_i|+|\dot{d}_i|+|\dot{p}_{i+1}|\big)+\big\|\dot{F}_1\big\|_{TY}.\]

    Using similar arguments, it is straightforward to prove that the map $\dot{\mathcal{S}}_{(w,F_1)}$ is continuous in $(w,F_1)$ and $\dot{\mathcal{S}}_{(w,F_1)}=d\mathcal{S}_{(w,F_1)}$. We conclude that the map $\mathcal{S}:W\times Y\rightarrow Z$ is $C^1.$
    
    \end{proof}

In order to find $\underline{F_1}$ that satisfies the vanishing conditions in \eqref{vanishing conditions F1 bar}, we define the map: 
\[\mathcal{C}:W\times\mathbb{R}\times Y\rightarrow\mathbb{R}^3,\] 
\[\mathcal{C}(w,p_0,F_1)=\bigg(\int_{r_0}^{\infty}F_1dr+p_2,\ \int_{r_0}^{\infty}\frac{B_0A_1}{\sqrt{r}}dr+p_1,\ \int_{r_0}^{\infty}\frac{A_0}{\sqrt{r}}dr+p_0\bigg),\]
where $(A_0,A_1,B_0,B_1,D_0,D_1,P_1,P_2)=\mathcal{S}(w,F_1)$. 

Using \Cref{S is well defined lemma} and \Cref{dot S is well defined lemma}, we have that $\mathcal{C}$ is $C^1$ and:
\[d\mathcal{C}_{(w,p_0,F_1)}(0,0,\dot{F}_1)=\bigg(\int_{r_0}^{\infty}\dot{F}_1dr,\ \int_{r_0}^{\infty}\frac{B_0\dot{A}_1+\dot{B}_0A_1}{\sqrt{r}}dr,\ \int_{r_0}^{\infty}\frac{\dot{A}_0}{\sqrt{r}}dr\bigg),\]
where $(\dot{A}_0,\dot{A}_1,\dot{B}_0,\dot{B}_1,\dot{D}_0,\dot{D}_1,\dot{P}_0,\dot{P}_1)=d\mathcal{S}_{(w,F_1)}(0,\dot{F}_1).$

We establish the essential non-degeneracy condition required for applying the Implicit Function Theorem in \Cref{IFT corollary}:
\begin{lemma}\label{dC is surjective lemma}
    The map $d\mathcal{C}_{(w,p_0,F_1)}(0,0,\cdot):TY\rightarrow\mathbb{R}^3$ is surjective.
\end{lemma}
\begin{proof}
    We make the notation convention for the purpose of this proof that we write $x\lesssim y$ for any quantities $x,y>0$ if there is a constant $C(\Lambda,M,r_0)>0$ such that $x\leq Cy.$
    
    We assume the contrary, so there exists $(\eta_1,\eta_2,\eta_3)\in\mathbb{R}^3\backslash\{(0,0)\}$ such that for all $\dot{F}_1\in TY$:
    \begin{equation}\label{linear dep eq}
        \eta_1\int_{r_0}^{\infty}\dot{F}_1dr+\eta_2\int_{r_0}^{\infty}\frac{B_0\dot{A}_1+\dot{B}_0A_1}{\sqrt{r}}dr+\eta_3\int_{r_0}^{\infty}\frac{\dot{A}_0}{\sqrt{r}}dr=0.
    \end{equation}

    For any $r_1>r_0$ sufficiently large, we consider $\dot{F}_1$ to be a smooth non-negative bump function supported in $[r_1,r_1+1]$, such that:
    \begin{equation}\label{conditions on F1 dot}
        \int_{r_0}^{\infty}\dot{F}_1(r)dr=1,\ \frac{\sqrt{r_1}}{2}<\int_{r_0}^{\infty}\sqrt{r}\dot{F}_1(r)dr<2\sqrt{r_1},\ \frac{r_1}{2}<\|\dot{F}_1\|_{TY}<2r_1.
    \end{equation}
    As a result, we have that $\dot{A}_i=\dot{B}_i=\dot{D}_i=\dot{P}_{i+1}=0$ for $i=1,2$ on $[r_0,r_1].$ 

    We first apply \Cref{dot S is well defined lemma}, and notice that $r_1/2<K<2r_1$ because of \eqref{conditions on F1 dot}. Thus, the pointwise bounds on $\dot{A}_0$ and $\dot{A}_1$ proved in \Cref{dot S is well defined lemma} imply that $\dot{A}_0,\dot{A}_1\lesssim r_1r^{-3/2}.$ Using these bounds in equations \eqref{dot F0}-\eqref{dr dot P1} and \eqref{dr dot B1}-\eqref{dr dot P2}, we obtain the following preliminary bounds:
    \begin{equation}\label{preliminary dot bounds 1}
        \dot{A}_0,\dot{A}_1\lesssim r_1r^{-3/2},\ \dot{B}_0,\dot{B}_1\lesssim r_1^{-1},\ \dot{D}_0,\dot{D}_1\lesssim r_1^{-1}r^2,\ \dot{P}_1,\dot{P}_2\lesssim1,\ \dot{F}_0\lesssim r^{-1/2}.
    \end{equation}
    Our strategy is to use the preliminary bounds in equations \eqref{dr dot A0}, \eqref{dr dot A1} to improve the bounds for $\dot{A}_0$ and $\dot{A}_1$, then use equations \eqref{dot F0}-\eqref{dr dot P1} and \eqref{dr dot B1}-\eqref{dr dot P2} to improve the bounds for the remaining quantities. We do not write the estimates in detail here, but we notice that integrating \eqref{dr dot A0} and \eqref{dr dot A1} implies:
    \[\big|\dot{A}_0\big|(r)\lesssim r^{-5/2}\int_{r_1}^r\sqrt{\Tilde{r}}\big|\dot{F_0}\big|+\big|\dot{B_0}\big|+\Tilde{r}^{-2}\big|\dot{D_0}\big|d\Tilde{r},\]
    \[\big|\dot{A}_1\big|(r)\lesssim r^{-5/2}\int_{r_1}^r\sum_{i=0}^1\bigg[\sqrt{\Tilde{r}}\big|\dot{F_i}\big|+\big|\dot{B_i}\big|+\Tilde{r}^{-2}\big|\dot{D_i}\big|\bigg]+\Tilde{r}^{-5/2}\big|\dot{P_2}\big|+\Tilde{r}^{3/2}\big|\dot{A_0}\big|+\Tilde{r}^{-3/2}\big|\dot{D_0} F_1\big|d\Tilde{r}.\]
    After the first iteration, we obtain the improved preliminary bounds:
    \begin{equation}\label{preliminary dot bounds 2}
        \dot{A}_0,\dot{A}_1\lesssim r^{-3/2},\ \dot{B}_0,\dot{B}_1\lesssim r_1^{-2},\ \dot{D}_0,\dot{D}_1\lesssim r_1^{-2}r^2,\ \dot{P}_1\lesssim r_1^{-1},\ \dot{P}_2\lesssim 1,\ \dot{F}_0\lesssim r_1^{-1}r^{-1/2}.
    \end{equation}
    Repeating our argument by starting with \eqref{preliminary dot bounds 2}, we further improve to: 
    \[\dot{A}_0\lesssim r_1^{-1}r^{-3/2},\ \dot{A}_1\lesssim r_1^{-1}r^{-3/2}+\sqrt{r_1}r^{-5/2},\ \dot{B}_0\lesssim r_1^{-3},\ \dot{B}_1\lesssim r_1^{-5/2},\] \[\dot{D}_0\lesssim r_1^{-3}r^2,\ \dot{D}_1\lesssim r_1^{-5/2}r^2,\ \dot{P}_1\lesssim r_1^{-3/2},\ \dot{P}_2\lesssim1,\ \dot{F}_0\lesssim r_1^{-3/2}r^{-1/2}.\]
    Finally, we iterate our argument one last time to obtain the sharp pointwise bounds:
    \begin{align}
        \dot{A}_0\lesssim r_1^{-3/2}r^{-3/2},\ \dot{A}_1\lesssim r_1^{-3/2}r^{-3/2}+\sqrt{r_1}r^{-5/2},\ \dot{B}_0\lesssim r_1^{-7/2},\ \dot{B}_1\lesssim r_1^{-5/2},\label{sharp dot bounds 1}\\
        \dot{D}_0\lesssim r_1^{-7/2}r^2,\ \dot{D}_1\lesssim r_1^{-5/2}r^2,\ \dot{P}_1\lesssim r_1^{-3/2},\ \dot{P}_2\lesssim1,\ \dot{F}_0\lesssim r_1^{-3/2}r^{-1/2}.\label{sharp dot bounds 2}
    \end{align}

    We use \eqref{sharp dot bounds 1}-\eqref{sharp dot bounds 2} to bound the terms in \eqref{linear dep eq} as follows:
    \begin{equation}\label{bound for each term in lin dep eq}
        \int_{r_0}^{\infty}\frac{\big|B_0\dot{A}_1\big|}{\sqrt{r}}dr\lesssim r_1^{-3/2},\ \int_{r_0}^{\infty}\frac{\big|\dot{B}_0A_1\big|}{\sqrt{r}}dr\lesssim r_1^{-7/2},\ \int_{r_0}^{\infty}\frac{\big|\dot{A}_0\big|}{\sqrt{r}}dr\lesssim r_1^{-5/2}.
    \end{equation}
    We notice that equations~\eqref{linear dep eq}, \eqref{conditions on F1 dot}, and \eqref{bound for each term in lin dep eq} imply that $|\eta_1|\lesssim|\eta_2|r_1^{-3/2}+|\eta_3|r_1^{-5/2}$. Since $r_1$ can be arbitrarily large, we obtain that $\eta_1=0.$ 

    Integrating equation \eqref{dr dot A1}, we have that:
    %\[\dot{A}_1(r)=\frac{1}{D_0\sqrt{r}}\int_{r_1}^{\infty}\sqrt{r'}\dot{F}_1(r')\exp\bigg(-\int_{r'}^rA_0^2dr''\bigg)dr'+O\big(r_1^{-3/2}r^{-3/2}\big).\]
    \[\dot{A}_1(r)=\frac{1}{-D_0\sqrt{r}}\int_{r_1}^{r}\sqrt{\Tilde{r}}\dot{F}_1(\Tilde{r})d\Tilde{r}+O\big(r_1^{-3/2}r^{-3/2}\big).\]
    We note that the first term on the RHS is non-negative. Using \eqref{bounds for B0 D0 P1} and \eqref{conditions on F1 dot}, we obtain that for $r>r_1+1$:
    \begin{equation}\label{exact rate for dot A1}
        \dot{A}_1(r)\sim\sqrt{r_1}r^{-5/2}+O\big(r_1^{-3/2}r^{-3/2}\big).
    \end{equation}
    Therefore, equations \eqref{bounds for B0 D0 P1}, \eqref{bound for each term in lin dep eq}, and \eqref{exact rate for dot A1} imply that:
    \[\int_{r_0}^{\infty}\frac{B_0\dot{A}_1+\dot{B}_0A_1}{\sqrt{r}}dr>\int_{r_1+1}^{\infty}\frac{B_0\dot{A}_1}{\sqrt{r}}dr+O\big(r_1^{-5/2}\big)\gtrsim r_1^{-3/2}+O\big(r_1^{-5/2}\big).\]
    We use this estimate in equation \eqref{linear dep eq}, in order to get that $|\eta_2|r_1^{-3/2}\lesssim|\eta_3|r_1^{-5/2}$. Since $r_1$ can be arbitrarily large, we obtain that $\eta_2=0.$ 

    Integrating equation \eqref{dr dot P1}, we have that:
    \[\dot{P}_1(r)=\int_{r_1}^{r}\frac{B_0}{\sqrt{r'}}\bigg(\frac{1}{-D_0\sqrt{r'}}\int_{r_1}^{r'}\sqrt{\Tilde{r}}\dot{F}_1(\Tilde{r})d\Tilde{r}\bigg) dr'+O\big(r_1^{-5/2}\big).\] 
    The first term on the RHS is non-negative, and we get that for all $r>2r_1:$
    \[\dot{P}_1(r)\sim r_1^{-3/2}+O\big(r_1^{-5/2}\big).\] 
    According to \eqref{dot F0}, we also get the equation:
    \[\dot{F}_0(r)=\frac{B_0}{-D_0\sqrt{r}}\int_{r_1}^{r}\sqrt{\Tilde{r}}\dot{F}_1(\Tilde{r})d\Tilde{r}+\frac{1}{2\sqrt{r}}\int_{r_1}^{r}\frac{B_0}{\sqrt{r'}}\bigg(\frac{1}{-D_0\sqrt{r'}}\int_{r_1}^{r'}\sqrt{\Tilde{r}}\dot{F}_1(\Tilde{r})d\Tilde{r}\bigg) dr'+O\big(r_1^{-5/2}r^{-1/2}\big),\]
    where the first two terms are non-negative. Thus, for all $r>2r_1:$
    \[\dot{F}_0(r)\gtrsim r_1^{-3/2}r^{-1/2}+O\big(r_1^{-5/2}r^{-1/2}\big).\]
    Next, we integrate equation \eqref{dr dot A0}: 
    %\[\dot{A}_0(r)=\frac{1}{D_0\sqrt{r}}\int_{r_0}^r\sqrt{r'}\bigg(\dot{F}_0+F_0\frac{\dot{D}_0}{-D_0}+\frac{\Lambda A_0 \dot{B}_0 r'}{2}+\frac{\Lambda A_0 B_0 \dot{D}_0 r'}{-2D_0}\bigg)dr'\]
    \begin{align*}
        \dot{A}_0(r)&=\frac{1}{-D_0\sqrt{r}}\int_{r_1}^r\sqrt{r'}\dot{F}_0(r')dr'+O\big(r_1^{-5/2}r^{-3/2}\big)\\
        &=\frac{1}{-2D_0\sqrt{r}}\int_{r_1}^r\bigg[\int_{r_1}^{r'}\frac{B_0}{\sqrt{r''}}\bigg(\frac{1}{-D_0\sqrt{r''}}\int_{r_1}^{r''}\sqrt{\Tilde{r}}\dot{F}_1(\Tilde{r})d\Tilde{r}\bigg) dr''\bigg]dr'\\
        &\ +\frac{1}{-D_0\sqrt{r}}\int_{r_1}^r\bigg(\frac{B_0}{-D_0}\int_{r_1}^{r'}\sqrt{\Tilde{r}}\dot{F}_1(\Tilde{r})d\Tilde{r}\bigg)dr'+O\big(r_1^{-5/2}r^{-3/2}\big).
    \end{align*}
    We notice that the first two terms are non-negative and for any $r>3r_1$, we have:
    \[\dot{A}_0(r)\gtrsim r_1^{-3/2}r^{-3/2}+O\big(r_1^{-5/2}r^{-3/2}\big).\]
    As a result, we also have:
    \[\int_{r_0}^{\infty}\frac{\dot{A}_0}{\sqrt{r}}dr>\int_{3r_1}^{\infty}\frac{\dot{A}_0}{\sqrt{r}}dr+O\big(r_1^{-7/2}\big)\gtrsim r_1^{-5/2}.\]
    Thus, equation \eqref{linear dep eq} implies that $\eta_3=0$ as well, which contradicts $(\eta_1,\eta_2,\eta_3)\neq(0,0,0).$
\end{proof}

According to \Cref{S is well defined lemma}, \Cref{dot S is well defined lemma}, and \Cref{dC is surjective lemma}, the map $\mathcal{C}$ satisfies the required conditions to apply the Implicit Function Theorem. We obtain the following result:
    \begin{corollary}\label{IFT corollary}
        Let $(w,p_0,F_1)\in W\times\mathbb{R}\times Y$ such that $\mathcal{C}(w,p_0,F_1)=0$ and $r_0\geq R_0$ is large enough, as required in Lemmas~\ref{S is well defined lemma}, \ref{dot S is well defined lemma}, and \ref{dC is surjective lemma}. There exist neighborhoods $(w,p_0)\in U\subset W\times\mathbb{R}$, $F_1\in V\subset Y,$ and a $C^1$ function $\mathcal{F}:U\rightarrow V$ such that $F_1=\mathcal{F}(w,p_0)$ and  $\mathcal{C}\big(\underline{w},\underline{p_0},\mathcal{F}(\underline{w},\underline{p_0})\big)=0$ for any $(\underline{w},\underline{p_0})\in U.$
    \end{corollary}

    We conclude this section by proving that in the set-up of \Cref{local perturbation extension proposition}, we can extend $\partial_r\underline{\mathcal{T}}\big(\sqrt{r}\underline{\gamma\mathcal{T}\phi}\big)$ to $[r_0,\infty),$ while satisfying the vanishing conditions \eqref{vanishing conditions F1 bar L2 conditions}-\eqref{vanishing conditions F1 bar}:
\begin{proposition}\label{existence of F1 proposition}
	Let $\mathfrak{D}$ be a $C^2$ asymptotically AdS characteristic data set. We fix $R_0>0$ large enough depending on $\mathfrak{D}$ and $\Lambda$, so that Lemmas~\ref{S is well defined lemma}, \ref{dot S is well defined lemma}, \ref{dC is surjective lemma}, and \Cref{IFT corollary} apply for any $r_0\geq R_0$ with:
    \begin{equation}\label{w p0 F1}
        w=\big(\theta(r_0),\mathcal{T}\theta(r_0),\ldots,\mathcal{T}\big(\sqrt{r}\gamma\mathcal{T}\phi\big)(r_0)\big),\ p_0=\phi(r_0), \  F_1=\partial_r\mathcal{T}\big(\sqrt{r}\gamma\mathcal{T\phi}\big).
    \end{equation}

    There exists a constant $C_1>0,$ such that the following holds: for all $\varepsilon>0$ small enough in terms of $\mathfrak{D},\Lambda,$ and $r_0$, and any partial data set $\underline{\mathfrak{D}}$ on $[0,r_0]$ which satisfies \eqref{smallness in a finite region} and $\underline{\partial_U^ir}(0)=\partial_U^ir(0)$ for $i\in\{1,2\}$, we denote:
    \[\underline{w}=\big(\underline{\theta}(r_0),\underline{\mathcal{T}\theta}(r_0),\ldots,\underline{\mathcal{T}}\big(\sqrt{r}\underline{\gamma}\underline{\mathcal{T}\phi}\big)(r_0)\big),\ \underline{p_0}=\underline{\phi}(r_0).\]
    Then, there exists $\underline{F_1}:[r_0,\infty)\rightarrow\mathbb{R}$ such that $\underline{F_1}\in Y,\ \|\underline{F_1}-F_1\|_{Y}<C_1\varepsilon,$ and: 
    \begin{equation}\label{vanishing conditions eq}
    \int_{r_0}^{\infty}\partial_r\underline{\mathcal{T}}\big(\sqrt{r}\underline{\gamma\mathcal{T}\phi}\big)dr=-\underline{\mathcal{T}}\big(\sqrt{r}\underline{\gamma\mathcal{T}\phi}\big)(r_0),\ \int_{r_0}^{\infty}\frac{\underline{\gamma\mathcal{T}\theta}}{\sqrt{r}}dr=-\underline{\gamma\mathcal{T}\phi}(r_0),\ \int_{r_0}^{\infty}\frac{\underline{\theta}}{\sqrt{r}}dr=-\underline{\phi}(r_0),
    \end{equation}
where $\big(\underline{\theta},\underline{\mathcal{T}\theta},\ldots,\underline{\mathcal{T}}\big(\sqrt{r}\underline{\gamma}\underline{\mathcal{T}\phi}\big)\big)=\mathcal{S}\big(\underline{w},\underline{F_1}\big).$
\end{proposition}
\begin{proof}
    We note that \Cref{asympt ads data def}, the null constraint equations \eqref{dr theta}-\eqref{dr T phi}, and \eqref{L2 condition for F1}-\eqref{vanishing conditions for D} imply that there exists $M>0$ depending on $\mathfrak{D}$ such that $w,p_0,$ and $F_1$ defined in \eqref{w p0 F1} satisfy  $w\in W,\ F_1\in Y,$ and $\mathcal{C}(w,p_0,F_1)=0,$ for all $r_0$ sufficiently large in terms of $\mathfrak{D}$ and $\Lambda$. Given this value of $M,$ we fix $R_0$ such that Lemmas~\ref{S is well defined lemma}, \ref{dot S is well defined lemma}, \ref{dC is surjective lemma}, and \Cref{IFT corollary} apply for any $r_0\geq R_0$.
    
    According to \Cref{IFT corollary}, we get neighborhoods $(w,p_0)\in U\subset W\times\mathbb{R}$, $F_1\in V\subset Y,$ and a $C^1$ function $\mathcal{F}:U\rightarrow V$ such that $F_1=\mathcal{F}(w,p_0)$ and  $\mathcal{C}\big(\underline{w'},\underline{p_0'},\mathcal{F}(\underline{w'},\underline{p_0'})\big)=0$ for any $(\underline{w'},\underline{p_0'})\in U.$

    The smallness condition \eqref{smallness in a finite region} and the boundary conditions at $\Gamma$ imply that there exists a constant $C_1'>0$ independent of $\varepsilon$ and $r_0,$ such that the following estimates hold at $r_0:$
    \begin{align}
        \big|\underline{\gamma}-\gamma\big|(r_0)<C_1'\varepsilon,\ &\big|\underline{\partial_Ur}-\partial_Ur\big|(r_0)<C_1'\varepsilon r_0^2,\ \big|\underline{\phi}-\phi\big|(r_0)<C_1'\varepsilon r_0^{-1},\\ \big|\underline{\theta}-\theta\big|(r_0)<C_1'\varepsilon r_0^{-3/2},\
        &\big|\underline{\mathcal{T}\gamma}-\mathcal{T}\gamma\big|(r_0)<C_1'\varepsilon,\ \big|\underline{\mathcal{T}\partial_Ur}-\mathcal{T}\partial_Ur\big|(r_0)<C_1'\varepsilon r_0^2,\\ \big|\underline{\mathcal{T}\phi}-\mathcal{T}\phi\big|(r_0)<C_1'\varepsilon r_0^{-1},\ &\big|\underline{\mathcal{T}\theta}-\mathcal{T}\theta\big|(r_0)<C_1'\varepsilon r_0^{-3/2},\ \big|\underline{\mathcal{T}^2\phi}-\mathcal{T}^2\phi\big|(r_0)<C_1'\varepsilon r_0^{-1}.
    \end{align}
    % scalar field bounds follow from Sobolev embedding (first show r\phi=0 at the center), the rest follows by integrating transport equations using boundary conditions at the center
    
    Therefore, for all $\varepsilon>0$ sufficiently small in terms of $\mathfrak{D},\Lambda,$ and $r_0,$ we have that $(\underline{w},\underline{p_0})\in U$. We define $\underline{F_1}=\mathcal{F}(\underline{w},\underline{p_0}),$ which satisfies $\underline{F_1}\in V$ and $\mathcal{C}(\underline{w},\underline{p_0},\underline{F_1})=0,$ so \eqref{vanishing conditions eq} holds. Since $\mathcal{F}\in C^1$, there exists a constant $C_1>0$ depending on $\mathfrak{D},\Lambda,$ and $r_0,$ such that $\|\underline{F_1}-F_1\|_{Y}\leq C_1\varepsilon.$ 
\end{proof}

\subsection{Proof of \Cref{local perturbation extension proposition}}\label{extending local perturbations section}

\begin{proof}[Proof of \Cref{local perturbation extension proposition}] We prove the result for $k=2,$ and we refer the reader to \Cref{remark on general k} for the general case. Let $\mathfrak{D}$ be a $C^2$ asymptotically AdS characteristic data set. We fix $r_0>0$ large enough in terms of $\mathfrak{D}$ and $\Lambda$, as in the proof of \Cref{existence of F1 proposition}. We make the notation convention for the purpose of this proof that we write $x\lesssim y$ for any quantities $x,y>0$ if there is a constant $C(\Lambda,\mathfrak{D},r_0)>0$ such that $x\leq Cy.$

For any $\varepsilon>0$ sufficiently small in terms of $\mathfrak{D},\Lambda,$ and $r_0$, let $\underline{\mathfrak{D}}$ be a $C^2$ characteristic partial data set satisfying \eqref{smallness in a finite region}. At first, we use double null coordinates adapted to $\mathfrak{D}.$ We choose an ingoing coordinate for $\underline{\mathfrak{D}}$ such that $\underline{\partial_U^ir}(0)=\partial_U^ir(0),\ i\in\{1,2\}$, and we use the asymptotically Eddington-Finkelstein coordinate $V$ of  $\mathfrak{D}$ for $\underline{\mathfrak{D}}$ as well.

According to \Cref{existence of F1 proposition}, there exists $\underline{F_1}\in Y$ such that $\|\underline{F_1}-F_1\|_{Y}<C_1\varepsilon$ and the vanishing conditions \eqref{vanishing conditions eq} hold, where $\big(\underline{\theta},\underline{\mathcal{T}\theta},\ldots,\underline{\mathcal{T}}\big(\sqrt{r}\underline{\gamma}\underline{\mathcal{T}\phi}\big)\big)=\mathcal{S}\big(\underline{w},\underline{F_1}\big).$ As a consequence, we get that:
\begin{equation}\label{Tphi using vanishing}
    \underline{\gamma\mathcal{T}\phi}(r)=-\int_{r}^{\infty}\frac{\underline{\gamma\mathcal{T}\theta}}{\sqrt{r'}}dr'.
\end{equation}
In particular, we improve the pointwise bounds in \Cref{S is well defined lemma} as follows:
\[\big|\underline{\gamma\mathcal{T}\phi}\big|\lesssim r^{-1},\ \big|\partial_r(\sqrt{r}\underline{\gamma\mathcal{T}\phi})\big|\lesssim r^{-3/2}.\]
Integrating equations \eqref{dr A_0} and \eqref{dr A_1} as in the proof of \Cref{S is well defined lemma}, we improve the pointwise bounds \eqref{pointwise bound for A0}, \eqref{pointwise bound for A1}:
\begin{equation}\label{imrpoved pointwise bounds for theta Ttheta}
    |\underline{\theta}|\lesssim r^{-5/2}\log r,\ |\partial_r\underline{\theta}|\lesssim r^{-7/2}\log r,\ |\underline{\mathcal{T}\theta}|\lesssim r^{-5/2}(\log r)^2.
\end{equation}

Using \eqref{vanishing conditions eq}, we can extend $\underline{\phi}$ as a $C^2$ function to $[r_0,\infty)$ by:
    \[\underline{\phi}(r)=-\int_r^{\infty}\frac{\underline{\theta}}{\sqrt{r'}}dr',\]
which satisfies $\underline{\phi}\in L^2([r_0,\infty)).$ The pointwise bounds in \eqref{imrpoved pointwise bounds for theta Ttheta} further imply that $\underline{\phi}\in H^2([r_0,\infty)).$ Differentiating equation \eqref{dr theta}, we also obtain that:
\[|\partial_r^3\underline{\phi}|\lesssim r^{-5}(\log r)^2+r^{-9/2}|\underline{F_1}|,\]
%actually log instead of log^2
which implies that $\underline{\phi}\in H^3([r_0,\infty)).$ 

We use the improved pointwise bounds \eqref{imrpoved pointwise bounds for theta Ttheta} in \eqref{Tphi using vanishing} to get:
\[|\underline{\mathcal{T}\phi}|\lesssim r^{-2}(\log r)^3,\ |\partial_r\underline{\mathcal{T}\phi}|\lesssim r^{-3}(\log r)^2,\]
so we get that $\underline{\mathcal{T}\phi}\in H^1([r_0,\infty)).$ Equation \eqref{dr T theta} also implies that:
\[|\partial_r^2\underline{\mathcal{T}\phi}|\lesssim r^{-4}(\log r)^2+r^{-5/2}|\underline{F_1}|,\]
so we further get that $\underline{\mathcal{T}\phi}\in H^2([r_0,\infty)).$ 

Additionally, using \eqref{vanishing conditions eq} we can extend $\underline{\mathcal{T}^2\phi}$ as a $C^0$ function to $[r_0,\infty)$:
\[\underline{\mathcal{T}^2\phi}(r)=-\frac{1}{\sqrt{r}\underline{\gamma}(r)}\int_r^{\infty}\underline{F_1}(r')dr'-\frac{\underline{\mathcal{T}\gamma}}{\underline{\gamma}}\underline{\mathcal{T}\phi},\]
which also implies that $\underline{F_1}=\partial_r\underline{\mathcal{T}}\big(\sqrt{r}\underline{\gamma\mathcal{T}\phi}\big).$ We also notice that:
\begin{equation}\label{bound for first term in Hardy}
    \lim_{r\rightarrow\infty}r^2\big(\underline{\mathcal{T}^2\phi}\big)^2\leq\lim_{r\rightarrow\infty}r\bigg(\int_r^{\infty}\underline{F_1}(r')dr'\bigg)^2\leq\lim_{r\rightarrow\infty}\int_r^{\infty}\big(r'\underline{F_1}\big)^2dr'=0,
\end{equation}
\begin{equation}\label{bound for second term in Hardy}
    \int_r^{\infty}\Tilde{r}^2\Big[\partial_r\big(\sqrt{r}\underline{\gamma}\underline{\mathcal{T}^2\phi}\big)\Big]^2d\Tilde{r}\lesssim\big\|\underline{F_1}\big\|_{Y}^2+\int_r^{\infty}\Tilde{r}^2\Big[\partial_r\big(\sqrt{r}\underline{\mathcal{T}\gamma}\underline{\mathcal{T}\phi}\big)\Big]^2d\Tilde{r}\lesssim\big\|\underline{F_1}\big\|_{Y}^2+\big\|\underline{\mathcal{T}\phi}\big\|_{H^1}^2\lesssim1.
\end{equation}
Similarly to \eqref{Hardy ineq near I}, we have the following Hardy inequality for any function $f$ such that the RHS of \eqref{Hardy ineq near I p=0} is finite:
\begin{equation}\label{Hardy ineq near I p=0}
    rf^2+\int_r^{\infty}f^2d\Tilde{r}\leq2\lim_{\Tilde{r}\rightarrow\infty}\Tilde{r}f^2+ 8\int_r^{\infty}\Tilde{r}^2\big(\partial_rf\big)^2d\Tilde{r}.
\end{equation}
Using \eqref{Hardy ineq near I p=0} for $\sqrt{r}\underline{\gamma}\cdot\underline{\mathcal{T}^2\phi},$ together with \eqref{bound for first term in Hardy}-\eqref{bound for second term in Hardy}, we obtain that $\underline{\mathcal{T}^2\phi}\in H^1([r_0,\infty)).$ 

So far, we proved that $\underline{\phi}\in H^3([r_0,\infty)),\ \underline{\mathcal{T}\phi}\in H^2([r_0,\infty)),\ \underline{\mathcal{T}^2\phi}\in H^1([r_0,\infty)).$ Using a similar argument for differences of solutions, it is straightforward to obtain that:
\begin{equation}\label{main difference estimate for extension}
    \big\|\underline{\phi}-\phi\big\|_{H^3[r_0,\infty)}+\big\|\underline{\mathcal{T}\phi}-\mathcal{T}\phi\big\|_{H^2[r_0,\infty)}+\big\|\underline{\mathcal{T}^2\phi}-\mathcal{T}^2\phi\big\|_{H^1[r_0,\infty)}\lesssim\varepsilon.
\end{equation}

According to \Cref{solving compatibility conditions remark}, the tuple $\underline{\mathfrak{D}}$ is determined uniquely on $[r_0,\infty)$ by the solution to the null constraint equations \eqref{dr theta}-\eqref{dr T phi} that we constructed, and it satisfies the compatibility conditions in \Cref{asympt ads data def}. In view of \eqref{main difference estimate for extension}, we have that the extended $\underline{\mathfrak{D}}$ satisfies \eqref{smallness of extension} as desired. Next, we choose coordinates $(\underline{U},\underline{V})$ that are asymptotically Eddington-Finkelstein with respect to $\underline{\mathfrak{D}}$. We define $\underline{M}=\lim_{r\rightarrow\infty}\underline{m}(r)$, then choose an asymptotically Eddington-Finkelstein coordinate $\underline{V}$ satisfying \eqref{asympt EF coordinate} and \eqref{asympt behavior dv r}. We rescale the ingoing coordinate $\underline{U}=U\cdot(1+O(\varepsilon))$ such that $\lim_{r\rightarrow\infty}\underline{\gamma}(r)=2,$ which implies that \eqref{asympt behavior du r}-\eqref{asympt behavior Lapse} hold at zero order. Furthermore, using the already established pointwise estimates, it is straightforward to prove the asymptotic rates \eqref{asympt behavior du r}-\eqref{asympt behavior phi}, concluding that $\underline{\mathfrak{D}}$ is a $C^2$ asymptotically AdS characteristic data set.
\end{proof}

%\textbf{Remark.} I don't ask that $\underline{\mathcal{T}\gamma}$ and $\underline{\mathcal{T}\partial_Ur}$ vanish at infinity. The way the proof of \cite{localwellposed} works is that they find a coordinate system in which this holds afterwards.

\begin{remark}\label{remark on general k}
    In the general case $k\geq2$, we have the commuted system for all $0\leq n\leq k-1:$
\begin{align}
    \partial_r\mathcal{T}^n\theta+\bigg(\frac{1}{2r}+\frac{\Lambda\gamma r}{2\partial_Ur}\bigg)\mathcal{T}^n\theta=&\frac{\partial_r\mathcal{T}^n\big(\sqrt{r}\gamma\mathcal{T}\phi\big)}{-\partial_Ur}+Err_1^n, \\
    \partial_r\mathcal{T}^n\gamma=& Err_2^n,\\
    \partial_r\mathcal{T}^n\partial_Ur=& \frac{\Lambda}{2}r\mathcal{T}^n\gamma+Err_3^n,
\end{align}
where we have the schematic error terms:
\begin{align*}
    Err_1^n=&\sum_{i+j+l=n,\ l\geq 1}\mathcal{T}^i\theta\mathcal{T}^j\theta\mathcal{T}^l\theta+\sum_{h+i+j+l=n,\ l\geq 1}\mathcal{T}^h\bigg(\frac{\sqrt{r}}{-\partial_Ur}\bigg)\mathcal{T}^i\theta\mathcal{T}^j\theta\mathcal{T}^l\big(\gamma\mathcal{T}\phi\big)+ \\
    &+\sum_{i=1}^n\mathcal{T}^i\bigg(\frac{1}{-\partial_Ur}\bigg)\partial_r\mathcal{T}^{n-i}\big(\sqrt{r}\gamma\mathcal{T}\phi\big)+\frac{\Lambda r}{2}\sum_{i=1}^n\mathcal{T}^i\bigg(\frac{\gamma}{\partial_Ur}\bigg)\mathcal{T}^{n-i}\theta,\\
    Err_2^n=& \sum_{i+j+l=n}\mathcal{T}^i\theta\mathcal{T}^j\theta\mathcal{T}^l\gamma,\\
    Err_3^n=&\sum_{i+j+l=n,\ l\geq 1}\mathcal{T}^i\theta\mathcal{T}^j\theta\mathcal{T}^l\partial_Ur.
\end{align*}
%[induction for $i>0$ gives $\dot{A}_{n-i}\sim r^{-3/2}r_1^{-1/2-i}+r^{-5/2}r_1^{1/2-i}$]

This system has the same structure as the null constraint equations \eqref{dr theta}-\eqref{dr T phi}, so it is clear that the argument in the proof of \Cref{local perturbation extension proposition} can be adapted to the general case $k\geq2$ using an induction argument.
\end{remark}

\bibliographystyle{amsalpha}
\bibliography{refs}

\end{document}